\newif\ifappendixshow  %
\appendixshowtrue

\newif\ifdraft     %
\draftfalse

\newif\ifexternal  %
\externaltrue

\newif\ifpaper     %
\paperfalse

\documentclass[acmsmall,nonacm]{acmart}
\let\terms\relax %

\newcommand{\frontpagedeclarations}{%
\title{Shapley Revisited: Tractable Responsibility Measures for Query Answers}

\newcommand{\afflabri}{%
  \institution{Univ. Bordeaux, CNRS,  Bordeaux INP, LaBRI, UMR 5800}
  \city{F-33400, Talence}
  \country{France}
}
\author{Meghyn Bienvenu}
\email{meghyn.bienvenu@cnrs.fr}
\orcid{0000-0001-6229-8103}
\affiliation{%
  \institution{Univ. Bordeaux, CNRS,  Bordeaux INP, LaBRI, UMR 5800}
  \city{F-33400, Talence}
  \country{France}
}

\author{Diego Figueira}
\email{diego.figueira@cnrs.fr}
\orcid{0000-0003-0114-2257}
\affiliation{%
  \institution{Univ. Bordeaux, CNRS,  Bordeaux INP, LaBRI, UMR 5800}
  \city{F-33400, Talence}
  \country{France}
}

\author{Pierre Lafourcade}
\email{pierre.lafourcade@u-bordeaux.fr}
\orcid{0009-0004-4810-1289}
\affiliation{%
  \institution{Univ. Bordeaux, CNRS,  Bordeaux INP, LaBRI, UMR 5800}
  \city{F-33400, Talence}
  \country{France}
}

\renewcommand{\shortauthors}{Meghyn Bienvenu, Diego Figueira, and Pierre Lafourcade}
}
\PassOptionsToPackage{fleqn}{amsmath}  %
  \usepackage{amsmath}

\ifappendixshow%
\usepackage[bibliography=common,appendix=append]{apxproof} %
\else
\usepackage[bibliography=common,appendix=strip]{apxproof}
\fi

\NewCommandCopy{\proofqedsymbol}{\qedsymbol}%
\newcommand{\remarkqedsymbol}{{$\triangle$}}%
\AtBeginEnvironment{proof}{\renewcommand{\qedsymbol}{\proofqedsymbol}}
\AtBeginEnvironment{nestedproof}{\renewcommand{\qedsymbol}{\nestedproofqedsymbol}}
\AtBeginEnvironment{example}{%
  \pushQED{\qed}\renewcommand{\qedsymbol}{\exampleqedsymbol}%
}
\AtEndEnvironment{example}{\popQED\endexample}
\AtBeginEnvironment{remark}{%
  \pushQED{\qed}\renewcommand{\qedsymbol}{\remarkqedsymbol}%
}
\AtEndEnvironment{remark}{\popQED\endexample}

\newif\ifappendix  %
\appendixfalse

\usepackage{algpseudocodex} %
\usepackage{scrhack} %
\usepackage{listings} %
\usepackage{tabularx} %
\usepackage{subfig}

\usepackage{graphicx} %

\usepackage{tikz}
\usepackage{tikzit} %
\usetikzlibrary{decorations.pathreplacing}
\usetikzlibrary{decorations.pathmorphing}
\usetikzlibrary{positioning}
\usetikzlibrary{patterns.meta}
\usetikzlibrary {arrows.meta}

\PassOptionsToPackage{british}{babel} %
    \usepackage{babel}
\usepackage{bbding} %
\usepackage{csquotes} %
\usepackage[inline,shortlabels]{enumitem} %
\usepackage{lineno} %
\usepackage{xspace} %

\usepackage{hyperref}
\usepackage{cleveref} %
   \crefformat{part}{Part #2\MakeUppercase{#1}#3}
\usepackage{crossreftools} %

\usepackage{amssymb} %
\usepackage{amsthm} %
\usepackage{mathcommand} %
\usepackage{mathtools} %
\usepackage{nicefrac} %
\usepackage{stmaryrd} %

\ifdraft %
\usepackage[backgroundcolor=orange!20, textcolor={Dark Ruby Red}, textsize=tiny]{todonotes}
\else
\usepackage[backgroundcolor=orange!20, textcolor={Dark Ruby Red}, textsize=tiny,disable]{todonotes}
\fi
\definecolor{green}{RGB}{0,120,0}
\definecolor{hlyellow}{RGB}{250, 250, 190}
\definecolor{diegoeditcolor}{RGB}{210,210,255}
\definecolor{diegochangescolor}{RGB}{41, 32, 138}
\definecolor{meghyneditcolor}{RGB}{210,255,210}
\definecolor{pierreeditcolor}{RGB}{255, 225, 186}
\newcommand{\sidediego}[1]{}
\newcommand{\sidemeghyn}[1]{}
\newcommand{\sidepierre}[1]{}
\newcommand{\meghyn}[1]{}
\newcommand{\pierre}[1]{}

\newcommand{\diego}[1]{}

\newcommand{\siderev}[1]{}
\newcommand{\rev}[1]{}

\renewcommand{\hookrightarrow}{\to} %

\newcommand{\aka}{a.k.a.\ }
\newcommand{\cf}{cf.\ }
\newcommand{\eg}{e.g.}

\newcommand{\ie}{i.e.,\ }

\newcommand{\st}{s.t.\ }
\newcommand{\wrt}{w.r.t.\ }

\DeclarePairedDelimiter\set{\{}{\}} %

\ifdraft %

\else

\fi

\LoopCommands\lettersUppercase[l#1] %
   {\newmathcommand#2{\mathbb{#1}}}
\newcommand{\Nat}{\lN} %

\LoopCommands{ABCDEFGIJKMNQRTUVWXYZ}[#1] %
   {\newmathcommand#2{\mathcal{#1}}}
\LoopCommands{HLOPS}[#1] %
   {\renewmathcommand#2{\mathcal{#1}}}

\newcommand{\vect}[1]{\mathbf{#1}} %

\newcommand{\la}{\leftarrow}

\newcommand{\LRa}{\Leftrightarrow}
\newcommand{\defeq}{\vcentcolon=}
\newcommand{\eqdef}{=\vcentcolon}
\renewcommand{\le}{\leqslant}
\renewcommand{\ge}{\geqslant}

\newcommand{\inc}{\subseteq} %

\definecolor{light-gray}{gray}{0.9} %
\newcommand{\proofcase}[1]{\noindent\colorbox{light-gray}{#1}~~}

\newenvironment{equation-inline}{ %
\refstepcounter{equation}}{\hfill(\theequation)\\}

\theoremstyle{plain}

\newtheoremrep{theorem}{Theorem}

\newtheoremrep{claim}[theorem]{Claim}
   \AddToHook{env/claim/begin}{\crefalias{theorem}{claim}} 
\newtheoremrep{corollary}[theorem]{Corollary}
   \AddToHook{env/corollary/begin}{\crefalias{theorem}{corollary}} 
\newtheoremrep{example}[theorem]{Example}
   \AddToHook{env/example/begin}{\crefalias{theorem}{example}} 
\newtheoremrep{lemma}[theorem]{Lemma}
   \AddToHook{env/lemma/begin}{\crefalias{theorem}{lemma}} 
\newtheoremrep{proposition}[theorem]{Proposition}
   \AddToHook{env/proposition/begin}{\crefalias{theorem}{proposition}} 

\theoremstyle{acmdefinition}
\newtheorem{open}[theorem]{Open Question}
   \AddToHook{env/open/begin}{\crefalias{theorem}{open}} 
\newtheoremrep{remark}[theorem]{Remark}
   \AddToHook{env/remark/begin}{\crefalias{theorem}{remark}} 

\crefname{theorem}{Theorem}{Theorems}
\crefname{claim}{Claim}{Claims}
\crefname{corollary}{Corollary}{Corollaries}
\crefname{example}{Example}{Examples}
\crefname{lemma}{Lemma}{Lemmas}
\crefname{proposition}{Proposition}{Propositions}
\crefname{remark}{Remark}{Remarks}

\tikzset{
	subtree/.pic={
		\coordinate (-a) at (0,0);
		\coordinate (-b) at (-1,-2);
		\coordinate (-c) at (1,-2);
		\coordinate (-d) at (-0.35,-0.35);
		\coordinate (-e) at (0.35,-0.35);
		\coordinate (-h) at (-1.26,-2.16);
		\coordinate (-i) at (1.26,-2.16);
		\coordinate (-j) at (0.35,0.35);
		\coordinate (-k) at (-0.35,0.35);
		\coordinate (-west) at (-0.4,0);
		\coordinate (-east) at (0.4,0);
		
		\draw (-a) -- (-b) -- (-c) -- cycle;
		\draw[fill=white] (-a) circle (0.4);
	}
}

\tikzset{
	subtree_small/.pic={
		\coordinate (-a) at (0,0);
		\coordinate (-b) at (-.75,-1);
		\coordinate (-c) at (.75,-1);
		\coordinate (-d) at (-0.35,-0.35);
		\coordinate (-e) at (0.35,-0.35);
		\coordinate (-h) at (-0.92,-1.06);
		\coordinate (-i) at (0.92,-1.06);
		\coordinate (-j) at (0.35,0.35);
		\coordinate (-k) at (-0.35,0.35);
		\coordinate (-west) at (-0.4,0);
		\coordinate (-east) at (0.4,0);
		
		\draw (-a) -- (-b) -- (-c) -- cycle;
		\draw[fill=white] (-a) circle (0.4);
	}
}

\tikzset{
	subtrinvisible/.pic={
		\coordinate (-a) at (0,0);
		\coordinate (-b) at (-1,-2);
		\coordinate (-c) at (1,-2);
		\coordinate (-d) at (-0.35,-0.35);
		\coordinate (-e) at (0.35,-0.35);
		\coordinate (-h) at (-1.26,-2.16);
		\coordinate (-i) at (1.26,-2.16);
		\coordinate (-j) at (0.35,0.35);
		\coordinate (-k) at (-0.35,0.35);
		\coordinate (-west) at (-0.4,0);
		\coordinate (-east) at (0.4,0);
	}
}

\tikzset{
	graphbox/.pic={
		\node [draw, circle, minimum height=17pt] (-0) at (-1, 0.75) {};
		\node [draw, circle, minimum height=17pt] (-1) at (-1, -0.25) {};
		\node [draw, circle, minimum height=17pt] (-2) at (-1, -1.25) {};
		\node [draw, circle, minimum height=17pt] (-3) at (1, -1.25) {};
		\node [draw, circle, minimum height=17pt] (-4) at (1, -0.25) {};
		\node [draw, circle, minimum height=17pt] (-5) at (1, 0.75) {};
		\node [] (-6) at (0, 1.25) {\Huge $G$};
		\coordinate (-7) at (-1.5, -1.75) {};
		\coordinate (-8) at (1.5, 1.75) {};
		\coordinate (-north) at (0, 1.75);
		\coordinate (-west) at (-1.5, 0);
		\coordinate (-east) at (1.5, 0);
		\coordinate (-south) at (0, -1.75);
		
		\draw (-8) rectangle (-7);
		\draw (-0) to (-5);
		\draw (-1) to (-3);
		\draw (-2) to (-4);
	}
}

\tikzset{
	graphbox-ns/.pic={
		\coordinate (-7) at (-1.5, -1.75) {};
		\coordinate (-8) at (1.5, 1.75) {};
		\draw[fill=white] (-8) rectangle (-7);
		
		\node [draw, circle, minimum height=17pt] (-0) at (-1, 0.75) {};
		\node [draw, circle, minimum height=17pt] (-1) at (-1, -0.25) {};
		\node [draw, circle, minimum height=17pt] (-2) at (-1, -1.25) {};
		\node [draw, circle, minimum height=17pt] (-3) at (1, -1.25) {};
		\node [draw, circle, minimum height=17pt] (-4) at (1, -0.25) {};
		\node [draw, circle, minimum height=17pt] (-5) at (1, 0.75) {};
		\node [] (-6) at (0, 1.25) {\Huge $G$};
		\coordinate (-north) at (0, 1.75);
		\coordinate (-west) at (-1.5, 0);
		\coordinate (-east) at (1.5, 0);
		\coordinate (-south) at (0, -1.75);
		
		\draw[very thick,double distance=3pt] (-north) -- ++(0,0.7);
		\draw[very thick,double distance=3pt,arrows = {-Implies[]}] (-south) -- ++(0,-1.6);
		\draw (-8) rectangle (-7);
		\draw (-0) to (-5);
		\draw (-1) to (-3);
		\draw (-2) to (-4);
	}
}

\tikzset{
	graphbox-ew/.pic={
		\node [draw, circle, minimum height=17pt] (-0) at (-1, 0.75) {};
		\node [draw, circle, minimum height=17pt] (-1) at (-1, -0.25) {};
		\node [draw, circle, minimum height=17pt] (-2) at (-1, -1.25) {};
		\node [draw, circle, minimum height=17pt] (-3) at (1, -1.25) {};
		\node [draw, circle, minimum height=17pt] (-4) at (1, -0.25) {};
		\node [draw, circle, minimum height=17pt] (-5) at (1, 0.75) {};
		\node [] (-6) at (0, 1.25) {\Huge $G$};
		\coordinate (-7) at (-1.5, -1.75) {};
		\coordinate (-8) at (1.5, 1.75) {};
		\coordinate (-north) at (0, 1.75);
		\coordinate (-west) at (-1.5, 0);
		\coordinate (-east) at (1.5, 0);
		\coordinate (-south) at (0, -1.75);
		
		\draw[thick,double distance=3pt] (-west) -- ++(-1,0);
		\draw[thick,double distance=3pt,arrows = {-Implies[]}] (-east) -- ++(1,0);
		\draw (-8) rectangle (-7);
		\draw (-0) to (-5);
		\draw (-1) to (-3);
		\draw (-2) to (-4);
	}
}
\ifexternal
   \usetikzlibrary{external}
   \tikzexternaldisable %

\else
   
\fi
\ifdraft
\usepackage[xcolor, hyperref, cleveref, notion, quotation, composition]{knowledge}
\else\ifpaper
\usepackage[xcolor, hyperref, cleveref, notion, quotation, paper]{knowledge}
\else
\usepackage[xcolor, hyperref, cleveref, notion, quotation, electronic]{knowledge}
\fi\fi

\usepackage{mathcommand}
\knowledgeconfigure{quotation, protect quotation={tikzcd}}
\knowledgeconfigure{diagnose line=true, diagnose bar=true}

\definecolor{Dark Ruby Red}{HTML}{5d1416}
\definecolor{Dark Blue Sapphire}{HTML}{004c5c} %
\definecolor{Dark Gamboge}{HTML}{be7c00}

\IfKnowledgePaperModeTF{
    \knowledgestyle{notion}{color=black}
    \hypersetup{}
}{
    \knowledgestyle{intro notion}{color={Dark Ruby Red}, emphasize}
    \knowledgestyle{notion}{color={Dark Blue Sapphire}}
    \hypersetup{ %
        colorlinks=true,
        breaklinks=true,
        linkcolor={Dark Blue Sapphire}, %
        citecolor={Dark Blue Sapphire}, %
        filecolor={Dark Blue Sapphire}, %
        urlcolor={Dark Blue Sapphire},
    }
    \IfKnowledgeElectronicModeTF{
    }{
        \knowledgeconfigure{anchor point color={Dark Ruby Red}, anchor point shape=corner}
        \knowledgestyle{intro unknown}{color={Dark Gamboge}, emphasize}
        \knowledgestyle{intro unknown cont}{color={Dark Gamboge}, emphasize}
        \knowledgestyle{kl unknown}{color={Dark Gamboge}}
        \knowledgestyle{kl unknown cont}{color={Dark Gamboge}}
    }
}

\makeindex

\knowledgenewrobustcmd{\ann}{\cmdkl{\nu}}
\knowledgenewrobustcmd{\oneann}{\cmdkl{\mathbf{1}}}
\knowledgenewrobustcmd{\adom}{\cmdkl{\textit{adom}}} %

\knowledgenewrobustcmd{\Minsups}[1]{\cmdkl{\mathsf{MS}_{#1}}} %

\knowledgenewrobustcmd{\BPP}{\cmdkl{\mathsf{BPP}}}
\knowledgenewrobustcmd{\coNP}{\cmdkl{\mathsf{coNP}}}
\knowledgenewrobustcmd{\FP}{\cmdkl{\mathsf{FP}}} 
\knowledgenewrobustcmd{\FPsNP}{\cmdkl{\mathsf{FP}^{\mathsf{\#NP}}}} 
\knowledgenewrobustcmd{\FPsP}{\cmdkl{\mathsf{FP}^\mathsf{\#P}}}
\knowledgenewrobustcmd{\FPsPH}{\cmdkl{\mathsf{FP}^{\mathsf{\#PH}}}}
\knowledgenewrobustcmd{\NP}{\cmdkl{\mathsf{NP}}}
\knowledgenewrobustcmd{\PH}{\cmdkl{\mathsf{PH}}}
\knowledgenewrobustcmd{\PsP}{\cmdkl{\mathsf{P}^\mathsf{\#P}}}
\knowledgenewrobustcmd{\Ptime}{\cmdkl{\mathsf{P}}}
\knowledgenewrobustcmd{\sNP}{\cmdkl{\mathsf{\#NP}}}
\knowledgenewrobustcmd{\sP}{\cmdkl{\mathsf{\#P}}} 
\knowledgenewrobustcmd{\sPH}{\cmdkl{\mathsf{\#PH}}}

\knowledgenewmathcommand{\numBipVerCov}{\cmdkl{\mathsf{\#BipVerCov}}}
\knowledgenewrobustcmd{\numBipIndep}{\cmdkl{\mathsf{\#BipIndepSet}}} %
\knowledgenewrobustcmd{\numPerMatch}{\cmdkl{\mathsf{\#PerfMatch}}} %
\knowledgenewrobustcmd{\numSTConn}{\cmdkl{\mathsf{\#stConnect}}} %
\knowledgenewmathcommand{\VerCov}{\cmdkl{\mathsf{VerCov}}}
\knowledgenewmathcommand{\VerCovFun}{\cmdkl{\mathrm{VerCov}}}

\knowledgenewmathcommandPIE{\countMS}{%
   \cmdkl{\#}^{\cmdkl{\textsf{ms}}}#1#2#3}
\knowledgenewmathcommandPIE{\countFMS}{%
   \cmdkl{\#}^{\cmdkl{\textsf{fms}}}#1#2#3}
\knowledgenewmathcommandPIE{\countAns}{%
   \cmdkl{\#}^{\cmdkl{\textsf{hom}}}#1#2#3}

\knowledgenewmathcommandPIE{\evalCountMS}{%
   \cmdkl{\textsc{eval-}\#}^{\cmdkl{\textsf{ms}}}#1#2#3}
\knowledgenewmathcommandPIE{\evalCountFMS}{%
   \cmdkl{\textsc{eval-}\#}^{\cmdkl{\textsf{fms}}}#1#2#3}
\knowledgenewmathcommandPIE{\evalCountAns}{%
   \cmdkl{\textsc{eval-}\#}^{\cmdkl{\textsf{hom}}}#1#2#3}

\knowledgenewmathcommandPIE{\GIMC}{%
   \cmdkl{\mathsf{GIMC}#1#2#3}}
\knowledgenewmathcommandPIE{\FGIMC}{%
   \cmdkl{\mathsf{FGIMC}#1#2#3}}

\knowledgenewmathcommandPIE{\PQE}{%
   \cmdkl{\mathsf{PQE}#1#2#3}}
\knowledgenewmathcommandPIE{\SPQE}{%
   \cmdkl{\mathsf{SPQE}#1#2#3}}
\knowledgenewmathcommandPIE{\SPPQE}{%
   \cmdkl{\mathsf{SPPQE}#1#2#3}}
\knowledgenewmathcommandPIE{\PQEPhalf}{%
   \cmdkl{\mathsf{PQE}#1#2#3\!\left(\nicefrac 1 2\right)}}
\knowledgenewmathcommandPIE{\PQEPhalfOne}{%
   \cmdkl{\mathsf{PQE}#1#2#3\!\left(\nicefrac 1 2 ; 1\right)}}
\knowledgenewmathcommandPIE{\CPQE}{%
   \cmdkl{\mathsf{CPQE}#1#2#3}}
\knowledgenewmathcommandPIE{\CSPQE}{
   \cmdkl{\mathsf{CSPQE}#1#2#3}}
\knowledgenewmathcommandPIE{\CSPPQE}{
   \cmdkl{\mathsf{CSPPQE}#1#2#3}}
\knowledgenewmathcommandPIE{\CPQEPhalf}{
   \cmdkl{\mathsf{CPQE}#1#2#3\!\left(\nicefrac 1 2\right)}}
\knowledgenewmathcommandPIE{\CPQEPhalfOne}{
   \cmdkl{\mathsf{CPQE}#1#2#3\!\left(\nicefrac 1 2 ; 1\right)}}

\knowledgenewmathcommandPIE{\stShapley}{%
   \cmdkl{\mathsf{SVC}}^{\cmdkl{\star}}#1#2#3}

\knowledgenewmathcommandPIE{\dShapley}{%
   \cmdkl{\mathsf{SVC}}^{\cmdkl{\textsf{dr}}}#1#2#3}
\newmathcommandPIE{\dShapleyNoKL}{%
   \mathsf{SVC}^{\textsf{dr}}#1#2#3}
\knowledgenewmathcommandPIE{\dnShapley}{%
   \cmdkl{\mathsf{nSVC}}^{\cmdkl{\textsf{dr}}}#1#2#3}
\knowledgenewmathcommandPIE{\dcShapley}{%
   \cmdkl{\mathsf{CSVC}}^{\cmdkl{\textsf{dr}}}#1#2#3}
\knowledgenewmathcommandPIE{\dcnShapley}{%
   \cmdkl{\mathsf{nCSVC}}^{\cmdkl{\textsf{dr}}}#1#2#3}

\knowledgenewmathcommandPIE{\mcShapley}{%
   \cmdkl{\mathsf{SVC}}^{\cmdkl{\textsf{MC}}}#1#2#3}
\knowledgenewmathcommandPIE{\pShapley}{%
   \cmdkl{\mathsf{SVC}}^{\cmdkl{\textsf{P}}}#1#2#3}
\knowledgenewmathcommandPIE{\rShapley}{%
   \cmdkl{\mathsf{SVC}}^{\cmdkl{\textsf{R}}}#1#2#3}
\knowledgenewmathcommandPIE{\saShapley}{%
   \cmdkl{\mathsf{SVC}}^{\cmdkl{\textsf{hom}}}#1#2#3}

\newcommand{\minsupindex}{\textsf{ms}}
\knowledgenewmathcommandPIE{\msShapley}{%
   \cmdkl{\mathsf{SVC}}^{\cmdkl{\minsupindex}}#1#2#3}
\knowledgenewmathcommandPIE{\sShapley}{%
   \cmdkl{\mathsf{SVC}}^{\cmdkl{\textsf{s}}}#1#2#3}
\knowledgenewmathcommandPIE{\sharpShapley}{%
   \cmdkl{\mathsf{SVC}}^{\cmdkl{\textsf{\#}}}#1#2#3}
\knowledgenewmathcommandPIE{\wShapley}{%
   \cmdkl{\mathsf{SVC}}#1#2#3}

\newcommand{\subendo}{{\textup{\textsf{n}}}}
\newcommand{\subexo}{{\textup{\textsf{x}}}}
\knowledgenewrobustcmd{\Dn}[1][\D]{#1_{\cmdkl{\subendo}}}
\knowledgenewrobustcmd{\Dx}[1][\D]{#1_{\cmdkl{\subexo}}}

\knowledgenewrobustcmd{\constn}{\cmdkl{\textit{const}_{\cmdkl{\subendo}}}}
\knowledgenewrobustcmd{\constx}{\cmdkl{\textit{const}_{\cmdkl{\subendo}}}}

\knowledgenewmathcommandPIE{\games}{%
   \cmdkl{\mathcal{F}^{\emptyset \mapsto 0}#1#2#3}}
\knowledgenewmathcommandPIE{\sgames}{%
   \cmdkl{\mathcal{SF}^{\emptyset \mapsto 0}#1#2#3}}

\knowledgenewrobustcmd{\Sh}{\cmdkl{\mathrm{Sh}}} %
\knowledgenewrobustcmd{\Bz}{\cmdkl{\mathrm{Bz}}} %

\knowledgenewmathcommandPIE{\scorefun}{%
   \cmdkl{\xi}#1#2#3}
\knowledgenewmathcommandPIE{\STscorefun}{%
   \cmdkl{\Xi}^{\cmdkl{\mathsf{\star}}}#1#2#3}
\knowledgenewmathcommandPIE{\stscorefun}{%
   \cmdkl{\xi}^{\cmdkl{\mathsf{\star}}}#1#2#3}

\knowledgenewmathcommandPIE{\onems}{%
   \cmdkl{\xi}#1#2#3}%
\knowledgenewmathcommandPIE{\SHAPscorefun}{%
   \cmdkl{\Xi}^{\cmdkl{\textsf{SHAP}}}#1#2#3}
\knowledgenewmathcommandPIE{\shapscorefun}{%
   \cmdkl{\xi}^{\cmdkl{\textsf{SHAP}}}#1#2#3}
\knowledgenewmathcommandPIE{\zetascorefun}{%
   \cmdkl{\zeta}#1#2#3}

\knowledgenewmathcommandPIE{\Dscorefun}{%
   \cmdkl{\Xi}^{\cmdkl{\textsf{dr}}}#1#2#3}
\knowledgenewmathcommandPIE{\dscorefun}{%
   \cmdkl{\xi}^{\cmdkl{\textsf{dr}}}#1#2#3}

\knowledgenewmathcommandPIE{\MCscorefun}{%
   \cmdkl{\Xi}^{\cmdkl{\textsf{MC}}}#1#2#3}
\knowledgenewmathcommandPIE{\mcscorefun}{%
   \cmdkl{\xi}^{\cmdkl{\textsf{MC}}}#1#2#3}
\knowledgenewmathcommandPIE{\Pscorefun}{%
   \cmdkl{\Xi}^{\cmdkl{\textsf{P}}}#1#2#3}
\knowledgenewmathcommandPIE{\pscorefun}{%
   \cmdkl{\xi}^{\cmdkl{\textsf{P}}}#1#2#3}
\knowledgenewmathcommandPIE{\Rscorefun}{%
   \cmdkl{\Xi}^{\cmdkl{\textsf{R}}}#1#2#3}
\knowledgenewmathcommandPIE{\rscorefun}{%
   \cmdkl{\xi}^{\cmdkl{\textsf{R}}}#1#2#3}
\knowledgenewmathcommandPIE{\SAscorefun}{%
   \cmdkl{\Xi}^{\cmdkl{\textsf{hom}}}#1#2#3}
\knowledgenewmathcommandPIE{\sascorefun}{%
   \cmdkl{\xi}^{\cmdkl{\textsf{hom}}}#1#2#3}

\knowledgenewmathcommandPIE{\MSscorefun}{%
   \cmdkl{\Xi}^{\cmdkl{\minsupindex}}#1#2#3}
\knowledgenewmathcommandPIE{\msscorefun}{%
   \cmdkl{\xi}^{\cmdkl{\minsupindex}}#1#2#3}
\knowledgenewmathcommandPIE{\Sscorefun}{%
   \cmdkl{\Xi}^{\cmdkl{\textsf{s}}}#1#2#3}
\knowledgenewmathcommandPIE{\sscorefun}{%
   \cmdkl{\xi}^{\cmdkl{\textsf{s}}}#1#2#3}
\knowledgenewmathcommandPIE{\SHARPscorefun}{%
   \cmdkl{\Xi}^{\cmdkl{\textsf{\#}}}#1#2#3}
\knowledgenewmathcommandPIE{\sharpscorefun}{%
   \cmdkl{\xi}^{\cmdkl{\textsf{\#}}}#1#2#3}
\knowledgenewmathcommandPIE{\Wscorefun}{%
   \cmdkl{\Xi}#1#2#3}
\knowledgenewmathcommandPIE{\wscorefun}{%
   \cmdkl{\xi}#1#2#3}

\knowledgenewrobustcmd{\sigmaless}[1]{\cmdkl{\sigma}_{\!\cmdkl{<}#1}}
\knowledgenewrobustcmd{\sigmaleq}[1]{\cmdkl{\sigma}_{\!\cmdkl{\le}#1}}

\knowledgenewrobustcmd{\bse}[1]{\cmdkl{X_{#1}}} %

\knowledgenewrobustcmd{\atoms}{\cmdkl{\textit{atoms}}}
\knowledgenewrobustcmd{\arity}{\cmdkl{\mathrm{arity}}}
\knowledgenewrobustcmd{\const}{\cmdkl{\textit{const}}}
\knowledgenewrobustcmd{\Const}{\cmdkl{\mathsf{Const}}} %
\knowledgenewrobustcmd{\dom}{\cmdkl{\mathrm{dom}}} %
\knowledgenewrobustcmd{\mterms}{\cmdkl{\textit{term}}}
\knowledgenewrobustcmd{\vars}{\cmdkl{\textit{vars}}}
\knowledgenewrobustcmd{\Var}{\cmdkl{\mathsf{Var}}}

\knowledgenewrobustcmd{\Unif}[1][q]{\cmdkl{\mathcal{M}}_{#1}}

\knowledgenewmathcommandPIE{\IsubA}{%
   \cmdkl{\I#1#2#3}}

\knowledgenewrobustcmd{\dnames}{\cmdkl{\mathsf{N_D}}}
\knowledgenewrobustcmd{\cnames}{\cmdkl{\mathsf{N_C}}}
\knowledgenewrobustcmd{\rnames}{\cmdkl{\mathsf{N_R}}}
\knowledgenewrobustcmd{\inames}{\cmdkl{\mathsf{N_I}}}
\knowledgenewrobustcmd{\nulls}{\cmdkl{\mathsf{N_U}}}
\knowledgenewrobustcmd{\vnames}{\cmdkl{\mathsf{N_V}}}
\knowledgenewrobustcmd{\irnames}{\cmdkl{\mathsf{N^{\pm}_R}}}
\knowledgenewrobustcmd{\terms}{\cmdkl{\mathsf{terms}}}
\knowledgenewrobustcmd{\mods}{\cmdkl{\mathsf{Mod}}}

\knowledgenewrobustcmd{\withT}[1]{(\T,#1)} %

\knowledgenewrobustcmd{\omqsat}{\mathrel{\cmdkl{\models}}} %

\knowledgenewrobustcmd{\dllitec}{\cmdkl{\ensuremath{\mathsf{DL\text{-}Lite}_{\mathsf{core}}}}}
\knowledgenewrobustcmd{\dlliter}{\cmdkl{\ensuremath{\mathsf{DL\text{-}Lite}_{\mathsf{core}}^{\mathcal{H}}}}}
\knowledgenewrobustcmd{\dlliteh}{\cmdkl{\ensuremath{\mathsf{DL\text{-}Lite}_{\mathsf{Horn}}}}}
\knowledgenewrobustcmd{\EL}{\cmdkl{\mathcal{EL}}}
\knowledgenewrobustcmd{\elhibot}{\cmdkl{\mathcal{ELHI}_{\bot}}}

\knowledgenewrobustcmd{\Lmin}{\cmdkl{\L_{\min}}} %

\knowledgenewrobustcmd{\partsof}[1]{\cmdkl{\mathcal{P}}(#1)} %
\knowledgenewrobustcmd{\fpartsof}[1]{\cmdkl{\mathcal{P}_{\mathsf{f}}}(#1)} %
\knowledgenewrobustcmd{\Perm}{\cmdkl{\mathbb{P}}} %
\knowledgenewrobustcmd{\Totord}{\cmdkl{\mathcal{O}}} %

\knowledgenewrobustcmd{\lb}{\cmdkl{\mathrm{lb}}}%

\knowledgenewrobustcmd{\Esp}{\mathbf{E}} %
\knowledgenewrobustcmd{\Prob}{\mathbf{P}} %

\knowledgenewrobustcmd{\dcup}{\mathbin{\cmdkl{\uplus}}} %
\knowledgenewrobustcmd{\bigdcup}{\mathop{\cmdkl{\biguplus}}} %

\knowledgenewrobustcmd{\homto}[1][]{\mathrel{\cmdkl{\xrightarrow{\smash{\textit{\tiny #1 \!hom}}}}}} %
\knowledgenewrobustcmd{\Chomto}[1][]{\mathrel{\cmdkl{\xrightarrow{\smash{\textit{\tiny #1 \!$\C$-hom}}}}}} %
\knowledgenewrobustcmd{\Pichomto}[1][]{\mathrel{\cmdkl{\xrightarrow{\smash{\textit{\tiny #1 \!$\Pic$-h}}}}}} %
\knowledgenewrobustcmd{\polyrx}{ %
   \mathrel{\cmdkl{\le_{\mathsf{P}}}}}
\knowledgenewrobustcmd{\polyeq}{ %
   \mathrel{\cmdkl{\equiv_{\mathsf{P}}}}}

\knowledgenewrobustcmd{\class}{\mathcal{C}} %

\knowledgenewmathcommand{\ACQ}{\cmdkl{\mathsf{ACQ}}} %
\knowledgenewmathcommand{\CQ}{\cmdkl{\mathsf{CQ}}} %
\knowledgenewmathcommand{\CQeq}{\cmdkl{\mathsf{CQ}^{=}}} %
\knowledgenewmathcommand{\CQeqneq}{\cmdkl{\mathsf{CQ}^{\neq,=}}} %
\knowledgenewmathcommand{\CQneq}{\cmdkl{\mathsf{CQ}^{\neq}}} %
\knowledgenewmathcommand{\CRPQ}{\cmdkl{\mathsf{CRPQ}}} %
\knowledgenewmathcommand{\IQ}{\cmdkl{\mathsf{IQ}}} %
\knowledgenewmathcommand{\RPQ}{\cmdkl{\mathsf{RPQ}}} %
\knowledgenewmathcommand{\sjfACQ}{\cmdkl{\mathsf{sjfACQ}}} %
\knowledgenewmathcommand{\sjfCQ}{\cmdkl{\mathsf{sjfCQ}}} %
\knowledgenewmathcommand{\UCQ}{\cmdkl{\mathsf{UCQ}}} %
\knowledgenewtextcommand{\UCQneg}{\cmdkl{UCQ$^\lnot$}} %
\knowledgenewrobustcmd{\UCQneq}{\cmdkl{\mathsf{UCQ}^{\neq}}} %
\knowledgenewmathcommand{\UCQneg}{\cmdkl{\mathsf{UCQ}^\lnot}} %
\knowledgenewmathcommand{\UCRPQ}{\cmdkl{\mathsf{UCRPQ}}} %
\knowledgenewmathcommand{\UsjfACQ}{\cmdkl{\mathsf{UsjfACQ}}} %

\knowledgenewmathcommandPIE{\onemsq}{%
   \cmdkl{q#2}#1#3}%

\knowledgenewrobustcmd{\aC}{\cmdkl{C}} %
\knowledgenewrobustcmd{\aCij}[1][i,j]{\cmdkl{C_{#1}}} %
\knowledgenewrobustcmd{\SNij}[1][i,j]{\cmdkl{\S^N_{#1}}} %
\knowledgenewrobustcmd{\SNijtld}[1][i,j]{\cmdkl{\tilde{\S}^N_{#1}}} %
\knowledgenewrobustcmd{\Ak}[1][k]{\cmdkl{\A^{#1}}} %
\knowledgenewrobustcmd{\Akchi}[2][k]{\cmdkl{\A_{#2}^{#1}}}
\knowledgenewrobustcmd{\Ao}[1][]{\cmdkl{\A^{\circ}_{#1}}} %

\knowledgenewrobustcmd{\AG}[1][]{\cmdkl{\A^{G}_{#1}}} %

\knowledgenewrobustcmd{\PG}[1][]{\cmdkl{\P^{G}_{#1}}} %

\knowledgenewrobustcmd{\Axto}{\cmdkl{\A^{\to}_{\subexo}}} %
\knowledgenewrobustcmd{\Axla}{\cmdkl{\A^{\la}_{\subexo}}} %
\knowledgenewrobustcmd{\Below}{\cmdkl{\mathbf{B}}} %
\knowledgenewrobustcmd{\Left}{\cmdkl{\mathbf{L}}} %
\knowledgenewrobustcmd{\Right}{\cmdkl{\mathbf{R}}} %

\knowledgenewrobustcmd{\qreach}[1][st]{\cmdkl{#1\textit{-reach}}}
\knowledgenewrobustcmd{\Reach}[1]{\cmdkl{\textsf{CC}(}#1\cmdkl{)}}
\knowledgenewrobustcmd{\rhoA}[1][\A]{\cmdkl{\rho_{#1}}}

\knowledgenewrobustcmd{\core}{\cmdkl{\textit{core}}}
\knowledgenewrobustcmd{\emptytup}{\cmdkl{()}}
\knowledgenewrobustcmd{\hyperq}[1][q]{\cmdkl{\mathbf{G}_{#1}}}
\knowledgenewrobustcmd{\primalq}[1][q]{\cmdkl{\mathbf{G}^p_{#1}}}
\knowledgenewrobustcmd{\dimtup}{\cmdkl{\dim}}
\knowledgenewrobustcmd\vertex[1]{\cmdkl{V}(#1)}
\knowledgenewrobustcmd\edges[1]{\cmdkl{E}(#1)}

\knowledgenewrobustcmd{\maxSize}[1]{\cmdkl{\|}#1\cmdkl{\|}}
\knowledgenewrobustcmd{\sizeofD}[1][D]{\cmdkl{|}#1\cmdkl{|}}
\knowledgenewrobustcmd{\normOne}[1]{\cmdkl{\|}#1\cmdkl{\|_1}}
\knowledgenewrobustcmd{\normInf}[1]{\cmdkl{\|}#1\cmdkl{\|_{\infty}}}

\knowledgenewrobustcmd{\sizeofq}[1][q]{\cmdkl{|}#1\cmdkl{|}}
\knowledgenewrobustcmd{\sizeofgamma}[1][\gamma]{\cmdkl{|}#1\cmdkl{|}}

\knowledgenewrobustcmd{\reducesto}{\mathrel{\cmdkl{\leq_{\textit{poly}}}}}

\knowledgenewrobustcmd\bagmap{\cmdkl{\mathbf{b}}}
\knowledgenewrobustcmd\atommap{\cmdkl{\mathbf{a}}}
\knowledgenewrobustcmd\atommaplab{\cmdkl{\mathbf{\tilde a}}}
\knowledgenewrobustcmd\tagmap{\cmdkl{\mathbf{t}}}
\knowledgenewrobustcmd\tagmappath[1]{\cmdkl{\mathbf{t}[#1]}}
\newrobustcmd\tagmappathprime[1]{%
  \withkl{\kl[\tagmappath]}{%
    \cmdkl{\mathbf{t}'[#1]}%
  }%
}
\knowledgenewrobustcmd{\ghw}{\cmdkl{\textit{ghw}}}
\knowledgenewrobustcmd{\fghw}{\cmdkl{\textit{fghw}}}
\knowledgenewrobustcmd{\pghw}{\cmdkl{\textit{pghw}}}
\knowledgenewrobustcmd{\hw}{\cmdkl{\textit{hw}}}
\knowledgenewrobustcmd{\tw}{\cmdkl{\textit{tw}}}

\knowledgenewrobustcmd{\sjoin}[1][]{\mathop{\cmdkl{\ltimes_{#1}}}}
\knowledgenewrobustcmd{\costjoin}{\cmdkl{c_{sj}}}

\knowledgenewrobustcmd{\abotimes}[2]{\cmdkl{\bigotimes_{#1}}#2} %
\knowledgenewrobustcmd{\aboplus}[2]{\cmdkl{\bigoplus_{#1}}#2} %

\knowledgenewrobustcmd{\contr}{\cmdkl{\textit{contract}}}
\knowledgenewrobustcmd{\Mfacts}{\cmdkl{\mathbf{M}}}
\knowledgenewrobustcmd{\DBs}[1][\Sigma]{\cmdkl{\textup{DB}_{#1}}}
\knowledgenewrobustcmd{\evalPb}[1]{\cmdkl{\textup{\textsc{Eval-}}}#1}
\knowledgenewrobustcmd{\aug}[1]{{#1}^{\cmdkl{+}}}%
\knowledgenewrobustcmd{\restrictG}[2][G]{#1\cmdkl{[}{#2}\cmdkl{]}}
\knowledgenewrobustcmd{\evalCounting}[2]{\cmdkl{\#}{#1}\cmdkl{\langle}#2\cmdkl{\rangle}}
\knowledgenewrobustcmd{\counting}[1]{\cmdkl{\#}{#1}}
\knowledgenewrobustcmd{\PRel}{\cmdkl{P}}

\knowledgenewrobustcmd{\qterms}{\cmdkl{\mathsf{terms}}}
\knowledgenewrobustcmd{\Equivs}{\cmdkl{\E}}
\knowledgenewrobustcmd{\qE}[1][E]{\cmdkl{q_{#1}}}
\knowledgenewrobustcmd{\qEneq}[1][E]{\cmdkl{q_{#1}^{\neq}}}
\knowledgenewrobustcmd{\Punif}{\cmdkl{\mathbf{P}}_q}
\knowledgenewrobustcmd{\iso}{\cmdkl{\simeq}}
\knowledgenewrobustcmd{\Auto}[1]{\cmdkl{\textnormal{Aut}}(#1)}
\knowledgenewrobustcmd{\elimEq}[1]{\cmdkl{\widehat{#1}}}

\knowledgenewrobustcmd{\poneone}[1]{\cmdkl{q_{#1}^{\exists!}}}%

\knowledgenewrobustcmd{\RelatedTerms}{\cmdkl{\+R}}

\knowledgenewrobustcmd{\Qneq}{\cmdkl{\mathcal{Q}^{\neq}}}

\input{kl/knowledges.kl} %

\begin{document}
\selectlanguage{british}  %

\frontpagedeclarations

\begin{abstract}
   The Shapley value, originating from cooperative game theory, has been employed to define responsibility measures that quantify the contributions of database facts to obtaining a given query answer. For non-numeric queries, this is done by considering a cooperative game whose players are the facts and whose wealth function assigns 1 or 0 to each subset of the database, depending on whether the query answer holds in the given subset. While conceptually simple, this approach suffers from a notable drawback: the problem of computing such Shapley values is \#P-hard in data complexity, even for simple conjunctive queries. 
This motivates us to revisit the question of what constitutes a reasonable responsibility measure and to introduce a new family of responsibility measures -- weighted sums of minimal supports (WSMS) -- which satisfy intuitive properties. Interestingly, while the definition of WSMSs is simple and bears no obvious resemblance to the Shapley value formula, 
we prove that every WSMS measure can be equivalently seen as the Shapley value of a suitably defined cooperative game. Moreover, WSMS measures enjoy tractable data complexity for a large class of queries, including all unions of conjunctive queries. We further explore the combined complexity of WSMS computation and establish (in)tractability results for various subclasses of conjunctive queries. 
Interestingly, the proofs of the tractability results for combined complexity reveal connections between counting minimal supports and counting answers to conjunctive queries.

\end{abstract}



\maketitle

\noindent
\raisebox{-.4ex}{\HandRight}\ \ This pdf contains internal links: clicking on a "notion@@notice" leads to its \AP ""definition@@notice"".%

\smallskip
\noindent\raisebox{-.4ex}{\HandRight}\ \ This is an extended version of the PODS'25 paper \cite{ourpods25paper} (see \Cref{sec:deltaconference} for differences).

\section{Introduction}
\label{sec:intro}
Responsibility measures assign scores to database facts based upon how 
much they contribute to the obtention of a given query answer, thereby providing a quantitative notion of explanation for query results. 
In the present paper, we shall take a fresh look at how to define and compute such responsibility measures, focusing on 
the class of monotone non-numeric queries,
which notably includes %
(unions of) conjunctive queries and path queries. 
Although several responsibility measures have been explored in the database literature \cite{DBLP:journals/pvldb/MeliouGMS11,abramovichBanzhafValuesFacts2024,salimiQuantifyingCausalEffects2016,livshitsShapleyValueTuples2021}, 
one particular measure, based upon the Shapley value, 
has been the focus of much of the recent work \cite{livshitsShapleyValueTuples2021,DeutchFKM22ComputingShapley,KhalilK23ShapleyRPQ,KaraOlteanuSuciuShapleyBack,ourpods24,ReshefKL20,karmakarExpectedShapleyLikeScores2024}.

We recall that the Shapley value \cite{shapley1953value} was first introduced in game theory for the purpose of %
fairly distributing wealth %
amongst the players in a cooperative game. Formally, such a game is defined using a  ``wealth function'' $\omega: 2^{\text{Players}} \to \lR_{\geq 0}$ which assigns a number to every subset of players. The Shapley value then defines a function $\phi_\omega: \text{Players} \to \lR$ from players to numbers, which measures their individual contributions. The Shapley function has been famously characterized as the only function that satisfies some %
fundamental axioms (as detailed in \Cref{sec:drastic-shapley}).

In the database context, the Shapley value has been utilized for distributing the responsibility for a given query answer among database facts. This is done by considering games in which the facts are the players and the query is captured using the wealth function.
\AP
For ""numeric queries"" (\ie queries yielding a number when evaluated on a database), the sensible choice for the wealth function is the query itself \cite{livshitsShapleyValueTuples2021}. 
That is, for a given input numeric query $q$, database $\D$, and set of facts $S\subseteq \D$ over which we want to distribute the contribution -- known as `endogenous' facts for historical reasons --, the Shapley value distributes the total wealth $q(\D)=N$ among the facts of $S$ (\ie the `players') using $q$ itself as the wealth function.

For non-numeric queries, whose output is a set of tuples of constants rather than a number,
it is less obvious how to define the wealth function. The solution that was originally proposed in \cite{livshitsShapleyValueTuples2021}
and adopted in all subsequent studies %
 is to first reduce to the Boolean case by considering 
the Boolean query $q(\vect{a})$ associated with the 
input query $q$ and answer tuple $\vect{a}$, and then to treat Boolean queries as numeric queries that output 0 or 1
depending on whether the query holds true in a set of facts. We shall refer to the Shapley value obtained by 
encoding queries using such a 0/1 wealth function %
as the \emph{drastic}
Shapley value, due to it being very similar in spirit to the drastic inconsistency measure \cite{hunterMeasureConflictsShapley2010} (see \Cref{ssec:inconsistency} for more discussion on inconsistency measures).

The use of the drastic Shapley value for responsibility attribution of non-numeric %
queries has received %
significant attention of late \cite{livshitsShapleyValueTuples2021,DeutchFKM22ComputingShapley,KhalilK23ShapleyRPQ,KaraOlteanuSuciuShapleyBack,ReshefKL20,ourpods24,karmakarExpectedShapleyLikeScores2024}. %
The main focus has been on 
precisely identifying the data complexity %
of computing the drastic Shapley value, 
considering various classes of queries and drawing connections to the related tasks of probabilistic
query evaluation and model counting. The key takeaway %
is that 
drastic Shapley value computation is computationally challenging, %
being $\FPsP$-complete in data complexity 
even for some very simple conjunctive queries. 
The situation is more favourable if one moves from exact to approximate computation:  Monte-Carlo sampling provides a fully polynomial randomized approximation scheme (FPRAS) for every UCQ \cite[Corollary 4.13]{livshitsShapleyValueTuples2021}, even if practical implementations appear to be challenging (see experiments and discussions in \cite[§6.2]{DeutchFKM22ComputingShapley}).

Another potential drawback to the drastic Shapley value is 
that it may be hard to interpret the resulting scores,
in particular to understand %
how a change in the database would impact the relative scores of a pair of facts, which may be influenced by seemingly unrelated facts (as exemplified by \Cref{ex:aotbe}). 
In its original conception, the Shapley value uses a wealth function whose numbers correspond to some `cost' or `wealth' in some implicit \emph{unit of measure} (\eg\  dollars, 
number of workers,
number of tuples satisfying some property). 
This same idea carries over to numeric queries, as  %
the `total wealth' to distribute is given by the query itself, which outputs a number with a clear meaning. 
However, the drastic Shapley approach casts the Boolean values `true' and `false' into the (real) numbers 1 and 0, in order to 
apply operations like mean.
This can make the interpretation of the drastic Shapley value not always intuitive to a non-expert end user, which might hence hinder its use for explainability tasks.

    Given both the high computational complexity and non-obvious interpretability of the drastic Shapley value, 
a natural question is whether we can find another responsibility measure for non-numeric queries that 
is computationally and conceptually simpler.
One might be tempted to reply `no', arguing that the uniqueness
result for the Shapley axioms forces us to adopt the drastic Shapley value. However, these axioms were 
introduced for cooperative games, rather than responsibility measures for queries, and their meaning 
crucially depends on the manner in which queries are translated into games. In fact, we shall see that 
when these axioms are translated from games to responsibility measures, using the drastic wealth function,
only two of the four axioms yield unarguably desirable properties for such measures (and one axiom
is actually nonsensical).

In view of the preceding discussion, it is both natural and possible to explore alternative 
responsibility measures for non-numeric queries. This will be the main topic of the present paper, 
whose conceptual and technical contributions are summarized in what follows.

\subsection{Contributions} 

Our first conceptual contribution is to propose a set of  %
\emph{desirable properties for responsibility measures}, focusing on the case of monotone non-numeric queries. 
We begin by examining in \Cref{sec:drastic-shapley}  the properties that result from mapping the Shapley axioms into the database setting using the drastic Shapley wealth function. We observe that only two %
yield clearly desirable properties in our setting, while the other two are either nonsensical or of debatable interest. %
As this leaves us with only two, rather weak, properties, we introduce in \Cref{sec:desirableaxioms} two additional properties 
that capture intuitions as to how a reasonable responsibility measure should behave.

Our second conceptual contribution is to propose in \Cref{ssec:wsmsdef} a \emph{novel family of responsibility measures} specifically tailored to monotone non-numeric queries, which we name \emph{weighted sums of minimal supports} (\emph{WSMS}). As the name suggests, the measure of a fact consists of a weighted sum of the sizes of minimal supports of the query (answer) that contain the fact. 
These measures have an intuitive definition, they enjoy the identified desired properties, and the values arising from these measures %
have a simple interpretation.
Further,  we can show that each measure from this family can be equivalently seen as the Shapley value of a suitable cooperative game (\Cref{ssec:apr2}).
 We also argue that the family is flexible and can accommodate weight functions giving more importance to small supports than to the number of supports, or vice-versa (\Cref{ssec:otherwsms}).

Our main technical contribution is a \emph{complexity analysis of WSMS computation}.  In \Cref{sec:datacomplexity},  we show that, in contrast to the drastic Shapley value,  WSMSs enjoy tractable data complexity for all unions of conjunctive queries (and more generally, bounded queries) via a straightforward algorithm. 
For combined complexity, we obtain a general $\FPsP$ bound for a wide range of query classes (\Cref{sec:combined-complexity}). 
We then focus on conjunctive queries (CQs) and show that, unfortunately, even for acyclic CQs without constants, computing WSMS is $\sP$-hard, and that this hardness results also holds for unions of acyclic "self-join free" queries (\Cref{sec:hardness-(U)CQs}). 
However, for acyclic self-join free CQs, the WSMS computation can be reduced to evaluating `counting CQs' (more precisely, counting the number of homomorphisms from acyclic CQs), which is known to be tractable. 
In \Cref{ssec:counting-homs-ms-char} we investigate the conditions under which counting homomorphisms coincides with counting minimal supports and present a characterization result. %
Further, in \Cref{ssec:cq-bounded-sj}
we show that if both the generalized hypertree width and the amount of self-joins (via a measure we define for this purpose) are bounded, then WSMSs on CQs can be computed in polynomial time.

As a final contribution, we explore the \emph{properties of some related measures}. %
First, having observed that WSMSs are closely connected to the `MI inconsistency measure' of \cite{hunterMeasureConflictsShapley2010}, we define in \Cref{ssec:inconsistency}
three new responsibility measures based upon inconsistency measures and show that they are less suitable for semantic or computational reasons. We likewise show in \Cref{ssec:issues_sash} that redefining WSMSs by counting homomorphisms rather than minimal supports yields a measure which violates a basic desirable property. Finally, in \Cref{ssec:shap}, we recall the SHAP score \cite{lundbergUnifiedApproachInterpreting2017} that has been used for explaining classifiers in machine learning and suggest how the WSMS definition could be adapted to fix some semantic problems with the SHAP score (albeit at the cost of a higher computational complexity).

\subsection{Differences with Respect to Conference Paper}
\label{sec:deltaconference}
The current article is based on the conference paper \cite{ourpods25paper}. 
While the main results are essentially the same, we have included detailed proofs, added discussion, and improved explanations. In particular, the definition of the Symmetry axiom has been simplified to reflect better the intuition behind it, still yielding a set of axioms equivalent to that of Shapley's original seminal paper \cite{shapley1953value} (\cf \Cref{sec:originalShapleyAxioms}). 
We have also included additional technical content not present in \cite{ourpods25paper}. 
Firstly, we include the result that counting minimal supports for unions of "self-join free" "acyclic" queries is $\sP$-hard (\Cref{prop:UsjfACQ-MS-hard}) together with its full proof.
We also include the result, with its full proof, that computing "WSMS" scores of "UCQ"s can be implemented by evaluating simple SQL queries in parallel (\Cref{thm:rewritting-SQL,cor:redux-to-SQL}).
Both these results first appeared in another conference paper \cite[Proposition 13, Theorem 5, and Corollary 8]{ourKR25} but with only short proof sketches. 
Finally, all of 
\Cref{ssec:counting-homs-ms-char}, which studies and characterizes classes of "CQs" for which "homomorphisms" are in bijection with  "minimal supports", 
 is entirely new.

\section{Preliminaries} 
\label{sec:prelims}
\AP
We fix disjoint infinite sets $\intro*\Const$, $\intro*\Var$ of ""constants"" and ""variables"", respectively. 
For any syntactic object $O$ (\eg\ database, query), we will use \AP$\intro*\vars(O)$ and $\intro*\const(O)$ 
to denote the sets of "variables" and "constants" contained in $O$, and let \AP$\intro*\mterms(O)\defeq \vars(O) \cup \const(O)$ denote its set of ""terms"".\phantomintro{\qterms}

\AP A ""(relational) schema"" %
is a finite set of relation symbols, each associated with a (positive) arity. 
\AP
A ""(relational) atom"" over a "schema" $\Sigma$ takes the form $R(\vect t)$ where $R$ is a ""relation name"" from $\Sigma$ of some arity $k$, and $\vect t \in (\Const \cup \Var)^k$ is a tuple of "terms" (here and later, bold font is used for tuples / vectors).\
\AP
A ""fact"" is an "atom" which contains only "constants".
\AP
A ""database"" $\D$ over a "schema" $\Sigma$ is a finite set of "facts" over $\Sigma$.

\AP
A ""(non-numeric) query"" $q$ of arity $k \geq 0$ over schema $\Sigma$ is a syntactic object associated with a \AP""semantics"" consisting of a function that maps every 
database $\D$ over  $\Sigma$ to a set $q(\D)$ of $k$-tuples of elements from $\const(\D)$,\footnote{%
Note the common notation between the syntax $q$ and the (implicit) "semantics" $\D \mapsto q(\D)$.
}
called the ""answers"" to $q$ on~$\D$. %
\AP
When $k=0$, we say that $q$ is a ""Boolean query"",
and we shall write $\D \models q$ in place of $ \D(q)=\{()\}$. 
As explained in Section \ref{ssec:resp}, it will be sufficient to work with "Boolean queries", 
so \emph{except where otherwise indicated, we will henceforth abuse terminology and write "query" to mean "Boolean query"}.%

\AP
We shall further focus our attention on ""monotone"" queries, that is, queries $q$ such that %
$\D \models q$ implies $\D \cup \D' \models q$ for every pair of databases $\D, \D'$. 
\AP
When $\D \models q$ and $q$ is  "monotone", %
we call $\D$ a ""support"" for $q$ and say it is a ""minimal support"" if $\D$ properly contains no other "support". 
We denote by \AP$\intro*\Minsups q (\D)$ the set of all "minimal supports" for $q$ that are subsets of $\D$.
A "monotone" "query" $q$ is \AP""bounded"" if the sizes of "minimal supports" for $q$ are bounded by a constant (independent of the "database").
We say that a fact $\alpha \in \D$ is \AP ""relevant"" for $q$ in $\D$ if $\alpha \in S$ for some $S \in \Minsups q (\D)$, and \reintro{irrelevant} otherwise.

We now define some concrete classes of "monotone" (Boolean) queries. 
\AP
A ""conjunctive query"" (\reintro{CQ}) over $\Sigma$ takes the form 
$q = \exists \vect{x}\, \varphi$ where $\varphi$ is a conjunction of atoms over $\Sigma$ and $\vect{x}=\vars(\varphi)$. 
We use $\intro*\atoms(q)$ to denote the set of "atoms" of a "CQ"~$q$. 
An \AP""acyclic CQ"" ("ACQ") is a "CQ" $q = \exists \vect{x}\, \varphi$ such that there is a tree whose vertices are $\atoms(q)$ %
and such that for vertices $\alpha,\alpha',\beta$, $\vars(\alpha) \cap \vars(\alpha') \subseteq \vars(\beta)$ whenever $\beta$ lies on the unique simple path between $\alpha,\alpha'$. 
We shall 
denote by  \AP $\intro*\CQ$ (resp. $\intro*\ACQ$) the class of "CQs" (resp. "ACQs").
The class $\intro*\CQneq$ consists of "CQs" extended with \AP""inequality atoms"" of the form $x \neq t$ where $x$ is a "variable" and $t$ is a "term" (\ie a "variable" or a "constant").
Such queries have the obvious semantics, wherein the relation associated to $\neq$ is $\set{(c,c') \in (\const(\D) \cup \const(q))^2 : c \neq c'}$.
Other prominent classes of "monotone" queries include: \AP""unions of conjunctive queries (with inequalities)@unions of conjunctive queries"" ($\intro*\UCQ$, resp.\ $\intro*\UCQneq$), which are finite disjunctions of "CQ"s (resp.\ of "CQneq"s), 
and \AP""path queries"", which take 
the form $\+L(c,d)$ where $\+L$ is a language over a set of binary relation names and $c,d$ are "constants", with the "semantics" of testing whether there is a directed path from $c$ to $d$ whose label is in $\+L$ in the input "database". In particular, a (Boolean) \AP""regular path query"" (\reintro{RPQ}) is a "path query" where $\+L$ is a regular language.

 \AP We recall that the "semantics" of a "CQ"  $q$ is defined by $\D \models q$ if and only if there is a function 
 $h:\mterms(q) \rightarrow \mterms(\D)$ such that
 \begin{enumerate*}[(i)]
\item $R(\vect{t}) \in \atoms(q)$ implies $R(h(\vect{t})) \in \D$, and
\item\label{cas:hom.2} $h(c)=c$ for every $c \in \const(q)$.
\end{enumerate*}
Such a function is called a (query)  ""homomorphism"" (from $q$ to $\D$), and we write 
$q \homto \D$  to indicate the existence of such a "homomorphism". 
We shall also consider homomorphisms of a "CQ" $q$ to another "CQ" $q'$, denoted  $q \homto q'$ and defined similarly.
The notion of \reintro{homomorphism}
between queries
is adapted to "CQneq" queries in the obvious way: restricting the valid mappings $h$ to verify $h(t)\neq h(t')$ if the "inequality atom" $t \neq t'$ is in the query. 
The \AP""core"" of a "CQ" $q$, denoted $\intro*\core(q)$, is the unique (up to isomorphism) minimal equivalent query %
\cite{ChandraM77}.
An \AP""automorphism@@cq"" of a "CQ" (resp.\ "CQneq") $q$ is a "homomorphism" $q \homto q$, and we denote by $\intro*\Auto{q}$ the set of all "automorphisms@@cq" of $q$.
A query $q$ is ""homomorphism-closed"" if $\D' \models q$ whenever $\D \homto \D'$ and $\D \models q$, where $\D \homto \D'$ is defined similarly to "query homomorphisms@homomorphism" but without condition \ref{cas:hom.2}.
The \AP""canonical database"" of a "CQ" is the "database" obtained from the query by turning atoms over a relation name $R$ into tuples of the relation $R$ replacing each variable with a fresh constant.

\AP We use the standard notion of ""polynomial-time Turing reductions"" between numeric problems $\Psi_1,\Psi_2$, and we write $\Psi_1 \intro*\polyrx \Psi_2$ if there is a polynomial-time algorithm for $\Psi_1$ using a $\Psi_2$ oracle.
If the algorithm makes only a single call to the $\Psi_2$ oracle we call it a ""$1$-Turing reduction"".

\section{Responsibility Measures and the Drastic Shapley Value} %
\label{sec:drastic-shapley}
The (drastic) Shapley value has been widely studied as a way of defining a "responsibility measure" for database queries. A seminal and oft-cited result \cite{shapley1953value} in cooperative game theory establishes that the Shapley value is the unique wealth distribution measure that satisfies a specific set of `Shapley axioms'. It is tempting to think that this unicity result transfers to the setting of database "responsibility measures", and thus that the drastic Shapley value considered in prior work is the only reasonable approach. 
However, we consider that there is no single valid notion of a "responsibility measure", and we shall show in this section that the Shapley axioms are insufficient to argue the contrary. 
We start by introducing %
the notion of "responsibility measure" for "non-numeric" database queries, before recalling the definition of the Shapley value from game theory and how it has been utilized to define a concrete such "responsibility measure". 
Finally, we discuss the relevance of the Shapley axioms in the database setting and motivate the interest of exploring alternative measures.%

\subsection{Responsibility Measures}\label{ssec:resp}
In line %
with prior studies \cite{DBLP:journals/pvldb/MeliouGMS11,livshitsShapleyValueTuples2021,KaraOlteanuSuciuShapleyBack,ourpods24}, we %
 work with \AP""partitioned databases"", in which %
the "database" $\D$ is partitioned into \AP""endogenous"" and \AP""exogenous"" "facts", denoted by $\intro*\Dn$ and $\intro*\Dx$ respectively.
By definition, the "endogenous" facts share the whole responsibility while the "exogenous" ones are treated as always present.
The aim is to quantify the responsibility of a given "endogenous" fact to the obtention of some query answer. 
\AP
Formally, we shall be interested in ""responsibility measures"", defined as functions $\phi$ which take as input %
a (possibly non-Boolean) "query" $q$, a "partitioned database" $\D$, an answer $\vect a$ to $q$ in $\D$, and which output
 a quantitative score measuring, for each  "endogenous" fact $\alpha\in \Dn$, how much it contributes to $\vect{a}$ being an answer of $q(\D)$.
We simplify the presentation by replacing 
the input answer $\vect{a}$ and (possibly non-Boolean) query $q$
by  
the associated Boolean query $q(\vect{a})$, defined by letting $\D \models q(\vect{a})$ iff $\vect{a} \in q(\D)$. %
In this manner, %
we can eliminate the answer tuple from the arguments of the "responsibility measure"  %
and work instead with "responsibility measures" of the form $\phi_{q,\D}(\alpha)$ with Boolean~$q$. 
Henceforth, we shall thus assume w.l.o.g.\ %
that the input query %
is always "Boolean", as well as %
"monotone".\footnote{See \Cref{sec:conclusions} for a brief discussion on non-monotone queries.} %
We discuss in \Cref{sec:desirableaxioms,sec:originalShapleyAxioms}
the desirable properties of "responsibility measures".

\subsection{Shapley Value and Shapley-Based Responsibility Measures}
\AP%
A ""cooperative game"" is traditionally given by a finite set of players $\intro*\bse{\scorefun}$ and a "wealth function" $\scorefun : \partsof{\bse{\scorefun}} \to \lQ$ such that $\scorefun(\emptyset)=0$, where $\intro*\partsof{\bse{\scorefun}}$ denotes the set of all subsets of $\bse{\scorefun}$. Conceptually, $\scorefun(S)$ represents the wealth attributable to the subset $S\inc \bse{\scorefun}$ of players  (often called `coalition').
The "Shapley value", introduced in the seminal work by Shapley \cite{shapley1953value}, is a method for quantifying each player's contribution within a game. 
Formally, this corresponds to taking as input a "cooperative game" and producing a function  $\Sh_{\scorefun} : \bse{\scorefun} \to \lQ$ 
that attributes to each element of $\bse{\scorefun}$ a share of the total wealth of the game.

Note that the set of players $\AP\bse{\scorefun}$ is merely the domain of the "wealth function" $\scorefun$, and hence $\scorefun$ completely specifies the game.
In fact, there is nothing special about players, and $\bse{\scorefun}$ can be any finite set.
We then simplify the jargon by defining a \AP""wealth function"" to be a function $\intro*\scorefun : \partsof{\bse{\scorefun}} \to \lQ$ over a finite \AP""base set"" $\bse{\scorefun}$ such that $\scorefun(\emptyset)=0$. 
The "Shapley value" is hence a function $\psi$ that takes as input such a "wealth function" $\scorefun$  and outputs a numerical function $\psi_{\scorefun} : \bse{\scorefun} \to \lQ$.
Of course, there could be many possible such $\psi_{\scorefun}$ functions unless some restrictions are imposed. The following axioms, collectively equivalent to those proposed by Shapley  and deemed desirable in many applications, characterize a unique $\psi_{\scorefun}$ function, which shall be henceforth denoted by $\Sh_{\scorefun}$:
\begin{enumerate}[leftmargin=\widthof{(wSym)}+\labelsep]
   \item[{\crtcrossreflabel{(wSym)}[Sh:1]}] \textit{Weak Symmetry:}
      if
      $\scorefun(S\cup\{\alpha\}) = \scorefun(S\cup\{\beta\})$ for all $S\subseteq \bse{\scorefun}\setminus\{\alpha,\beta\}$, then $\psi_{\scorefun}(\alpha) = \psi_{\scorefun}(\beta)$.
   \item[{\crtcrossreflabel{(Null)}[Sh:2]}] \AP\textit{Null element:} any element of the "base set" that does not contribute to increasing the wealth (\ie a so-called ""null@@player"" element $\alpha$ such that $\scorefun(S \cup \set \alpha)= \scorefun(S)$ for all $S$) must obtain 0 as contribution.
   \item[{\crtcrossreflabel{(Lin)}[Sh:4]}] \textit{Linearity:} the value of the sum $\scorefun_1 {+} \scorefun_2$ of two "wealth functions" $\scorefun_1$ and $\scorefun_2$ over the same "base set" (where $(\scorefun_1 {+} \scorefun_2)(S) \defeq  \scorefun_1(S) + \scorefun_2(S)$ for every $S$) is the sum of values over the separate wealth functions: $\psi_{\scorefun_1+\scorefun_2}(\alpha)=\psi_{\scorefun_1}(\alpha)+\psi_{\scorefun_2}(\alpha)$.
   \item[{\crtcrossreflabel{(Eff)}[Sh:3]}] \textit{Efficiency:} the sum
      $\sum_{\alpha\in \bse{\scorefun}} \psi_{\scorefun}(\alpha)$
      of all contributions equals the total wealth $\scorefun(\bse{\scorefun})$ of the "base set".
\end{enumerate}
Thus, the ""Shapley value"" is defined to take a "wealth function" $\scorefun$ as input and output the only function $\Sh_{\scorefun}$ satisfying the axioms above. 

\begin{remark}\label{rk:different-axioms}
We shall provide our own proof of unicity (\Cref{th:shapunique}), even though it merely consists in a slight adaptation of Shapley's original proof, because the formalism of "cooperative games" and the statement of the axioms themselves has evolved over the years. As we shall see in \Cref{sec:originalShapleyAxioms}, these evolutions, although subtle, mean that \Cref{th:shapunique} does not, strictly speaking, follow from \cite[main Theorem]{shapley1953value}, and to the best of our knowledge, there is no published proof of the unicity of the above variant set of axioms. 
\end{remark}
\color{black}

The "Shapley value" of $\scorefun$ applied to $\alpha \in \bse{\scorefun}$ turns out to be the average wealth contribution of $\alpha$ over all the total linear orderings $\intro*\Totord(\bse{\scorefun})$  of $\bse{\scorefun}$, and can then be expressed via the following closed-form formula:
\begin{equation}\label{formul:sh1}
	\intro*\Sh_{\scorefun}(\alpha) = \frac{\sum_{\sigma\in \Totord(\bse{\scorefun})}  \left(\scorefun(\sigmaleq{\alpha}) - \scorefun(\sigmaless{\alpha})\right)}{|\Totord(\bse{\scorefun})|}
\end{equation}
where 
$\intro*\sigmaless{\alpha}$ (resp.~$\intro*\sigmaleq{\alpha}$) denotes the set of elements smaller than $\alpha$ (resp.\ smaller than or equal to $\alpha$) by the order $\sigma$. 
\Cref{formul:sh1} is often presented as the expected earnings of the player $\alpha$ in a probabilistic experiment where the players arrive in a random order and each one earns the additional wealth it created upon arrival.
It should however be noted that, although this probabilistic experiment is very convenient to remember the formula, it is not its raison d’être: as explained above, \Cref{formul:sh1} has been chosen because it is the formula that satisfies the axioms \ref{Sh:1}--\ref{Sh:3}.

The "Shapley value" admits a second classic closed-form formula, which is obtained by grouping together the contributions of the orderings $\sigma$ that have the same $S \defeq \sigmaless{\alpha}$:

\begin{equation}\label{formul:sh}
   \reintro*\Sh_{\scorefun}(\alpha) = \sum_{S\subseteq \bse{\scorefun} \setminus \set \alpha } \!\!\!  \frac{|S|!(|\bse{\scorefun}| - |S| -1)!}{|\bse{\scorefun}|!}\left(\scorefun(S \cup \{\alpha\}) - \scorefun(S)\right)
\end{equation}

To employ the "Shapley value" in the database setting and obtain a "responsibility measure", 
one must specify how to build games from %
queries and databases.
\AP Concretely, we need a ""wealth function family"" $\intro*\STscorefun$ which takes as parameters a query $q$ and 
a "partitioned database" $\D$ and defines a "wealth function"~$\intro*\stscorefun_{q,\D}$ over the base set $\bse{\stscorefun_{q,\D}}\defeq \Dn$ of "endogenous facts". For instance  the standard "wealth function family" in the literature \cite{livshitsShapleyValueTuples2021,ourpods24} is what we call the \AP ""drastic wealth function family"",\footnote{The name \emph{drastic} is borrowed from the closely related \emph{drastic inconsistency measure} (see, \eg, measure `$I_d$' in \cite{hunterMeasureConflictsShapley2010,Thimm17, LivshitsK22}) assigning 0 or 1 depending on whether the "database" is consistent \wrt a property. See also \Cref{ssec:inconsistency}.} denoted by $\intro*\Dscorefun$ and defined as follows:
$\intro*\dscorefun_{q,\D}$ assigns 1 to every set $S\inc \Dn$ of "endogenous facts" such that $S \cup \Dx \models q$ and 0 to the others, unless $\Dx \models q$ in which case it assigns 0 to all sets.

Once a "wealth function family" has been defined, it can be combined with the "Shapley value" to obtain a "responsibility measure", whose output is $\Sh_{\stscorefun_{q,\D}}(\alpha)$. 
In particular, the "Shapley value" instantiated with the "drastic wealth function family" $\Dscorefun$ is what we call the \AP""drastic Shapley value"".\footnote{In previous works, such as \cite{livshitsShapleyValueTuples2021}, this value has been referred to simply as the `Shapley value'.} %
Given a class of queries $\+C$, we denote by $\intro*\stShapley_{\C}$ the associated computational problem %
of computing, for a given "query" $q\in\+C$, "partitioned database" $\D$ and "fact" $\alpha \in \Dn$, the value $\Sh_{\stscorefun_{q,\D}}(\alpha)$. 
When $\C$ consists of a single query $q$, we simply write $\reintro*\stShapley_{q}$. Note the superscript $\star$ that is consistent across all notations. 
Prior work on Shapley value computation for Boolean queries considered $\star=\textsf{dr}$, that is, \AP$\intro*\dShapley_{q}$ and $\intro*\dShapley_{\C}$.\medskip

\AP Besides the "Shapley value", another wealth distribution scheme, the \reintro{Banzhaf power index}, has been applied to responsibility distribution for database queries \cite{abramovichBanzhafValuesFacts2024,salimiQuantifyingCausalEffects2016}. %
This alternative scheme is closely related to the "Shapley value", to the point where \cite{karmakarExpectedShapleyLikeScores2024} introduced the notion of ""Shapley-like"" scores to collectively designate a class of wealth distribution schemes that contains them both. Formally, 
"Shapley-like" scores are defined by the formula:
\begin{equation}\label{formul:slike}
  \+S^c_{\scorefun}(\alpha) \defeq \sum_{S\subseteq \bse{\scorefun} \setminus \set \alpha } c(|S|,|\bse{\scorefun}|)(\scorefun(S\cup\{\alpha\}) - \scorefun(S))
\end{equation}
where $c(|S|,|\bse{\scorefun}|)$ is a \AP""coefficient function"" characterising the measure. In particular, the "Shapley value" is obtained by setting $c(|S|,|\bse{\scorefun}|)=\frac{|S|!(|\bse{\scorefun}|-|S|-1)!}{|\bse{\scorefun}|!}$ and the ""Banzhaf power index"" by $c(|S|,|\bse{\scorefun}|)=1$.
Exactly like with the "Shapley value", the "wealth function family" $\Dscorefun$ can be applied to any "Shapley-like score" to define a "responsibility measure" called a \AP""drastic Shapley-like score""; the most popular of which (besides the "drastic Shapley value") being the \AP""drastic Banzhaf value"" \cite{abramovichBanzhafValuesFacts2024} \aka\ the \emph{Causal Effect} \cite{salimiQuantifyingCausalEffects2016}, since the two coincide \cite[Proposition~5.5]{livshitsShapleyValueTuples2021}.

\subsection{Shapley Axioms}\label{sec:originalShapleyAxioms}
We now
take a closer look at the Shapley axioms: we shall first acknowledge the differences between our modern statements and the original axioms in \cite{shapley1953value},
and then examine what these axioms mean for "responsibility measures" for database queries. 
\subsubsection*{Original Shapley axioms}
The importance of the Shapley value naturally stems from the satisfaction of Axioms \ref{Sh:1}-\ref{Sh:3}, deemed desirable in many scenarios for fair distribution of wealth or responsibility in cooperative games. Before going further we need to point out that Shapley’s seminal paper \cite{shapley1953value} uses
a different set of axioms.
The most obvious difference is that only 3 axioms were presented, because \ref{Sh:2} and \ref{Sh:3} were merged together. This is entirely inoccuous because the pair \ref{Sh:2},\ref{Sh:3} is precisely equivalent to \cite[Axiom 2]{shapley1953value}.
While the axiom \ref{Sh:1} we use here is in line with some modern presentations of Shapley's axioms (e.g.\ \cite[Axiom 2]{strumbeljEfficientExplanationIndividual2010a}, \cite[Chapter 17]{molnar2025}), in fact the following symmetry axiom \cite[Axiom 1]{shapley1953value} was originally proposed instead of \ref{Sh:1}.
\begin{enumerate}
   \item[{\crtcrossreflabel{(Sym)}[Sh:1s]}] \textit{Symmetry:}
      for every permutation $\pi\in\Perm(P)$ of the set of players we have that $\psi_{\scorefun\circ\pi}(\alpha) = \psi_{\scorefun}(\pi(\alpha))$ where $\scorefun\circ\pi$ is the "wealth function" defined by $S\mapsto \scorefun(\{\pi(p) : p\in S\})$.
\end{enumerate}
Although \ref{Sh:1} is a trivial consequence of \ref{Sh:1s} (by taking the permutation $\pi$ that only swaps $\alpha$ and $\beta$), the two properties are not equivalent: 
consider the 0/1 game $\xi$ on $P\defeq \{p_1\dots p_4\}$ whose "minimal supports" are $\{p_1,p_2\}$ and $\{p_3,p_4\}$; then the latter axiom implies $\psi_{\scorefun}(p_1) = \psi_{\scorefun}(p_3)$ because of the permutation $\pi$ that swaps $p_1$ with $p_3$ and $p_2$ with $p_4$, but the former doesn’t imply anything since $1=\xi(\{p_2\}\cup\{p_1\}) \neq \xi(\{p_2\}\cup\{p_3\})=0$.
As a consequence, the statement of \cite[main Theorem]{shapley1953value} cannot be used to justify that the axioms \ref{Sh:1}--\ref{Sh:3} uniquely characterize the Shapley value.
However, the proof of \cite[main Theorem]{shapley1953value} can be easily adapted to use \ref{Sh:1} instead of \ref{Sh:1s}, and therefore
\ref{Sh:1}--\ref{Sh:3} as a set of axioms is equivalent to \cite[Axioms 1--3]{shapley1953value}:
\begin{theorem}\label{th:shapunique}%
The "Shapley value" is the unique function that satisfies \ref{Sh:1}, \ref{Sh:2}, \ref{Sh:4} and \ref{Sh:3}.
\end{theorem}
\begin{proof}
The equivalence between the pair \ref{Sh:2}, \ref{Sh:3} and \cite[Axiom 2]{shapley1953value} is fairly trivial. We thus only need to justify the weakening of \ref{Sh:1s} into \ref{Sh:1}.
It can be seen that in the proof of the unicity of the Shapley value (main theorem of \cite{shapley1953value}), the axiom \ref{Sh:1s} is only used to show the following lemma:

\begin{lemma}[{\cite[Lemma 2]{shapley1953value}}]
   Let $\psi$ be a function that satisfies \ref{Sh:1s}, \ref{Sh:2}, and \ref{Sh:3}. Given a set $M\subseteq X$, define the wealth function $\xi_{M}$ of "base set" $X$ by $\xi_M(S)=1$ if $M\inc S$ and 0 otherwise. Then, for every $\alpha\in X$, $\psi_{\xi_M}(\alpha) = \frac{1}{|S|}$ if $\alpha\in S$ and 0 otherwise.\footnote{The statement makes use of the simple aforementioned observation that \ref{Sh:2}+\ref{Sh:3} is equivalent to \cite[Axiom 2]{shapley1953value}.}
\end{lemma}
It is easy to see that the above lemma still holds when \ref{Sh:1s} is replaced by \ref{Sh:1}: the axiom’s sole function is to show that all elements in $S$ must have the same value, and the weaker \ref{Sh:1} is sufficient for this specific case.
\end{proof}

\color{green}%
\subsubsection*{Axioms applied to responsibility measures}
\color{black}%
Now we have shown that the axioms \ref{Sh:1}-\ref{Sh:3}, although different from their original counterparts, are fit to the task, we are faced with a greater issue: these axioms are all are phrased in terms of cooperative games,  
whereas we are interested in properties of "responsibility measures".
As such measures take queries and databases as arguments, not numeric functions over a "base set", the desirable properties we are after should likewise refer to databases and queries.  %

Of course, the axioms \ref{Sh:1}--\ref{Sh:3} are reasonable inspirations for defining such desirable properties, but their interpretation in the database setting crucially depends on the way we model the queries as cooperative games, \ie the "wealth function family" we use.
Given that all prior work has focused on the drastic wealth function, $\Dscorefun$, it is instructive to instantiate the above axioms with $\Dscorefun$ 
in order to understand which properties of "responsibility measures" have actually been invoked to justify the "drastic Shapley value" for "non-numeric queries". We first consider \ref{Sh:1}.

\begin{enumerate}[align=left]
   \item[{\crtcrossreflabel{(wSym-dr)}[Shdr:1]}]
   If for $\alpha,\beta \in \Dn$ we have $S\cup\{\alpha\}\models q \LRa S\cup\{\beta\}\models q$ for all $S\subseteq \D\setminus\{\alpha,\beta\}$, then $\phi_{q,\D}(\alpha) = \phi_{q,\D}(\beta)$.%
\end{enumerate}
We consider \ref{Shdr:1} to be a very desirable property.
In fact, we would suggest to strengthen it as follows to ensure syntax independence not only for the input facts but also for the queries: 

\begin{enumerate}[align=left]
   \item[{\crtcrossreflabel{(wSym-db)}[Shdb:1]}]
   If $q,q'$ have the same "semantics" (\ie they define the same Boolean function) and for $\alpha,\beta \in \Dn$ we have $S\cup\{\alpha\}\models q \LRa S\cup\{\beta\}\models q'$ for all $S\subseteq \D\setminus\{\alpha,\beta\}$, then $\phi_{q,\D}(\alpha) = \phi_{q',\D}(\beta)$.
\end{enumerate}
\begin{remark}
Instead of the above, one could naturally define corresponding properties (Sym-dr) and (Sym-db) based on \ref{Sh:1s} instead of \ref{Sh:1}. These would indeed be desirable, and one might argue in their favour that they are slightly stronger. However, they are quite tedious to state, and the difference is minimal: given that any "Shapley-like" score satisfies \ref{Sh:1s} as well, every time we show a measure satisfies (wSym-dr) or (wSym-db), the same holds for the stronger counterpart.%
\end{remark}
\AP Turning now to \ref{Sh:2}, we observe that "null players" translate into "irrelevant" facts, yielding:

\begin{enumerate}[align=left]
   \item[{\crtcrossreflabel{(Null-dr)}[Shdr:2]}]\AP
   If a "fact" $\alpha\in\Dn$ is "irrelevant", then $\phi_{q,\D}(\alpha) = 0$.
\end{enumerate}
Again, we find \ref{Shdr:2} desirable but insufficient, as we wish to be able to decide "relevance":

\begin{enumerate}[align=left]
   \item[{\crtcrossreflabel{(Null-db)}[Shdb:2]}] If a "fact" $\alpha\in\Dn$ is "irrelevant" then $\phi_{q,\D}(\alpha) = 0$; otherwise $\phi_{q,\D}(\alpha) > 0$.\footnote{Recall that the query $q$ we consider is always "Boolean" and "monotone".}
\end{enumerate}
Moving on to \ref{Sh:4}, we come to our first issue. This axiom should intuitively translate as ``the sum of the responsibility for two queries is the responsibility for the sum query'', 
which makes sense when dealing with  "numeric" queries (for which the sum query is well defined), but whose meaning is unclear for "non-numeric" queries. When instantiated with $\Dscorefun$, the property
is nonsensical because for all non-trivial cases $\dscorefun_{q_1,\Dx}(\Dn)+\dscorefun_{q_2,\Dx}(\Dn) = 2$, 
which cannot coincide with $\dscorefun_{q,\Dx}(\Dn)$ (whose value is 0 or 1), 
no matter which Boolean $q$ we use to represent the `sum query'.
\begin{remark}\label{rem:transfer}
   For \AP""Boolean games"" (\aka \reintro{simple games}) where the "wealth function" returns true or false, the following alternative to the linearity axiom has been introduced in \cite[Axiom S3']{dubeyUniquenessShapleyValue1975}, and later used in more modern axiomatizations of the Shapley value restricted to monotone "Boolean games", also called Shapley-Shubik value \cite[§2]{laruelleShapleyShubikBanzhafIndices2001}:
   \begin{enumerate}[align=left]
      \item[(Tr)] \emph{Transfer:} Given two "wealth functions" $\scorefun_1$ and $\scorefun_2$ with values in $\lB$ over the same "base set",
	 \[\psi_{\scorefun_1\lor\scorefun_2}(\alpha) + \psi_{\scorefun_1\land\scorefun_2}(\alpha)=\psi_{\scorefun_1}(\alpha)+\psi_{\scorefun_2}(\alpha)\]
   \end{enumerate}
   where $\lor$ and $\land$ are the Boolean disjunction and conjunction respectively. Since this axiom is by design restricted to Boolean wealth functions, it can be instantiated with the $\Dscorefun$ family:
   \begin{enumerate}[align=left]
      \item[(Tr-dr)] \emph{Transfer:} Given two "Boolean queries" $q_1$ and $q_2$ with values in $\lB$,
	 \[\phi_{q_1\lor q_2,\D}(\alpha) + \phi_{q_1\land q_2,\D}(\alpha) = \phi_{q_1,\D}(\alpha) + \phi_{q_2,\D}(\alpha)\]
   \end{enumerate}

   Although (Tr) is interesting since it allows one to show the unicity of the "Shapley value" restricted to "Boolean games", we do not believe (Tr-dr) to be an obviously desirable feature for a "responsibility measure". Even if one were to believe so, it should be noted that all "WSMSs" (introduced in \Cref{ssec:wsmsdef}) satisfy it, by virtue of the fact that they are "Shapley values" in disguise (see \Cref{prop:wsmsissh}).
\end{remark}

Let us finally consider the meaning of \ref{Sh:3} once instantiated with the drastic wealth function:
\begin{enumerate}[align=left]
   \item[{\crtcrossreflabel{(Eff-dr)}[Shdr:3]}] %
   The sum $\sum_{\alpha\in\Dn} \phi_{q,\D}(\alpha) = 1$ if $\D\models q$ and $\Dx\not\models q$, and $0$ otherwise.
\end{enumerate}
While it seems reasonable enough to ask that the sum of contributions of the facts is equal to the value attained by the full set of endogenous facts, 
it is rather more questionable why we should always arrive at the same total value of 1 when $\D\models q$ and $\Dx\not\models q$,  irrespective of the particular $q$ and $\D$.
The fact that queries are modelled using constant-sum games means that the drastic Shapely value is highly sensitive to database changes: the addition of \emph{any} new "relevant" fact will affect the values of \emph{all} "relevant" facts, as they need to share the same total wealth (of 1). This, in turn, %
can make it challenging 
to interpret the drastic Shapley value of facts without appealing to the probabilistic process that underlies \Cref{formul:sh1}. %
Note that this differs markedly from the case of numeric queries for which the total value -- defined via the query result -- varies according to the input database and finds its meaning in the query itself. \medskip

In view of the identified issues with the translations of \ref{Sh:4} and \ref{Sh:3} into the database setting, 
the Shapley axioms cannot be used to argue that the drastic Shapley value is the only (or the best) approach to defining a reasonable "responsibility measure". 
Indeed, only \ref{Shdb:1} and \ref{Shdb:2} are truly natural, and they constitute rather weak minimal requirements.

This raises the obvious questions: \emph{What are desirable properties for %
responsibility measures of non-numeric queries?
And how can they be realized?} 
In \Cref{sec:desirableaxioms} we propose a set of minimal desirable properties, and in \Cref{sec:alt-measures} we introduce a novel family of measures conforming to the identified properties.

\section{Desirable Properties}
\label{sec:desirableaxioms}

\noindent
As discussed in Section \ref{sec:drastic-shapley},
the Shapley axioms 
provide insufficient
justification for arguing that the "drastic Shapley value" is the \emph{only reasonable} or the \emph{best} measure for "non-numeric queries".
We believe, however, that an axiomatic approach remains valuable in the sense that we ought to think about what are the truly desirable properties of "responsibility measures".
This section shall address this precise question, starting in \Cref{ssec:ax1} with an slightly informal but more intuitive discussion before getting into the formal details in \Cref{ssec:ax2}.

\subsection{Discussion on Desired Properties of Responsibility Measures}\label{ssec:ax1}

We already have from \Cref{sec:drastic-shapley} the properties \ref{Shdb:1} and \ref{Shdb:2} which capture some basic requirements. %
To identify further desirable properties, we shall go back to \cite{livshitsShapleyValueTuples2021} where the "drastic Shapley" value was first introduced. There, a previously defined quantitative measure (the causal responsibility \cite{DBLP:journals/pvldb/MeliouGMS11}) was argued to yield counterintuitive results on a specific example \cite[Example 5.2]{livshitsShapleyValueTuples2021}, a slight adaptation of which is presented in \Cref{fig:counter-ex}%
-(a). Indeed, the causal responsibility assigns the same score to all edges in this example, even though $e_0$ intuitively has more responsibility as it satisfies the query on its own, whereas $e_1$ and $e_2$ must be used together. %
More generally, this idea can be captured by %
the following axiom
\AP (""AOTBE"" stands for ``All Other Things Being Equal''):
\begin{itemize}
   \item[{\crtcrossreflabel{(MS1)}[MS:1]}] "AOTBE", a fact that appears in smaller "minimal supports" should have higher responsibility.
\end{itemize}

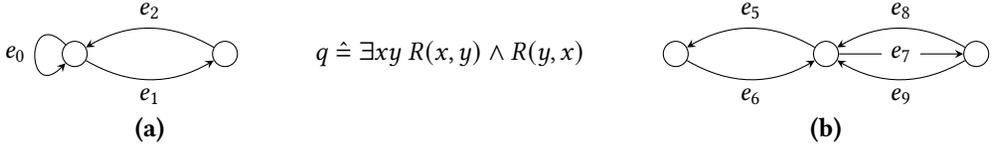
\begin{figure}[tb]
\centering
\scalebox{0.9}{%
      \begin{tikzpicture}
\coordinate (00) at (-5, 0);
\coordinate (01) at (-3, 0);
\coordinate (02) at (0, 0);
\coordinate (03) at (-4, -1);
\coordinate (04) at (5, -1);
\coordinate (05) at (5, 0);
\coordinate (06) at (3, 0);
\coordinate (07) at (7, 0);
\begin{pgfonlayer}{nodelayer}
\node [draw, circle] (0) at (00) {};
\node [draw, circle] (1) at (01) {};
\node [] (2) at (02) {$q\defeq \exists xy ~ R(x,y) \land R(y,x)$};
\node [] (3) at (03) {\textbf{(a)}};
\node [] (4) at (04) {\textbf{(b)}};
\node [draw, circle] (5) at (05) {};
\node [draw, circle] (6) at (06) {};
\node [draw, circle] (7) at (07) {};
\end{pgfonlayer}
\begin{pgfonlayer}{edgelayer}
\draw [->,>=stealth, in=-135, out=135, loop] (0) to node[midway,left] {$e_0$} ();
\draw [->,>=stealth, bend right] (0) to node[midway,below] {$e_1$} (1);
\draw [->,>=stealth, bend right] (1) to node[midway,above] {$e_2$} (0);
\draw [->,>=stealth, bend right] (5) to node[midway,above] {$e_5$} (6);
\draw [->,>=stealth, bend right] (6) to node[midway,below] {$e_6$} (5);
\draw [->,>=stealth] (5) to node[midway,fill=white] {$e_7$} (7);
\draw [->,>=stealth, bend right] (7) to node[midway,above] {$e_8$} (5);
\draw [->,>=stealth, bend left] (7) to node[midway,below] {$e_9$} (5);
\end{pgfonlayer}
\end{tikzpicture}%
   }
\caption{Pathological examples for some "responsibility measures". The databases are represented as graphs with nodes being constants and edges being $R$-tuples.
}
\label{fig:counter-ex}
\end{figure}

In a similar vein, it turns out that the causal responsibility measure also gives the same score to $e_5$ and $e_7$ in the example represented in \Cref{fig:counter-ex}%
-(b). Again, this is counterintuitive because the fact $e_5$ only appears in a single "minimal support" ($\{e_5,e_6\}$) while $e_7$ appears in two ($\{e_7,e_8\}$ and $\{e_7,e_9\}$). 
This suggests the following more general property: 
\begin{itemize}
   \item[{\crtcrossreflabel{(MS2)}[MS:2]}] "AOTBE", a fact that appears in more "minimal supports" should have higher responsibility.
\end{itemize}

The axioms %
\ref{MS:1} and \ref{MS:2} have been stated informally. %
While mathematically precise formulations %
can be achieved
and shall be provided in \Cref{ssec:ax2}%
, they turn out to be quite technical as the notion of
``"AOTBE"'' is non-trivial to formalize. 
For instance, the most na\"{i}ve definitions will fail to predict the behavior of the "drastic Shapley value" on the following example.

\begin{figure}[tb]
\centering
      \begin{tikzpicture}[scale=1.5]
	\coordinate (00) at (-2.5, 0);
	\coordinate (01) at (-1.75, 0.5);
	\coordinate (02) at (-0.75, 0.5);
	\coordinate (03) at (-1.75, -0.5);
	\coordinate (04) at (-0.75, -0.5);
	\coordinate (05) at (3.25, 0);
	\coordinate (06) at (2.75, 0.5);
	\coordinate (07) at (1.75, 0.5);
	\coordinate (08) at (2.75, -0.5);
	\coordinate (09) at (1.75, -0.5);
	\coordinate (010) at (0.75, -0.5);
	\coordinate (011) at (2.75, -0.25);
	\coordinate (012) at (2.75, -0.75);
	\coordinate (013) at (3.5, 0.25);
	\coordinate (014) at (3, -0.25);
	\coordinate (015) at (0.75, -0.25);
	\coordinate (016) at (0.75, -0.75);
	\coordinate (017) at (2.75, 0.25);
	\coordinate (018) at (2.75, 0.75);
	\coordinate (019) at (1.75, 0.75);
	\coordinate (020) at (1.75, 0.25);
	\coordinate (021) at (3, 0.25);
	\coordinate (022) at (3.5, -0.25);
	\coordinate (023) at (-2.75, -0.25);
	\coordinate (024) at (-2, -0.75);
	\coordinate (025) at (-0.5, -0.75);
	\coordinate (026) at (-0.5, -0.25);
	\coordinate (027) at (-2, -0.25);
	\coordinate (028) at (-2.25, 0.25);
	\coordinate (029) at (-2.25, -0.25);
	\coordinate (030) at (-2, 0.25);
	\coordinate (031) at (-0.5, 0.25);
	\coordinate (032) at (-0.5, 0.75);
	\coordinate (033) at (-2, 0.75);
	\coordinate (034) at (-2.75, 0.25);
	\coordinate (035) at (-2, 0.75);
	\coordinate (036) at (-2, -0.75);
	\coordinate (037) at (-1.5, -0.75);
	\coordinate (038) at (-1.5, 0.75);
	\coordinate (039) at (-1, 0.75);
	\coordinate (040) at (-1, -0.75);
	\coordinate (041) at (-0.5, -0.75);
	\coordinate (042) at (-0.5, 0.75);
	\coordinate (043) at (4, 0.5);
	\coordinate (044) at (4, -0.5);
	\coordinate (045) at (-1.75, 0);
	\coordinate (046) at (-0.75, 0);
	\begin{pgfonlayer}{nodelayer}
		\node [draw, circle, minimum height=1.5em] (0) at (00) {};
		\node at (00) {$n_0$};
		\node [draw, circle, minimum height=1.5em] (1) at (01) {};
		\node at (01) {$n_1$};
		\node [draw, circle, minimum height=1.5em] (2) at (02) {};
		\node at (02) {$n_2$};
		\node [draw, circle, minimum height=1.5em] (3) at (03) {};
		\node at (03) {$n_3$};
		\node [draw, circle, minimum height=1.5em] (4) at (04) {};
		\node at (04) {$n_4$};
		\node [draw, circle, minimum height=1.5em] (5) at (05) {};
		\node at (05) {$n_5$};
		\node [draw, circle, minimum height=1.5em] (6) at (06) {};
		\node at (06) {$n_6$};
		\node [draw, circle, minimum height=1.5em] (7) at (07) {};
		\node at (07) {$n_7$};
		\node [draw, circle, minimum height=1.5em] (8) at (08) {};
		\node at (08) {$n_8$};
		\node [draw, circle, minimum height=1.5em] (9) at (09) {};
		\node at (09) {$n_9$};
		\node [draw, circle, minimum height=1.5em] (10) at (010) {};
		\node at (010) {$n_{\mathrm{x}}$};
		\node [] (11) at ($(011)+(-0.1,0)$) {};
		\node [] (12) at ($(012)+(0.1,0)$) {};
		\node [] (13) at ($(013)+(-0.0732233,-0.0732233)+(0.1,-0.1)$) {};
		\node [] (14) at ($(014)+(0.0732233,+0.0732233)+(0.1,-0.1)$) {};
		\node [] (15) at (015) {};
		\node [] (16) at (016) {};
		\node [] (17) at ($(017)+(-0.1,0)$) {};
		\node [] (18) at ($(018)+(0.1,0)$) {};
		\node [] (19) at (019) {};
		\node [] (20) at (020) {};
		\node [] (21) at ($(021)+(0.0732233,-0.0732233)+(0.1,0.1)$) {};
		\node [] (22) at ($(022)+(-0.0732233,0.0732233)+(0.1,0.1)$) {};
		\node [] (23) at ($(023)+(0.0732233,0.0732233)+(-0.1,0.1)$) {};
		\node [] (24) at ($(024)+(-0.1,0)$) {};
		\node [] (25) at (025) {};
		\node [] (26) at (026) {};
		\node [] (27) at ($(027)+(0.1,0)$) {};
		\node [] (28) at ($(028)+(-0.0732233,-0.0732233)+(-0.1,0.1)$) {};
		\node [] (29) at ($(029)+(-0.0732233,0.0732233)+(-0.1,-0.1)$) {};
		\node [] (30) at ($(030)+(0.1,0)$) {};
		\node [] (31) at (031) {};
		\node [] (32) at (032) {};
		\node [] (33) at ($(033)+(-0.1,0)$) {};
		\node [] (34) at ($(034)+(0.0732233,-0.0732233)+(-0.1,-0.1)$) {};
		\node [] (35) at (035) {};
		\node [] (36) at (036) {};
		\node [] (37) at (037) {};
		\node [] (38) at (038) {};
		\node [] (39) at (039) {};
		\node [] (40) at (040) {};
		\node [] (41) at (041) {};
		\node [] (42) at (042) {};
		\node [draw, circle, minimum height=1.5em, color=gray] (45) at (045) {};
		\node [color=gray] at (045) {$x_1$};
		\node [draw, circle, minimum height=1.5em, color=gray] (46) at (046) {};
		\node [color=gray] at (046) {$x_2$};
	\end{pgfonlayer}
	\begin{pgfonlayer}{edgelayer}
		\draw [color=blue,rounded corners] (14.center) to (17.center) to (20.center);
		\draw [color=blue] (20.center) arc (-90:-270:0.25);
		\draw [color=blue,rounded corners] (19.center) to (18.center) to (13.center);
		\draw [color=blue] (13.center) arc (45:-135:0.25);
		\draw [color=blue,rounded corners] (22.center) to (12.center) to (16.center);
		\draw [color=blue] (16.center) arc (-90:-270:0.25);
		\draw [color=blue,rounded corners] (15.center) to (11.center) to (21.center);
		\draw [color=blue] (21.center) arc (135:-45:0.25);
		\draw [color=red,rounded corners] (23.center) to (24.center) to (25.center);
		\draw [color=red] (25.center) arc (-90:90:0.25);
		\draw [color=red,rounded corners] (26.center) to (27.center) to (28.center);
		\draw [color=red] (28.center) arc (45:225:0.25);
		\draw [color=red,rounded corners] (29.center) to (30.center) to (31.center);
		\draw [color=red] (31.center) arc (-90:90:0.25);
		\draw [color=red,rounded corners] (32.center) to (33.center) to (34.center);
		\draw [color=red] (34.center) arc (135:315:0.25);
		\draw [color=green] (35.center) to (36.center) arc (-180:0:0.25)
										to (38.center) arc (0:180:0.25);
		\draw [color=green] (39.center) to (40.center) arc (-180:0:0.25)
										to (42.center) arc (0:180:0.25);
										\node [anchor=west] (43) at (043) {
										   $\Sh_{\dscorefun_{q,\D}}(n_0) = \frac{2\,030\,400}{11!}$
											};
		\node [anchor=west] (44) at (044) {
		$\Sh_{\dscorefun_{q,\D}}(n_5) = \frac{3\,408\,480}{11!}$
			};
	\end{pgfonlayer}
\end{tikzpicture}%
   
\caption{Instance for \Cref{ex:aotbe}.
$\Dn\defeq \{n_1,\dots, n_{\mathrm{x}}\}$, $\Dx\defeq\{x_1,x_2\}$, and the query $q$ is such that the "minimal supports" are the colored regions (the red ones contain $e_0$, the blue ones $e_5$ and the green ones neither).}
\label{fig:sharing-ex-2}
\end{figure}

\begin{example}\label{ex:aotbe}
  Consider the instance depicted in \Cref{fig:sharing-ex-2}.
  Fact $n_0$ appears in two "minimal supports" of size 3 (in red) while $n_5$ appears in one of size 3 and one of size 4 (in blue).
  However, the situation isn’t entirely symmetrical between $n_0$ and $n_5$ because the "minimal supports" that contain neither of them (in green) interfere with $n_0$ only. As it turns out, the "drastic Shapley" value will consider that this difference is sufficient to attribute a smaller value to $n_0$. Note that if one were to remove $x_1$ and $x_2$, the green "minimal supports" would disappear and $n_0$ would get the highest value.  
   This may be unexpected to a non-expert end user, since these two facts don't share any "minimal support" with either $n_0$ nor $n_5$.
\end{example}

Before moving on to the formalization of \ref{MS:1} and \ref{MS:2}, we shall first introduce a %
simpler concrete criterion by focusing
on a particular class of scenarios -- depicted in \Cref{fig:sharing-ex-3} --
in which there are no non-trivial interactions between the "minimal supports". %
The fact that any two "minimal supports" may only intersect on %
$\{\alpha\}$ or $\{\beta\}$ prevents such unwanted interactions, %
while the condition 
that  $|A_i| \le |B_i|$ for all $i$ ensures a loose domination of $\alpha$ over $\beta$. 
By requiring that this domination is strict in some sense, 
we finally obtain the following necessary condition for a 
"responsibility measure" $\phi$. 

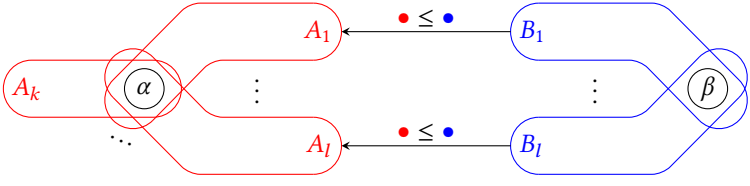
\begin{figure}[tb]
\centering
      \begin{tikzpicture}[scale=1.5]
	\coordinate (00) at (-2.5, 0);
	\coordinate (05) at (2.5, 0);
	\coordinate (011) at (2, -0.25);
	\coordinate (012) at (2, -0.75);
	\coordinate (013) at (2.75, 0.25);
	\coordinate (014) at (2.25, -0.25);
	\coordinate (015) at (1, -0.25);
	\coordinate (016) at (1, -0.75);
	\coordinate (017) at (2, 0.25);
	\coordinate (018) at (2, 0.75);
	\coordinate (019) at (1, 0.75);
	\coordinate (020) at (1, 0.25);
	\coordinate (021) at (2.25, 0.25);
	\coordinate (022) at (2.75, -0.25);
	\coordinate (023) at (-2.75, -0.25);
	\coordinate (024) at (-2, -0.75);
	\coordinate (025) at (-1, -0.75);
	\coordinate (026) at (-1, -0.25);
	\coordinate (027) at (-2, -0.25);
	\coordinate (028) at (-2.25, 0.25);
	\coordinate (029) at (-2.25, -0.25);
	\coordinate (030) at (-2, 0.25);
	\coordinate (031) at (-1, 0.25);
	\coordinate (032) at (-1, 0.75);
	\coordinate (033) at (-2, 0.75);
	\coordinate (034) at (-2.75, 0.25);
	\coordinate (041) at (-1, -0.75);
	\coordinate (042) at (-1, 0.75);
	\coordinate (043) at (-2.5, 0.25);
	\coordinate (044) at (-3.5, 0.25);
	\coordinate (045) at (-3.5, -0.25);
	\coordinate (046) at (-2.5, -0.25);
	\begin{pgfonlayer}{nodelayer}
		\node [draw, circle, minimum height=1.5em] (0) at (00) {};
		\node at (00) {$\alpha$};
		\node [draw, circle, minimum height=1.5em] (5) at (05) {};
		\node at (05) {$\beta$};
		\node [] (11) at ($(011)+(-0.1,0)$) {};
		\node [] (12) at ($(012)+(0.1,0)$) {};
		\node [] (13) at ($(013)+(-0.0732233,-0.0732233)+(0.1,-0.1)$) {};
		\node [] (14) at ($(014)+(0.0732233,+0.0732233)+(0.1,-0.1)$) {};
		\node [] (15) at (015) {};
		\node [] (16) at (016) {};
		\node [] (17) at ($(017)+(-0.1,0)$) {};
		\node [] (18) at ($(018)+(0.1,0)$) {};
		\node [] (19) at (019) {};
		\node [] (20) at (020) {};
		\node [] (21) at ($(021)+(0.0732233,-0.0732233)+(0.1,0.1)$) {};
		\node [] (22) at ($(022)+(-0.0732233,0.0732233)+(0.1,0.1)$) {};
		\node [] (23) at ($(023)+(0.0732233,0.0732233)+(-0.1,0.1)$) {};
		\node [] (24) at ($(024)+(-0.1,0)$) {};
		\node [] (25) at (025) {};
		\node [] (26) at (026) {};
		\node [] (27) at ($(027)+(0.1,0)$) {};
		\node [] (28) at ($(028)+(-0.0732233,-0.0732233)+(-0.1,0.1)$) {};
		\node [] (29) at ($(029)+(-0.0732233,0.0732233)+(-0.1,-0.1)$) {};
		\node [] (30) at ($(030)+(0.1,0)$) {};
		\node [] (31) at (031) {};
		\node [] (32) at (032) {};
		\node [] (33) at ($(033)+(-0.1,0)$) {};
		\node [] (34) at ($(034)+(0.0732233,-0.0732233)+(-0.1,-0.1)$) {};
		\node [] (41) at (041) {};
		\node [] (42) at (042) {};
		\node [] (43) at ($(043)+(0.08,0)$) {};
		\node [] (44) at (044) {};
		\node [] (45) at (045) {};
		\node [] (46) at ($(046)+(0.08,0)$) {};
		\node [anchor=center] at ($(00)+(1,0.05)$) {$\vdots$};
		\node [anchor=center] at ($(05)+(-1,0.05)$) {$\vdots$};
	\end{pgfonlayer}
	\begin{pgfonlayer}{edgelayer}
		\draw [color=blue,rounded corners] (14.center) to (17.center) to (20.center);
		\draw [color=blue] (20.center) arc (-90:-270:0.25) coordinate[midway] (a1);
		\draw [color=blue,rounded corners] (19.center) to (18.center) to (13.center);
		\draw [color=blue] (13.center) arc (45:-135:0.25);
		\draw [color=blue,rounded corners] (22.center) to (12.center) to (16.center);
		\draw [color=blue] (16.center) arc (-90:-270:0.25) coordinate[midway] (a3);
		\draw [color=blue,rounded corners] (15.center) to (11.center) to (21.center);
		\draw [color=blue] (21.center) arc (135:-45:0.25);
		\draw [color=red,rounded corners] (23.center) to (24.center) to (25.center);
		\draw [color=red] (25.center) arc (-90:90:0.25) coordinate[midway] (a4);
		\draw [color=red,rounded corners] (26.center) to (27.center) to (28.center);
		\draw [color=red] (28.center) arc (45:225:0.25);
		\draw [color=red,rounded corners] (29.center) to (30.center) to (31.center);
		\draw [color=red] (31.center) arc (-90:90:0.25) coordinate[midway] (a2);
		\draw [color=red,rounded corners] (32.center) to (33.center) to (34.center);
		\draw [color=red] (34.center) arc (135:315:0.25);
		\draw [color=red] (43.center) to (44.center) arc (90:270:0.25) coordinate[midway] (a5) to (46.center) arc (-90:90:0.25);
		\path ($(0)+(-0.5,0)$) arc (-180:-45:0.5) node[midway,sloped] {\dots};
		
		\draw [->,>=stealth] (a1) to node[midway,above=-.05] {${\color{red}\bullet} \le {\color{blue}\bullet}$} (a2);
		\draw [->,>=stealth] (a3) to node[midway,above=-.05] {${\color{red}\bullet} \le {\color{blue}\bullet}$}  (a4);
		\node[right,color=blue] at (a1) {$B_1$};
		\node[left,color=red] at (a2) {$A_1$};
		\node[right,color=blue] at (a3) {$B_l$};
		\node[left,color=red] at (a4) {$A_l$};
		\node[right,color=red] at (a5) {$A_k$};
	\end{pgfonlayer}
\end{tikzpicture}%
   
\caption{A visual representation of the setting considered by \ref{MS:t}. The $A_i,B_j$ are the "minimal supports" of $q$.}
\label{fig:sharing-ex-3}
\end{figure}

\begin{enumerate}[align=left]
\item[{\crtcrossreflabel{(MStest)}[MS:t]}] Let $A_1, \dots, A_k, B_1, \dots, B_l$ be sets of facts, with $k\ge l$, such that for all $i,j$ $A_i\cap B_j =\emptyset$, for all $i,j$ with $i\neq j$ $A_i \cap A_j = \{\alpha\}$, $B_i \cap B_j = \{\beta\}$ and $|A_i|\le |B_i|$.
   Consider an instance on the input database $\+D=A_1 \cup \dots \cup A_k \cup B_1 \cup \dots \cup B_l$ where $q$ is such that the minimal supports are the $A_i$ and $B_j$. If furthermore $k> l$ or $\exists i. |A_i| < |B_i|$, then $\phi_{q,\D}(\alpha) > \phi_{q,\D}(\beta)$.
\end{enumerate}

Note that although this condition may seem very specific,  
it %
suffices to cover the cases in \Cref{fig:counter-ex}:%

\begin{claim}
   Consider the example of \Cref{fig:counter-ex}. Then any "responsibility measure" $\phi$ that satisfies \ref{MS:t} will be such that $\phi_{q,\D_{\text{(a)}}}(e_0) > \phi_{q,\D_{\text{(a)}}}(e_1)$ and $\phi_{q,\D_{\text{(b)}}}(e_7) > \phi_{q,\D_{\text{(b)}}}(e_5)$.
\end{claim}
\begin{proof}
   \proofcase{(a)} By setting $k=l=1$, $\alpha= e_0$, $\beta= e_1$, $A_1= \{e_0\}$ and $B_1 = \{e_1,e_2\}$, \ref{MS:t} ensures that $\phi_{q,\D_{\text{(a)}}}(e_0) > \phi_{q,\D_{\text{(a)}}}(e_1)$.\medskip

   \proofcase{(b)} By setting $k=2$, $l=1$ $\alpha= e_7$, $\beta= e_5$, $A_1= \{e_7,e_8\}$, $A_2= \{e_7,e_9\}$, and $B_1= \{e_5,e_6\}$. \ref{MS:t} ensures that $\phi_{q,\D_{\text{(b)}}}(e_7) > \phi_{q,\D_{\text{(b)}}}(e_5)$. 
\end{proof}

As expected the "drastic Shapley" and "drastic Banzhaf" values both fulfill this condition:

\begin{proposition}\label{prop:weak-ax-sat}
Any "drastic Shapley-like" score with positive "coefficient function" satisfies \ref{MS:t}.
\end{proposition}
\begin{proof}
Although a direct proof is certainly possible, we shall see in \Cref{ssec:ax2} that the formal version of \ref{MS:1} and \ref{MS:2} imply \ref{MS:t} (\cref{lem:weak-ax}) and that any "drastic Shapley-like" score with positive "coefficient function" satisfies both of the former (\Cref{prop:true-ax-sat-sh}).
\end{proof}

\subsection{Formally Defining Desired Properties}\label{ssec:ax2}
We now present and explain formal statements of the axioms \ref{MS:1} and \ref{MS:2}.
The key will be to define a suitable notion of \emph{domination witness} $\Delta(\alpha,\beta)$ that should ensure that $\alpha$ must dominate $\beta$. 

A first attempt simply relies on the canonical way of comparing the cardinals of two sets, which consists in finding some injective mapping $f$ between them.
\AP
Given a fact $\alpha\in \D$, we say that $S\subseteq\D\setminus\{\alpha\}$ is a ""completing set"" for $\alpha$ if $S\cup\{\alpha\}\in \Minsups{q}(\D)$, and denote the set of all "completing sets" for $\alpha$ by $\mathsf{CS}(q,\D,\alpha)$.
\AP
An ""a-domination witness"" $\Delta(\alpha,\beta)$ of $\alpha$ over $\beta$ is given by:
\begin{itemize}
\item an injective mapping $f: \mathsf{CS}(q,\D,\beta) \hookrightarrow \mathsf{CS}(q,\D,\alpha)$;%
\item for every $S\in\mathsf{CS}(q,\D,\beta)$, an injective mapping $\eta_S : f(S) \hookrightarrow S$.
\end{itemize}
Using this notion, we can formulate the following axioms: 
\begin{itemize}
   \item[{\crtcrossreflabel{(MS1a)}[MS:1a]}] If there exists an "a-domination witness" $\Delta(\alpha,\beta)$ such that one of the $\eta_S : f(S) \hookrightarrow S$ isn’t surjective, then $\phi_{q,\D}(\alpha) > \phi_{q,\D}(\beta)$.
   \item[{\crtcrossreflabel{(MS2a)}[MS:2a]}] If there exists an "a-domination witness" $\Delta(\alpha,\beta)$ such that $f: \mathsf{CS}(q,\D,\beta) \hookrightarrow \mathsf{CS}(q,\D,\alpha)$ isn’t surjective, then $\phi_{q,\D}(\alpha) > \phi_{q,\D}(\beta)$.\medskip
\end{itemize}
Observe that in the simplified setting used to define \ref{MS:t} and depicted in \Cref{fig:sharing-ex-3}, $f$ corresponds to the black arrows and each $\eta_S$ certifies that $|f(S)|\le |S|$.\medskip

Alas, the notion of "a-domination witness" is too simplistic because it considers all "minimal supports" as independent when in reality they can share many elements. To illustrate this, consider the example represented in \Cref{fig:sharing-ex}. $e_0$ appears in two minimal supports of size 3, whereas $e_5$ only appears in a minimal support of size 3 and one of size 4; only comparing the cardinals would thus result in a claim that $e_0$ should necessarily have a higher responsibility than $e_1$. However, there is another difference between the two in the fact that the "completing sets" of $e_0$ intersect whereas those of $e_5$ don’t, and as it turns out the "drastic Shapley" value gives the score of $\frac{681\,120}{10!}$ to both. In order to compare ``all other things being equal'' we therefore need to add a hypothesis to the definition of "a-domination witness" to ensure the intersections between the "completing sets" are taken into consideration.
\begin{figure}[tb]
\centering
      \begin{tikzpicture}[scale=1.5]
	\coordinate (00) at (-1.25, 0);
	\coordinate (01) at (-0.75, 0.5);
	\coordinate (03) at (-0.75, -0.5);
	\coordinate (04) at (-2, 0);
	\coordinate (05) at (3.25, 0);
	\coordinate (06) at (2.75, 0.5);
	\coordinate (07) at (1.75, 0.5);
	\coordinate (08) at (2.75, -0.5);
	\coordinate (09) at (1.75, -0.5);
	\coordinate (010) at (0.75, -0.5);
	\coordinate (011) at (2.75, -0.25);
	\coordinate (012) at (2.75, -0.75);
	\coordinate (013) at (3.5, 0.25);
	\coordinate (014) at (3, -0.25);
	\coordinate (015) at (0.75, -0.25);
	\coordinate (016) at (0.75, -0.75);
	\coordinate (017) at (2.75, 0.25);
	\coordinate (018) at (2.75, 0.75);
	\coordinate (019) at (1.75, 0.75);
	\coordinate (020) at (1.75, 0.25);
	\coordinate (021) at (3, 0.25);
	\coordinate (022) at (3.5, -0.25);
	\coordinate (024) at (-1, -0.75);
	\coordinate (025) at (-0.5, -0.25);
	\coordinate (026) at (-2, 0.25);
	\coordinate (027) at (-2.25, -0.25);
	\coordinate (028) at (-1, 0.25);
	\coordinate (029) at (-1, -0.25);
	\coordinate (030) at (-2.25, 0.5);
	\coordinate (031) at (-0.5, 0.25);
	\coordinate (032) at (-2, -0.25);
	\coordinate (033) at (-1, 0.75);
	\coordinate (034) at (4, 0.5);
	\coordinate (035) at (4, -0.5);
	\begin{pgfonlayer}{nodelayer}
		\node [draw, circle, minimum height=1.5em] (0) at (00) {};
		\node at (00) {$n_1$};
		\node [draw, circle, minimum height=1.5em] (1) at (01) {};
		\node at (01) {$n_2$};
		\node [draw, circle, minimum height=1.5em] (3) at (03) {};
		\node at (03) {$n_3$};
		\node [draw, circle, minimum height=1.5em] (4) at (04) {};
		\node at (04) {$n_0$};
		\node [draw, circle, minimum height=1.5em] (5) at (05) {};
		\node at (05) {$n_5$};
		\node [draw, circle, minimum height=1.5em] (6) at (06) {};
		\node at (06) {$n_6$};
		\node [draw, circle, minimum height=1.5em] (7) at (07) {};
		\node at (07) {$n_7$};
		\node [draw, circle, minimum height=1.5em] (8) at (08) {};
		\node at (08) {$n_8$};
		\node [draw, circle, minimum height=1.5em] (9) at (09) {};
		\node at (09) {$n_9$};
		\node [draw, circle, minimum height=1.5em] (10) at (010) {};
		\node at (010) {$n_{\mathrm{x}}$};
		\node [] (11) at ($(011)+(-0.1,0)$) {};
		\node [] (12) at ($(012)+(0.1,0)$) {};
		\node [] (13) at ($(013)+(-0.0732233,-0.0732233)+(0.1,-0.1)$) {};
		\node [] (14) at ($(014)+(0.0732233,+0.0732233)+(0.1,-0.1)$) {};
		\node [] (15) at (015) {};
		\node [] (16) at (016) {};
		\node [] (17) at ($(017)+(-0.1,0)$) {};
		\node [] (18) at ($(018)+(0.1,0)$) {};
		\node [] (19) at (019) {};
		\node [] (20) at (020) {};
		\node [] (21) at ($(021)+(0.0732233,-0.0732233)+(0.1,0.1)$) {};
		\node [] (22) at ($(022)+(-0.0732233,0.0732233)+(0.1,0.1)$) {};
		\node [] (24) at ($(024)+(0.0732233,0.0732233)$) {};
		\node [] (25) at ($(025)+(-0.0732233,-0.0732233)$) {};
		\node [] (26) at (026) {};
		\node [] (27) at ($(027)$) {};
		\node [] (28) at ($(028)+(-0.0732233,-0.0732233)+(-0.1,0.1)$) {};
		\node [] (29) at ($(029)+(-0.0732233,0.0732233)+(-0.1,-0.1)$) {};
		\node [] (30) at ($(030)+(0.1,0)$) {};
		\node [] (31) at ($(031)+(-0.0732233,0.0732233)$) {};
		\node [] (32) at (032) {};
		\node [] (33) at ($(033)+(0.0732233,-0.0732233)$) {};
		\node [] (34) at ($(034)+(0.0732233,-0.0732233)$) {};
	\end{pgfonlayer}
	\begin{pgfonlayer}{edgelayer}
		\draw [color=blue,rounded corners] (14.center) to (17.center) to (20.center);
		\draw [color=blue] (20.center) arc (-90:-270:0.25);
		\draw [color=blue,rounded corners] (19.center) to (18.center) to (13.center);
		\draw [color=blue] (13.center) arc (45:-135:0.25);
		\draw [color=blue,rounded corners] (22.center) to (12.center) to (16.center);
		\draw [color=blue] (16.center) arc (-90:-270:0.25);
		\draw [color=blue,rounded corners] (15.center) to (11.center) to (21.center);
		\draw [color=blue] (21.center) arc (135:-45:0.25);
		\draw [color=red, rounded corners] (25.center) to (28.center) to (26.center);
		\draw [color=red] (26.center) arc (90:270:0.3) to [bend left=20] (24.center)
									  arc (-135:45:0.25);
		
		\draw [color=red, rounded corners] (31.center) to (29.center) to (32.center);
		\draw [color=red] (32.center) arc (-90:-270:0.3) to [bend right=20] (33.center) arc (135:-45:0.25);
	     \node [anchor=west] (34) at (034) {$\Sh_{\dscorefun_{q,(\D,\emptyset)}}(n_0) = \frac{681\,120}{10!}$};
	     \node [anchor=west] (35) at (035) {$\Sh_{\dscorefun_{q,(\D,\emptyset)}}(n_5) = \frac{681\,120}{10!}$};
	\end{pgfonlayer}
\end{tikzpicture}%
   
\caption{Example of intersecting minimal supports. The database $\D$ is the set $\{e_0\dots e_9,e_\mathrm{x}\}$, and the query $q$ is such that the "minimal supports" are the colored regions (the red ones contain $e_0$ and the blue ones contain $e_5$).}
\label{fig:sharing-ex}
\end{figure}
\AP To this end, let us define 
a ""b-domination witness"" $\Delta(\alpha,\beta)$ of $\alpha$ over $\beta$ as consisting of: %
   \begin{itemize}
      \item an injective mapping $f: \mathsf{CS}(q,\D,\beta) \hookrightarrow \mathsf{CS}(q,\D,\alpha)$;
      \item for every $S\in\mathsf{CS}(q,\D,\beta)$, an injective mapping $\eta_S : f(S) \hookrightarrow S$;
      \item a bijection $\eta : \D\setminus\{\alpha\} \hookrightarrow \D\setminus\{\beta\}$ such that for all $S\in \mathsf{CS}(q,\D,\beta)$, and all $\gamma \in f(S)$, $\eta_S(\gamma) = \eta(\gamma)$.
   \end{itemize}
Essentially, the new hypothesis states that the different $\eta_S$ must agree on the value of any $\gamma$. This solves the issue with the example on \Cref{fig:sharing-ex}
by preventing the two completing sets of $e_5$ from mapping to the two completing sets of $e_0$: say $f$ maps $\{e_6,e_7\}$ to  $\{e_1,e_2\}$ and $\{e_8,e_9,e_\mathrm{x}\}$ to $\{e_1,e_3\}$; then no correct $\eta$ would exist because $\eta(e_1)$ would need to be in $\{e_6,e_7\}\cap\{e_8,e_9,e_\mathrm{x}\}=\emptyset$.\medskip

While this second iteration is a step in the right direction, it isn’t quite enough. To see this, recall \Cref{ex:aotbe}
(depicted in \Cref{fig:sharing-ex-2}).
In this example, the "completing sets" for $e_0$ and $e_5$ are all disjoint from one another, meaning there exists a "b-domination witness" of $n_0$ over $n_5$,
but the "completing sets" for $n_0$ intersect some "minimal supports" (in green) that don’t contain it, thereby weakening its contribution (according to the "drastic Shapley" value at least). We therefore need to add a final condition to avoid this last kind of interference. 
\AP
We shall thus define 
a ""domination witness"" $\Delta(\alpha,\beta)$ of $\alpha$ over $\beta$ as being given by:
   \begin{itemize}
      \item an injective mapping $f: \mathsf{CS}(q,\D,\beta) \hookrightarrow \mathsf{CS}(q,\D,\alpha)$;
      \item for every $S\in\mathsf{CS}(q,\D,\beta)$, an injective mapping $\eta_S : f(S) \hookrightarrow S$;
      \item there exists a bijection $\eta : \D\setminus\{\alpha\} \hookrightarrow \D\setminus\{\beta\}$ such that for all $S,\gamma$, $\eta_S(\gamma) = \eta(\gamma)$;%
      \item for all $X\in\D\setminus\{\alpha\}$, if $X$ contains some "minimal support" then $\eta(X)$ does.
   \end{itemize}
With this notion at hand, we arrive at our proposed formalization of axioms \ref{MS:1} and \ref{MS:2}.
\begin{itemize}
   \item[{\crtcrossreflabel{(MS1)}[fMS:1]}]%
      If there exists a "domination witness" $\Delta(\alpha,\beta)$ such that one of the $\eta_S : f(S) \hookrightarrow S$ isn’t surjective, then $\phi_{q,\D}(\alpha) > \phi_{q,\D}(\beta)$.
   \item[{\crtcrossreflabel{(MS2)}[fMS:2]}] If there exists a "domination witness" $\Delta(\alpha,\beta)$ such that $f: \mathsf{CS}(q,\D,\beta) \hookrightarrow \mathsf{CS}(q,\D,\alpha)$ isn’t surjective, then $\phi_{q,\D}(\alpha) > \phi_{q,\D}(\beta)$.\medskip
\end{itemize}

We now show that the "drastic Shapley" and "drastic Banzhaf" values satisfy these axioms indeed.

\begin{proposition}\label{prop:true-ax-sat-sh}
Let $\phi^c$ be a "drastic Shapley-like" score with positive "coefficient function" $c$. Then $\phi^c$ satisfies \ref{Shdb:1}, \ref{Shdb:2}, \ref{fMS:1}, and \ref{fMS:2}.
\end{proposition}

\begin{proof}
Recall that "drastic Shapley-like" scores are defined by the following equation

\begin{align}\label{eq:d-sh-like}
   \phi^c_{q,\D}(\alpha)&= \S^c_{\dscorefun_{q,\D}}(\alpha)\notag\\
      &= \sum_{X\subseteq \Dn \setminus \set \alpha } c(|X|,|\D|)(v(X\cup\{\alpha\}) - v(X)) \quad\text{ if $\Dx\not\models q$, 0 otherwise}
\end{align}
with $c$ the "coefficient function" and $v$ the $0/1$ function associated with the query.

\ref{Shdb:1} derives from the known axiom \ref{Sh:1} of "Shapley-like" scores.
For \ref{Shdb:2}: by definition, if a player $\alpha$ is "irrelevant", then $v(X\cup\{\alpha\}) = v(X)$ for every $X$, hence $\phi^c_{q,\D}(\alpha)=0$. Otherwise there is at least a term in the sum \st $v(X\cup\{\alpha\}) - v(X)\neq 0$; now $c$ has positive values and $v(X\cup\{\alpha\}) - v(X)\in\{0;1\}$ since $q$ is "monotone" and "Boolean", therefore $\phi^c_{q,\D}(\alpha)>0$.\medskip

Now assume there exists a "domination witness" $\Delta(\alpha, \beta)$. We can apply $\eta$, which is a bijection, in \eqref{eq:d-sh-like} to express $\phi^c_{q,\D}(\beta)$ as a sum over subsets of $\D \setminus \set \alpha$:

\begin{align*}
   \phi^c_{q,\D}(\beta) &= \sum_{X\subseteq \Dn \setminus \set \alpha } c(|\eta(X)|,|\D|)(v(\eta(X)\cup\{\beta\}) - v(\eta(X)))
\end{align*}

Consider some $X\subseteq \D \setminus \set \alpha$ such that $v(\eta(X)\cup\{\beta\}) - v(\eta(X))\neq 0$, \ie some $\eta(X)$ that contains a "completing set" $S$ for $\beta$ and no "minimal support". By the definition of "domination witnesses", $X$ then contains $\eta^{-1}(S) = \eta_S^{-1}(S) = f(S)$, which is a "completing set" for $\alpha$, and no "minimal support". In other words, we have $v(X\cup\{\alpha\}) - v(X)\neq 0$, and since the only other possible value is 1 as explained earlier, $v(\eta(X)\cup\{\beta\}) - v(\eta(X)) = v(X\cup\{\alpha\}) - v(X)$, which yields the following:

\begin{align*}
   \phi^c_{q,\D}(\beta) &\le \sum_{X\subseteq \Dn \setminus \set \alpha } c(|X|,|\D|)(v(X\cup\{\alpha\}) - v(X))\\
			&= \phi^c_{q,\D}(\alpha)
\end{align*}

At this stage, all it takes for the above inequality to become strict is a subset $X$ such that $v(\eta(X)\cup\{\beta\}) - v(\eta(X)) < v(X\cup\{\alpha\}) - v(X)$, and this subset can be obtained from the definition of each axiom. \medskip

\proofcase{(MS1)} If one of the $\eta_S : f(S) \hookrightarrow S$ isn’t surjective, then $\eta(f(S))$ is a proper subset of $S$, meaning it isn’t (nor contains) a "completing set" for $\alpha$, therefore $X=f(S)$ gives the desired result. \medskip

\proofcase{(MS2)} If $f: \mathsf{CS}(q,\D,\beta) \hookrightarrow \mathsf{CS}(q,\D,\alpha)$ isn’t surjective, then there is a $X\in\mathsf{CS}(q,\D,\alpha)$ that has no inverse image by	$f$. We can further assume that all $\eta_S$ are bijective, otherwise the result is already given by \ref{fMS:1}; in other words, for every "completing set" $S$ for $\beta$, $\eta^{-1}(S)$ is a "completing set" for $\alpha$.
Since $\eta$ is bijective, if $\eta(X)$ were to contain some $S\in\mathsf{CS}(q,\D,\beta)$, then $\eta^{-1}(S)\subseteq X$ would be a "completing set" for $\alpha$, which is contradictory. At last, $\eta(X)$ cannot contain any "minimal support" because $X$ doesn’t. Overall, this $X$ gives the desired result.%
\end{proof}

Finally we prove that \ref{MS:t} is indeed a necessary condition to the pair \ref{MS:1} \ref{MS:2}: %

\begin{lemma}\label{lem:weak-ax}
Any "responsibility measure" that satisfies \ref{MS:1} and \ref{MS:2} also satisfies \ref{MS:t}.
\end{lemma}
\begin{proof}
   Take the notations from \ref{MS:t}, as summarized in \Cref{fig:sharing-ex-3}.
   The "completing sets" for $\alpha$ (resp. $\beta$) are the $A_i\setminus\{\alpha\}$ (resp. the $B_i\setminus\{\beta\}$), which we shorten to $A_i^\alpha$ (resp. $B_i^\beta$).
   Define the "domination witness" $\Delta(\alpha,\beta)$ as follows:
\begin{itemize}
   \item The injective mapping $f: \mathsf{CS}(\+D,q,\beta) \hookrightarrow \mathsf{CS}(\+D,q,\alpha)$ is defined by $f(B_i^\beta) = A_i^\alpha$ for %
      $i\le l$.
   \item For every $i\le l$, $|A_i|\le |B_i|$ implies the existence of some injective mapping $\eta_{B_i^\beta} : A_i^\alpha \hookrightarrow B_i^\beta$.
   \item The only elements that are changed by some $\eta_S$ are the elements of some $A_i^{\alpha}$, which only appear in one of them, hence only have one possible image. This means that the bijection $\eta : \+D\setminus\{\alpha\} \hookrightarrow \+D\setminus\{\beta\}$ that swaps $\alpha$ with $\beta$, every $\gamma\in A_i^\alpha$ with $\eta_{A_i^\alpha}(\gamma)$ and leaves the rest unchanged is such that for all $S,\gamma$, $\eta_S(\gamma) = \eta(\gamma)$.%
   \item For all $X\in\+D\setminus\{\alpha\}$, if $X$ contains some "minimal support" then it can only be some $B_i$, because all others contain $\alpha$. Now by definition $\eta$ swaps every element of $A_i$ with some element of $B_i$ meaning that $\eta(B_i)\subseteq \eta(X)$ will contain the "minimal support" $A_i$.
\end{itemize}

As shown above, $\Delta(\alpha,\beta)$ is a "domination witness" indeed.
We conclude with the last hypothesis of \ref{MS:t}: if $\exists i. |A_i| < |B_i|$ then $\eta_{A_i^\alpha}$ isn’t surjective and \ref{fMS:1} gives the desired result,
and if $k>l$ then $f$ isn’t surjective, we can conclude using \ref{fMS:2}.
\end{proof}

\section{Alternative Responsibility Measures}
\label{sec:alt-measures}
We have established what we want from our "responsibility measures" and shown that the "drastic Shapley value" perfectly fulfills them.
We have also seen, however, that these desiderata are distinct from the \ref{Sh:1}-\ref{Sh:3} axioms that uniquely characterize the "Shapley value", and that other responsibility measures such as the "drastic Banzhaf value" also fulfill them.
Coupled with the fact that $\dShapley_{q}$ is $\sP$-hard for some very simple "CQ"s \cite{livshitsShapleyValueTuples2021},
this motivates us to consider alternative measures, which can be done in two ways:
\begin{description}
    \item[\textnormal{\em Approach 1}:] we may consider alternative responsibility attribution functions more specifically tailored to %
    "monotone" "non-numeric queries" which may not be based on a Shapley score, or 
    \item[\textnormal{\em Approach 2}:] we may insist in the use of the Shapley value but for a different associated wealth function, one in which axioms \ref{Sh:3} and \ref{Sh:4} are more meaningful.%
\end{description}

This section is organized as follows. In \Cref{ssec:wsmsdef},
we follow the first approach by proposing an alternative family of measures conforming to \ref{Shdb:1}, \ref{Shdb:2}, \ref{MS:1} and \ref{MS:2}.
Then in \Cref{ssec:apr2},
we show that these alternative measures can all be seen as "Shapley values" for different wealth functions, and thus that they could also have been found by following the second approach. Finally \Cref{ssec:otherwsms} will suggest some further explicit measures within the family.

\subsection{Weighted Sums of Minimal Supports}
\label{ssec:wsmsdef}
\AP
Given that "minimal supports" are central to \ref{MS:1}, \ref{MS:2} and also \ref{Shdb:2}, it seems natural to employ this notion directly in the definition of responsibility measures. We shall thus introduce the
 ""weighted sum of minimal supports"" (or \reintro{WSMS}), as the quantitative responsibility measure 
\begin{equation}\label{eq:wsms}
   \phi_{q,\D}^w(\alpha) \defeq \sum_{\substack{M\in\Minsups q(\D)\\\alpha\in M}} w(|M|,|\D|) \quad\text{ if $\Dx\not\models q$, or 0 otherwise}
\end{equation}
for some \AP""weight function"" $w: \lN \times \lN \to \lR$.
We say that $w$ is a \AP""tractable weight function"" if it can further be computed in polynomial time. Note the special case for $\Dx\models q$, like the one for $\Dscorefun$.\medskip

In terms of theoretical guarantees, virtually all "WSMSs" satisfy all discussed axioms, even the simpler (and stronger) axioms \ref{MS:1a} and \ref{MS:2a}, which boil down to considering that the interactions between the minimal supports have little impact on responsibility.

\begin{proposition}\label{prop:true-ax-sat-ms}
If $w$ is \AP""positive and strictly decreasing"" (\ie such that $w(k,n) > 0$ and $k < l \Rightarrow w(k,n) > w(l,n)$ for every $k,l,n$), then $\phi^w$ satisfies \ref{Shdb:1}, \ref{Shdb:2}, \ref{MS:1a} and \ref{MS:2a}.
As a consequence, $\phi^w$ also satisfies \ref{MS:1} and \ref{MS:2}, as well as \ref{MS:t}.
\end{proposition}
\begin{proof}

\proofcase{\ref{Shdb:1}} It is a direct consequence of \ref{Sh:1} and the fact that $w$ is only a function of the "minimal supports" for $q$, which are invariant under query equivalence.\medskip

\proofcase{\ref{Shdb:2}} If a fact $\alpha$ is "irrelevant", then it appears in no "minimal support" hence the sum in \Cref{eq:wsms} is empty. Otherwise, $\alpha$ appears in some "minimal support" and the sum is non-empty. Further, all terms will be positive because $w$ is, hence the result.\medskip

\proofcase{\ref{MS:1a} and \ref{MS:2a}} If there exists an "a-domination witness" $\Delta(\alpha, \beta)$:

\[
\begin{array}{rl|l}
   \phi^w_{q,\D}(\beta)
      &= \displaystyle\sum_{S\in\mathsf{CS}(q,\D,\beta)} w(|S|+1,\D)
      &\text{by \eqref{eq:wsms} }\\
      &\le \displaystyle\sum_{S\in\mathsf{CS}(q,\D,\beta)}^{\phantom{a}} w(|f(S)|+1,\D)
      &|S|\ge |f(S)| \Rightarrow w(|S|+1,\D)\le w(|f(S)|+1,\D)\\
      &\le \displaystyle\sum_{X\in\mathsf{CS}(q,\D,\alpha)}^{\phantom{a}} w(|X|+1,\D)
      &\text{because $f$ is injective}
\end{array}
\]

and the hypothesis of \ref{MS:1a} (resp. \ref{MS:2a}) makes the second (resp. first) inequality strict.\medskip

\proofcase{Others} \ref{MS:1} and \ref{MS:2} are trivially implied by their counterparts and \ref{MS:t} by \Cref{lem:weak-ax}.
\end{proof}

\AP Arguably the most natural "positive and strictly decreasing" function  is the ""inverse weight function"" $\mathrm{"invw"}: k,s \mapsto \nicefrac{1}{k}$. As \Cref{prop:invw=ms}
shows, $\phi^\mathrm{"invw"}$ coincides with the Shapley value of the "wealth function family" $\intro*\MSscorefun$, defined by $\msscorefun_{q,\D}(S) = \sum_{M\in\mathsf{MS}_q(S\uplus\Dx)} \frac{|M\cap S|}{|M|}$ if $\Dx\not\models
q$, and 0 otherwise.
In the non-trivial case where $\Dx\not\models q$, a "minimal support" $S\subseteq \Dn$ collectively contributes 1 to the total, while a "minimal support" that is, say, $\nicefrac{1}{3}$ "exogenous", only contributes $\nicefrac{2}{3}$.
On purely "endogenous" databases, the formula therefore simplifies to $\msscorefun_{q,\Dn}(S) = |\mathsf{MS}_q(S)|$.
Intuitively, this measure considers that the fact that  $\D\models q$ 
 is more robust the more "minimal supports" for $q$ can be found in $\D$. A similar idea %
 underlies the so-called \textit{MI inconsistency measure} \cite{hunterMeasureConflictsShapley2010} (see \Cref{sec:relwork} for more details).

\begin{proposition}\label{prop:invw=ms}
   The \reintro{MS Shapley} value, that is, the "Shapley value" of $\MSscorefun$, satisfies the formula:%
   \[\Sh_{\msscorefun_{q,\D}}(\alpha) = \phi_{q,\D}^\mathrm{"invw"}(\alpha) = \sum_{\substack{M\in\Minsups{q}(\D)\\\alpha\in M}} \frac{1}{|M|} \quad\text{ if $\Dx\not\models q$, or $0$ otherwise.}\]
\end{proposition}
\begin{proof}
Consider a query $q$ and a "partitioned database" $\D\defeq\Dn\uplus\Dx$. %
For every $M\in \Minsups{q}(\D)$ define the query $\AP\intro*\onemsq_{M}$ whose only minimal support is $M$, that is $\D'\models \onemsq_{M}$ iff $M\subseteq \D'$. It can be easily seen that $\msscorefun_{q,\D}$ on $\D$ can be expressed as $\msscorefun_{q,\D}= \sum_{M\in \Minsups{q}(\D)} \msscorefun_{\onemsq_{M},\D}$ if $\Dx\not\models q$, 0 otherwise.

Since the other case is trivial, assume $\Dx\not\models q$ and consider some $M\in \Minsups{q}(\D)$. By \ref{Sh:2}, $\Sh_{\msscorefun_{\onemsq_{M},\D}}(\alpha)=0$ for any "endogenous" fact $\alpha$ outside $M$. By \ref{Sh:1}, the remaining "endogenous" facts (\ie the elements of $M\cap\Dn$) must all have the same value, and their sum is fixed to $\frac{|M\cap\Dn|}{|M|}$ by \ref{Sh:3} (recall we assume $\D \not\models q$). Overall this means that  $\Sh_{\msscorefun_{\onemsq_{M},\D}}(\alpha)=\nicefrac{1}{|M|}$ if $\alpha\in M$ and 0 otherwise.
Note that, so long as $\Dx \not\models q$, this value isn’t affected by the distribution of "exogenous" facts.
Finally by \ref{Sh:4} we can sum the contributions of all subgames to get the desired formula.
\end{proof}
   \Cref{prop:invw=ms}
   is essentially %
   the result
   \cite[Proposition 8]{hunterMeasureConflictsShapley2010}, 
   and
   solely relies on the Shapley axioms \ref{Sh:1}--\ref{Sh:3}.
   We believe this is %
   evidence that these axioms make more sense %
   for "WSMS"s. In particular, the axiom \ref{Sh:4}, when instantiated with $\MSscorefun$, implies `if $\Minsups{q_1}(\D)\cap \Minsups{q_2}(\D) = \emptyset$, then $\phi_{q_1,\D}(\alpha) + \phi_{q_2,\D}(\alpha) = \phi_{q_1\land q_2,\D}(\alpha)$ for all $\alpha$'. This property is %
    key to \Cref{prop:invw=ms} and further down the line 
    to the efficient computation of the "MS Shapley value" (see \Cref{sec:datacomplexity,sec:combined-complexity}), whereas the instantiation of \ref{Sh:4} by $\Dscorefun$ yields a nonsensical statement. %
\medskip

The first observation that can be made regarding $\phi^\mathrm{"invw"}$ in contrast with the "drastic Shapley value" is the fact that only the latter is affected by the way "minimal supports" may intersect. As such,
$\phi^\mathrm{"invw"}$
will ignore the `green' "minimal supports" in \Cref{ex:aotbe} and attribute a greater value to $e_0$.
In short the "drastic Shapley value" considers the interactions between all minimal supports, while "WSMS"s don’t. We find no conceptual reason to argue that one approach gives better explanations than the other. However, we can say that the explanations given by $\phi^\mathrm{"invw"}$ are \emph{simpler}, both conceptually (it is a straightforward average) and computationally (intuitively, even when the number of "minimal supports" is manageable, considering all their possible interactions can be costly).

\subsection{WSMSs are Shapley-like Scores}\label{ssec:apr2}
In \Cref{sec:drastic-shapley},
we insisted on the term `"drastic Shapley value"', which is simply called `Shapley value' in the database literature \cite{livshitsShapleyValueTuples2021,ourpods24}. This is because, as we have seen, the "Shapley value" proper isn’t a "quantitative measure", but rather a family of "measures" parameterized by a "wealth function family" $\STscorefun$. Moreover, while $\Dscorefun$ is a very natural choice, 
it is not the only reasonable "wealth function family" one can define for "non-numeric queries". %
Indeed, \Cref{prop:invw=ms} shows that instantiating the Shapley value using
$\MSscorefun$
yields the intuitive measure $\phi^\mathrm{"invw"}$.
Interestingly, we can generalize the latter result to show that \emph{every} "WSMS" is actually a Shapley value in disguise:

\begin{proposition}\label{prop:wsmsissh}
   For every "WSMS" $\phi^w$ and every "Shapley-like" score $\S^c$
   whose "coefficient function" $c$ is positive on all inputs, there exists a family $\Wscorefun\defeq\left(\wscorefun_{q,\D}\right)$ such that for every query $q$,  "database" $\D$ and $\alpha\in \D$, $\phi_{q,\D}^w(\alpha)=\S^c_{\wscorefun_{q,\D}}(\alpha)$.
\end{proposition}

\begin{proof}
   To obtain a "WSMS", we shall use a "wealth function" $\zetascorefun$ which itself is a weighted sum of minimal supports in some sense. Note that to simplify notations we simply write $\zetascorefun$, even though the function implicitly depends on the query etc.
	\begin{equation*}\label{formul:vgamma}
	\begin{aligned}
	   \intro*\zetascorefun : \partsof{\Dn} &\to \lR\\
	   S	 &\mapsto \sum_{\substack{M\in\Minsups{q}(\D)\\M\subseteq S}} \gamma(M,\D)%
	\end{aligned}
	\end{equation*}
	The coefficients $\gamma(M,\D)$ are to be set later.\medskip
	
	It is easy to see that any "Shapley-like" score follows the linearity axiom $\S^c_{\scorefun,\D}(\alpha)+\S^c_{\scorefun',\D}(\alpha)=\S^c_{\scorefun+\scorefun',\D}(\alpha)$. From this we can get:
	\begin{equation}\label{eq:blah}
	   \S^c_{\zetascorefun}(\alpha) = \sum_{M\in\Minsups{q}(\D)} \S^c_{\zetascorefun_{M}}(\alpha)
	\end{equation}
	where $\zetascorefun_{M}$ is defined as $\zetascorefun$ with $q$ replaced by $\onemsq_M$ (hence $\zetascorefun_{M}(S) = \gamma(M,\D)$ if $M\in S$ and 0 otherwise).\medskip
	
We now fix some $M\in\Minsups{q}(\D)$, and denote by $M_{\subendo}$
its "endogenous"
part%
. If $\alpha\notin M_{\subendo}$, it is easy to see that $\S^c_{\zetascorefun_{M}}(\alpha)=0$. Now consider $\alpha\in M_{\subendo}$, and denote $M_{\subendo}^-\defeq M_{\subendo}\setminus\{\alpha\}$. By applying \eqref{formul:slike} we get
	\begin{align*}
		\S^c_{\zetascorefun_{M}}(\alpha) &= \sum_{\substack{S\subseteq \Dn\setminus\{\alpha\}\\M_{\subendo}^-\subseteq S}} c(|S|,|\Dn|)\gamma(M,\D)\\
		&= \sum_{k=0}^{|\Dn|-1} c(k,|\Dn|)\gamma(M,\D)|\{S\subseteq \Dn\setminus\{\alpha\} : M_{\subendo}^- \subseteq S \land |S|=k\}|
	\end{align*}
	Computing the rightmost factor boils down to choosing the $k-|M_{\subendo}^-|$ elements of $S\setminus M_{\subendo}$ out of the $|\Dn|-|M_{\subendo}|$ possible ones, since there is no choice on the others. This gives
\begin{equation*}
\S^c_{\zetascorefun_{M}}(\alpha) = \gamma(M,\D)\sum_{k=|M_{\subendo}|-1}^{|\Dn|-1} \binom{|\Dn|-|M_{\subendo}|}{k-|M_{\subendo}|+1} c(k,|\Dn|)
\end{equation*}
	
	Hence by setting
	\begin{equation*}
		\gamma(M,\D) \defeq \frac{w(|M|,|\D|)}{\sum_{k=|M_{\subendo}|-1}^{|\Dn|-1} \binom{|\Dn|-|M_{\subendo}|}{k-|M_{\subendo}|+1} c(k,|\Dn|)}\medskip
	\end{equation*}
	
	we obtain $\S^c_{\zetascorefun}(\alpha) = \phi_{q,\D}^w(\alpha)$ by \Cref{eq:blah}, as desired. Note that there cannot be a division by 0 in the definition of $\gamma$ because the coefficient function $c$ is assumed to be positive on all inputs.\end{proof}

\AP In light of the previous proposition,  given a "WSMS" $\phi^w$, we shall denote by $\intro*\Wscorefun^{w}$ the "wealth function family" such that $\phi_{q,\D}^w(\alpha)=\Sh_{\wscorefun_{q,\D}^{w}}(\alpha)$ and $\reintro*\wShapley_{\C}^{w}$ the associated computational problem. We henceforth abandon the notation $\phi^w$ in favor of the Shapley-based alternative.

\subsection{Other Notable WSMSs}\label{ssec:otherwsms}

By its very nature, the "drastic Shapley value" gives a lot of importance to small "minimal supports", to the point where a fact $\alpha$ such that $\{\alpha\}\models q$ will always receive a maximal "drastic Shapley value" \cite[Lemma 6.3]{ourpods24}. This is not the case for the "MS Shapley value", which will favor three "minimal supports" of size 2 against one singleton support. It is however possible to choose a "positive and strictly decreasing WSMS" such that smaller "minimal support"s will always dominate. %

   \begin{proposition}\label{prop:sshapley}
      Define the "weight function" $w_s(k,n)\defeq 2^{-kn}$ and denote $\intro*\Sscorefun=\Wscorefun^{w_s}$ for short. The function $w_s$ is "tractable@@wf", "positive and strictly decreasing", and if $\alpha,\beta\in \Dn$ are such that the smallest "minimal support" containing $\alpha$ is smaller than the smallest "minimal support" containing $\beta$, then $\Sh_{\sscorefun_{q,\D}}(\alpha) > \Sh_{\sscorefun_{q,\D}}(\beta)$, unless $\Dx\models q$.
   \end{proposition}
   \begin{proof}
      The idea is that the binary representation of $\Sh_{\sscorefun_{q,\D}}(\alpha)$ encodes the vector of numbers of "minimal supports" for $q$ that contain $\alpha$, grouped by increasing size.

   Given its closed form formula, $w_s$ can be computed in polynomial time by elementary means. It is also obviously "positive and strictly decreasing".
   To show the remaining property, first recall that $\Wscorefun^{w}$ is tailored so as to make the following hold:
   \[\Sh_{\sscorefun_{q,\D}}(\alpha) = \sum_{\substack{M\in\Minsups{q}(\D)\\\alpha\in M}} 2^{-|M||\D|} \quad\text{ if $\Dx\not\models q$, 0 otherwise.}\]
   
   Now assume $\Dx\not\models q$ and consider $\alpha, \beta\in \Dn$ such that the smallest "minimal support" that contains $\alpha$ (resp. $\beta$) is of size $k_\alpha$ (resp. $k_\beta$), with $k_\alpha < k_\beta$. On the one hand $\Sh_{\sscorefun_{q,\D}}(\alpha) \ge 2^{-k_\alpha n}$ because any "minimal support" of size $k_\alpha$ will contribute that much to the above sum. On the other hand, if we denote by $m_\beta$ the number of "minimal supports" that contain $\beta$, the above formula implies $\Sh_{\sscorefun_{q,\D}}(\alpha) \le m_\beta\cdot 2^{-k_\beta n} \le m_\beta\cdot 2^{-(k_\alpha-1)n}$. Now $m_\beta \le 2^{n-1} < 2^n$ because $2^{n-1}$ is the total number of subsets of $\D$ that contain $\beta$, hence $\Sh_{\sscorefun_{q,\D}}(\beta) < 2^{-k_\alpha n} < \Sh_{\sscorefun_{q,\D}}(\alpha)$. 
   \end{proof}

Alternatively, one can also decide that the number of minimal supports should be the main factor for the responsibility measure.

   \begin{proposition}%
      Define the "weight function" $w_{\#}(k,n)\defeq 1+w_s(k,n)$ and denote $\intro*\SHARPscorefun=\Wscorefun^{w_{\#}}$ for short. $w_{\#}$ is "tractable@@wf", "positive and strictly decreasing", and if $\alpha,\beta\in \Dn$ are such that $\alpha$ appears in more "minimal supports" than $\beta$, then $\Sh_{\sharpscorefun_{q,\D}}(\alpha) > \Sh_{\sharpscorefun_{q,\D}}(\beta)$, unless $\Dx\models q$.
   \end{proposition}
   \begin{proof}
      Again, given its closed-form formula, $w_{\#}$ can be computed in polynomial time by elementary means. It is also obviously "positive and strictly decreasing" and we have the following formula by definition of $\Wscorefun^{w}$:
   \[\Sh_{\sharpscorefun_{q,\D}}(\alpha) = \sum_{\substack{M\in\Minsups{q}(\D)\\\alpha\in M}} \left(1+ 2^{-|M||\D|} \right) \quad\text{ if $\Dx\not\models q$, 0 otherwise.}\]

   Now if $\Dx\not\models q$, all "minimal supports" are of size at least one and there must necessarily be strictly fewer than $2^{|\D|}$ of them. This means that all the $2^{-|M||\D|}$ terms add up to a value strictly below $1$, hence $\left\lfloor \Sh_{\sharpscorefun_{q,\D}}(\alpha) \right\rfloor$ equals the number of "minimal supports" that contain $\alpha$. From this we can conclude immediately.
   \end{proof}

\section{Data Complexity}
\label{sec:datacomplexity}

We shall now investigate the complexity of computing "weighted sums of minimal supports", \ie computing $\wShapley_{\+C}^{w}$ for "weight functions" $w$. We start with the \AP""data complexity"", \ie the complexity of $\wShapley_{q}^{w}$ for individual queries $q$, independently of their representation.

The simplicity of \Cref{eq:wsms} defining "WSMS"s not only makes these measures easier to interpret, 
but it also yields a simple and tractable algorithm to compute them.
This is because
the problem boils down to counting the "minimal supports" of every size. This straightforward reduction is formalized in the following lemma, which makes use of the aggregate query 
\AP$\intro*\countFMS_q$ that, given a number $k$ and database $\D$, computes the number of "minimal supports" of $q$ in $\D$ %
having exactly $k$ facts (\textsf{\underline{f}ms} stands for \underline{f}ixed-size \textsf{ms}, in line with `fixed-size model counting' \cite[\S 3]{KaraOlteanuSuciuShapleyBack}).
For a class $\class$ of "Boolean" queries, we use \AP$\intro*\evalCountFMS_\class$ for the problem of, given a query $q \in \class$, a number $k$ and a database $\D$, evaluating $\countFMS_q$ on $k,\D$.
If $\+C$ contains a single query $q$, we simply write $\evalCountFMS_{q}$.

   \begin{lemma}\label{lem:wsmsvsmsc1}
      For a "tractable weight function" $w$ and "Boolean" "monotone" "query" $q$, $\wShapley_{q}^{w} \polyrx \evalCountFMS_{q}$.
   \end{lemma}
   \begin{proof}
      Let $w$ be a "tractable weight function" and $\D,\alpha$ define an instance of $\wShapley_{q}^{w}$. Then $\Sh_{\wscorefun_{q,\D}}(\alpha)$  can be expressed as follows using \Cref{eq:wsms}:
      \begin{align*}
        \Sh_{\wscorefun_{q,\D}}(\alpha) &= \sum_{k=1}^{|\D|}\sum_{\substack{S\in\Minsups q(\D)\\\alpha\in S,\,|S|=k}} w(k,|\D|) \quad\text{ if $\Dx\not\models q$, 0 otherwise} \\
        &=\sum_{k=1}^{|\D|} \big(\countFMS_{q}(k,\D) - \countFMS_{q}(k,\D \setminus \set\alpha)\big)  \cdot w(k,|\D|)\quad\text{ if $\Dx\not\models q$, 0 otherwise}
      \end{align*}
      Hence, $\wShapley_{q}^{w} \polyrx \evalCountFMS_{q}$.
   \end{proof}

With \Cref{lem:wsmsvsmsc1}, we can effectively compute $\wShapley_{q}^{w}$ for every query whose "minimal supports" can be enumerated in polynomial time, which includes "UCQ"s or more generally "bounded" "queries".
\begin{theorem}\label{th:data-ucq}
   $\wShapley_{q}^{w}\in\FP$ for every "tractable weight function" $w$ and every  "Boolean" "monotone" "query" $q$ that is "bounded" and 
   can be evaluated in polynomial time.
\end{theorem}
\begin{proof}
By definition, the "minimal supports" for $q$ are bounded by some constant $K$, so one can simply enumerate all subsets of $\+D$ of size at most $K$, and check which ones are "minimal supports" in order to compute $\countFMS_{q}(k,\D)$ for every $k$, and finally apply \Cref{lem:wsmsvsmsc1}.
\end{proof}

As a first step towards an implementation, we show in the next result how we can compute $\countFMS_{q}(k,\D)$  by evaluating simple SQL queries in parallel, for an input query $q \in \UCQneq$. 

\newcommand{\aUCQ}{\widetilde{q}}
\begin{theorem}
\label{thm:rewritting-SQL}
For every $\UCQneq$ $\aUCQ$ and $k \in \Nat$, there exists a set $\Qneq$%
of $\CQneq$ queries such that,  for every "database" $\D$,
$$\countFMS_{\aUCQ}(k,\D) = \sum_{q \in \Qneq} \countAns_{q}(\D) \cdot \gamma_q,$$ where \AP$\intro*\countAns_{q}(\D)$ is the number of "homomorphisms" from $q$ to $\D$, and $\gamma_q$ is a rational number computable from $q$. 
Further, every $q \in \Qneq$ is of at most quadratic size.
\end{theorem}
\begin{proof}
    Given a query $\aUCQ \in \UCQneq$ and a number $k$, we show how to build the set $\Qneq$ with the desired properties of the statement.
    Let us say that $q'$ is a \AP""reduct"" of $q \in \CQneq$ if $q'$ is the result of collapsing variables of $q$ and removing repeated atoms; in particular, $q'$ is a homomorphic image of~$q$ (\ie isomorphic to a query resulting from applying a "homomorphism" to every atom of $q$). Similarly, $q'$ is a "reduct" of $\aUCQ$ if it is a "reduct" of a $\CQneq$ therein.
  Now consider the set $\R_k$ of all "reducts" $q$ of $\aUCQ$ such that 
    \begin{enumerate}[(a)]
        \item $q$ has exactly $k$ \AP""relational atoms"" (\ie "atoms" which are not "inequality atoms");
        \item there is no "reduct" $q'$ of $\aUCQ$ having strictly less than $k$ "relational atoms" such that $q' \homto q$. \label{eq:minimality-relatoms}
    \end{enumerate}
Let $\Q$ be the result of removing from $\R_k$ all redundant queries. Concretely, we initialize $\Q \defeq \R_k$ and we apply the following \AP""redundancy-removal-rule"" until no more queries can be deleted:
    find a query $q \in \Q$ such that $q' \homto q$ for some distinct $q' \in \Q$ and update $\Q \defeq \Q \setminus \set{q}$.
    Finally, let \AP$\intro*\Qneq$ be the result of adding, for each $q \in \Q$ and distinct $x,y \in \vars(q)$, the %
     atom $x \neq y$ to $q$. %

        \begin{claim}\label{cl:invertible-homomorphisms}
            If $q \in \Qneq$, $M \in \Minsups q (\D)$ and $h: q \homto M$, then $h$ is injective. Further, for any query $p \in \CQneq$ such that $p \homto M$, we have $p \homto q$.
        \end{claim}
        \begin{nestedproof}
            The fact that $h$ is injective follows readily from the "inequality atoms" of $q$ by definition. Hence, it is invertible, and given $g: p \homto M$ we have $(h^{-1} \circ g) : p \homto q$.
        \end{nestedproof}
 
We show that the queries in $\Qneq$ partition the minimal supports in $\Minsups \aUCQ(\D)$ of size $k$:

    \begin{claim}\label{cl:minsup-partition}
    $\set{\Minsups q(\D)}_{q \in \Qneq}$ is a partition of $\set{M \in \Minsups \aUCQ(\D): |M| = k}$.
    \end{claim}
    \begin{nestedproof}
        Let us first show 
        \[\set{M \in \Minsups \aUCQ(\D): |M| = k} = \bigcup_{q \in \Qneq} \Minsups q(\D).\] 

        \proofcase{$\subseteq$}
        For the $\subseteq$-containment, first note that if $M$ is a "minimal support" of $\aUCQ$ of size $k$, there must be some $\hat q$ in $\aUCQ$ and some "homomorphism" $h:\hat q \homto M$. Such "homomorphism" $h$ induces a "reduct" $q$ of $\hat q$, in such a way that $q^{\neq} \homto M$, where $q^{\neq}$ is the result of adding $x \neq y$ to $q$ for all $x,y \in \vars(q)$. 
        
        We now claim that either $q^{\neq}$ or a larger query (i.e., a query $p \in \Qneq$ such that $p \homto q^{\neq}$) must be in $\Qneq$. 
        If this is not the case, then it must be that $q^{\neq}$ is not in $\R_k$ because it was discarded by \cref{eq:minimality-relatoms}, which would mean that there is another "reduct" $q'$ of $\aUCQ$ with less than $k$ "relational atoms" such that $q' \homto q^{\neq} \homto M$, implying that there is a "support" strictly included in $M$ and contradicting its "minimality@minimal support".

        \proofcase{$\supseteq$}
        For the $\supseteq$-containment, note that every "minimal support" of $q \in \Qneq$ has size exactly $k$. Take any such "minimal support" $M$ of $q$. By construction, there is some $\hat q$ in $\aUCQ$ such that $\hat q \homto q$ and thus $M$ is also a "support" of $\aUCQ$. By means of contradiction, suppose $M$ is not a "minimal support" of $\aUCQ$, that is, there is $M' \subsetneq M$ and $q'$ in $\aUCQ$ such that $h:q' \homto M'$. Consider the "reduct" $q''$ of $\aUCQ$ induced by $h$, and observe that: 
            (i) $q''$ has strictly less than $k$ "relational atoms" and 
            (ii) $q'' \homto q' \homto M$.
        By \Cref{cl:invertible-homomorphisms} we then have $q'' \homto q$.
        This means that $q$ must not exist in $\Qneq$ since it must have been discarded by \cref{eq:minimality-relatoms}. In view of this contradiction, $M$ is a "minimal support" of $\aUCQ$.

        \proofcase{Disjointness} Let us finally show that parts of the partition are disjoint, that is, $\Minsups q(\D) \cap \Minsups {q'}(\D) = \emptyset$ for all pair of distinct $q,q' \in \Qneq$. By means of contradiction, assume the contrary and let $M \in \Minsups q(\D) \cap \Minsups {q'}(\D)$. 
        Let $h: q \homto M$ and $h':q' \homto M$.
        By \Cref{cl:invertible-homomorphisms} we then have $q \homto q'$.
        By the "redundancy-removal-rule" it cannot be that both distinct queries are in $\Qneq$. This contradiction then shows that there cannot be such $M \in \Minsups q(\D) \cap \Minsups {q'}(\D)$.
    \end{nestedproof}
    
    We next show how to compute $\countAns_{q}(\D)$ for a given $q \in \Qneq$, where we recall that $\reintro*\Auto{q}$ is the set of all "automorphisms@@cq" of $q$.
    \begin{claim}
    $\countAns_{q}(\D) = |\Auto{q}| \cdot |\Minsups q(\D)|$ for all $q \in \Qneq$.
    \end{claim}
    \begin{nestedproof}
       In view of the previous \Cref{cl:minsup-partition}, all "minimal supports" of $q \in \Qneq$ have size $k$. Further, each "homomorphism" $q \homto \D$ induces a "minimal support", and the set of "homomorphisms" $\set{ h \mid h: q \homto \D}$ can be partitioned into $\set{h \mid h:q \homto M}_{M \in \Minsups q(\D)}$.
        Hence, it suffices to show, for an arbitrary $q \in \Qneq$ and $M \in \Minsups q(\D)$, that $\countAns_{q}(M) = |\Auto{q}|$.

        \proofcase{$\leq$}
        Consider the set $H$ of all (pairwise distinct) "homomorphisms" $h: q \homto M$. 
        Let us fix $g \in H$ and observe that, by \Cref{cl:invertible-homomorphisms}, each $(h^{-1}\circ g) : q \homto q$ is an "automorphism@@cq", and further that any two distinct $h_1,h_2 \in H$ give rise to distinct "automorphisms@@cq" $h_1^{-1}\circ g$ and $h_2^{-1}\circ g$ of $q$. Hence, $\countAns_{q}(\D) \leq |\Auto{q}|$.

        \proofcase{$\geq$} Since $M$ is a "support", let $g : q \homto M$.
        Suppose there are $n$ distinct "automorphisms@@cq" $h_1, \dotsc, h_n : q \homto q$. Observe that each $(g \circ h_i): q \homto M$ is a distinct "homomorphism" to $M$, and thus that $\countAns_{q}(\D) \geq n= |\Auto{q}|$.
    \end{nestedproof}
\noindent  It is easy to see that all queries $q \in \Qneq$ are of quadratic size. Together with the preceding claims, this 
  shows that $\Qneq$ and numbers $\gamma_q:=\nicefrac{1}{|\Auto{q}|}$ have the required properties. 
\end{proof}
\begin{corollary}\label{cor:redux-to-SQL}
Computing $\wShapley_{\aUCQ}^{w}$ for a $\UCQneq$ query $\aUCQ$ can be achieved by evaluating simple and short (quadratic) `\texttt{select count(*)}' SQL queries in parallel.
\end{corollary}

While efficient enumeration of "minimal supports" guarantees data tractability of "WSMS" computation, 
in some cases counting may be tractable even %
when the exponential number of "minimal supports" makes their enumeration intractable, as exemplified by the following proposition.

   \begin{proposition}%
	   Let $w$ be a "tractable weight function" and $q$ a "regular path query".
	   Then, $\wShapley_{q}^{w}$ restricted to acyclic graph databases (\ie %
	   only binary relations and no directed cycles) is in $\FP$.
    \end{proposition}
   \begin{proof}
      The intuition is that, when $\+D$ is acyclic, the number of paths of any given size from $c$ to $d$ that are labelled by a word in $\+L$ can be computed in polynomial time by dynamic programming. We now turn into the formal proof.

      This proof will assume some measure of familiarity with deterministic finite word automata. We refer the interested reader to the first chapters of \cite{sakarovitchElementsAutomataTheory2009} for the relevant definitions.

   Let $\D$ be an acyclic database, $\+L(c,d)$ a "regular path query" and a tuple $\mu\in \D$. Without loss of generality, we can assume that $\Const(\D)=[n]$, and that the natural order of the vertices form a topological sort of $\D$, i.e. $\exists R. R(i,j)\in \D \Rightarrow i < j$. Given that $\+L$ is regular, we can take its minimal automaton $\+A^{\+L} = (Q^{\+L},\Sigma^{\+L},\delta^{\+L},q^{\+L}_0,F^{\+L})$.\footnote{$Q^{\+L}$ is the set of states, $\Sigma^{\+L}$ the alphabet, $\delta^{\+L}:Q^{\+L}\times \Sigma^{\+L}\to Q^{\+L}$ the transition function, $q^{\+L}_0$ the initial state and $F^{\+L}$ the set of final states; we abuse notations and extend the transition function to a function $\delta^{\+L}:Q^{\+L}\times (\Sigma^{\+L})^*\to Q^{\+L}$ over words.}

   Since $w$ is "tractable@@wf", \Cref{eq:wsms}
   can be applied in polynomial time provided we compute, for every $k\in[n-1]$, the number of paths from $c$ to $d$ labelled with a word in $\+L$ of size $k$. Note that, since $\D$ is acyclic, there is no need to worry about minimality. Formally, this is the number $sp_k^{c,q_0\to\mu\to d,F}$ of paths of length $k$ from $s$ to $t$ passing through $\mu$ that are labelled with a word $w$ such that $\delta^{\+L}(q_0,w)\in F$. Define for every $i,j\in V$, $p,q\in Q^{\+L}$ and $k\in[n-1]$ the number $sp_k^{i,p\to j,q}$ of paths of length $k$ from $i$ to $j$ such that $\delta^{\+L}(p,w)=q$. Since $\D$ is acyclic, these values follow the following recursive equation:
	   
   \begin{equation}\label{eq_prog_dyn}
      sp_k^{i,p\to j,q} =
      \begin{dcases}
	 \sum_{\substack{R(j',j)\in \D\\
		  \delta^{\+L}(q',R) = q\\
		  i\le j'}} sp_{k-1}^{i,p\to j',q'} \text{ if } i<j\\
	 1 \text{ if } i=j \text{ , } p=q \text{ and } k=0\\
	 0 \text{ otherwise}
      \end{dcases}
   \end{equation}
	   
	   From \Cref{eq_prog_dyn}, the $sp_k^{i,p\to j,q}$ can all be computed by dynamic programming, which allows the computation of the number of paths we want because it can be expressed as follows, with $\mu\eqdef R_\mu(\mu_s,\mu_t)$ (we sum for every possible lengths $k',k''$, intermediate state $q$ and final state $q_f$).
	   
	   \[sp_k^{c,q_0\to\mu\to d,F} = \sum_{k'+k''=k-1} \sum_{q\in Q} \sum_{q_f\in F} sp_{k'}^{c,q_0\to\mu_s,q}\times sp_{k''}^{\mu_t,\delta^{\+L}(q,R_\mu) \to d,q_f}\]
   \end{proof}

   On  arbitrary graphs containing cycles, deciding whether an edge lies on a simple path is already a known $\NP$-complete problem \cite[Theorem 2]{fortuneDirectedSubgraphHomeomorphism1980}, so no measure that satisfies \ref{Shdb:2} will be in $\FP$ for "RPQ"s, unless $\mathsf{P}=\NP$.
In fact, we shall prove a $\FP/\sP$-hard dichotomy (\Cref{thm:apq}) which shows that, for a wide class of "WSMS"s, their computation is $\sP$-hard for (almost) any unbounded "regular path query".

Before we prove this dichotomy, we need to show a lower-bound equivalent of \Cref{lem:wsmsvsmsc1}.
\AP We say a "weight function" $w$ is ""reversible"" if the total number of "minimal supports" for any $q$ in any database $\D$ can be computed in polynomial time given access to an oracle for $\Sh_{\wscorefun_{q,(\D,\emptyset)}^{w}}(\alpha)$, for all $\alpha\in \D$.\footnote{Observe that we do not need to compute the number of "minimal supports" of each size, but only the total number.} We shall abuse terminology and say the family $\Wscorefun^{w}$ itself is "reversible" when $w$ is.

   \begin{lemma}\label{lem:reversibleWSMS}
   The three families $\MSscorefun$, $\Sscorefun$ and $\SHARPscorefun$ considered thus far are all "reversible".
   \end{lemma}
   \begin{proof}
   \proofcase{$\MSscorefun$} By definition, $\wscorefun_{q,(\D,\emptyset)}^{w}(\D)$ is the total number of "minimal supports" we wish to compute. It can be obtained by simply taking the sum over all $\alpha\in \D$ of $\Sh_{\wscorefun_{q,(\D,\emptyset)}^{w}}(\alpha)$, by \ref{Sh:3}.\\

   \proofcase{$\Sscorefun$} Consider the binary representation of $\Sh_{\sscorefun_{q,(\D,\emptyset)}}(\alpha)$:

   \[\Sh_{\sscorefun_{q,(\D,\emptyset)}}(\alpha) = \sum_{k} 2^{-k|\D|}\countFMS_{q,\alpha}(k,\D)\]
   where $\countFMS_{q,\alpha}(k,\D)$
   is the number of "minimal supports" of $q$ in $\D$ of size $k$ that contain $\alpha$. For any $k$, $\countFMS_{q,\alpha}(k,\D)\le 2^{|\D|-1} < 2^{|\D|}$ because $2^{|\D|-1}$ is the total number of subsets of $\D$ that contain $\alpha$. This means that 
   $\countFMS_{q,\alpha}(k,\D)$ will be written starting from the bit of weight $2^{-k|\D|}$ and will never overflow on the bit of weight $2^{-(k-1)|\D|}$. Subsequently, each of the 
   $\countFMS_{q,\alpha}(k,\D)$ can directly be read out of the binary representation of $\Sh_{\sscorefun_{q,(\D,\emptyset)}}(\alpha)$. From there, the number 
   $\countFMS_q(k,\D)$ of "minimal supports" of size $k$ for $q$ in $\D$ can be computed as 
   $\countFMS_q(k,\D) =
   \nicefrac{1}{k}\sum_{\alpha\in\D}\countFMS_{q,\alpha}(k,\D)$, because each "minimal support" of size $k$ will be counted $k$ times, and the total number follows suit.\\

   \proofcase{$\SHARPscorefun$} Since $\Sh_{\sscorefun_{q,(\D,\emptyset)}}(\alpha) < 1$, the fractional part of $\Sh_{\sharpscorefun_{q,(\D,\emptyset)}}(\alpha)$ equals $\Sh_{\sscorefun_{q,(\D,\emptyset)}}(\alpha)$, from which the total number of "minimal supports" can be derived as shown above.
   \end{proof}
\AP
The next lemma shows how the problem $\intro*\evalCountMS_{q}$ of evaluating $\intro*\countMS_{q}$, i.e., counting, for the given database $\D$, the \emph{total} number of "minimal supports" of $q$ in $\D$,
can yield %
lower bounds for $\wShapley_{q}^{w}$. 

   \begin{lemma}\label{lem:wsmsvsmsc}
      For a "reversible" "weight function" $w$ and "Boolean" "monotone" "query" $q$, $\evalCountMS_{q} \polyrx \wShapley_{q}^{w}$.
   \end{lemma}
   \begin{proof}
   Let $w$ be a "reversible" "weight function". This means that the total number of "minimal supports" for $q$ in any $\+D$ can be computed in polynomial time given access to an oracle for $\wShapley_{q}^{w}$, which is precisely the definition of the Turing reduction $\evalCountMS_{q} \polyrx \wShapley_{q}^{w}$.
   \end{proof}

We now give the announced dichotomy for "RPQ"s, whose proof relies upon Lemmas \ref{lem:wsmsvsmsc1} and \ref{lem:wsmsvsmsc}.

   \begin{theorem}\label{thm:apq}
      Let $w$ be a "reversible" "tractable weight function", and a "regular path query" $q\defeq\+L(c,d)$. %
      Then 
      $\wShapley_{q}^{w} \in \FP$ if $\+L$ is finite or $\epsilon\in\+L$ and $c=d$, and $\sP$-hard otherwise.
   \end{theorem}
   \begin{proof}
   We shall reduce from the problem of counting simple paths between two distinguished vertices $v_s$ and $v_t$ in a directed graph, using the pumping lemma.

   First observe that if $\+L$ is finite then $q$ can be expressed as a "UCQ", and tractability follows from \Cref{th:data-ucq}%
   . If $\epsilon\in\+L$ and $c=d$, then $q$ is trivially always true.
	   
   Otherwise, we show $\sP$-hardness by a reduction from the problem of counting the number of simple paths in directed graphs. Formally, we have a directed graph $G=(V,E)$, two vertices $v_s,v_t\in V$, and we wish to compute the number of simple paths from $v_s$ to $v_t$ in $G$. This problem is known to be $\sP$-complete \cite[Problem 14]{valiantComplexityEnumerationReliability1979}. 
   We shall assume, without any loss of generality, that $v_s \neq v_t$. Indeed, if we wish to count the non-empty simple paths that start and end at some vertex $v$, we can split this vertex into a pair $v_s,v_t$ with all outgoing edges leaving $v_s$ and all incoming edges entering $v_t$ (formally, we have the edge $(v_s,u)$ for every $(v,u)\in E$ and the edge $(u,v_t)$ for every $(u,v)\in E$).

   Since $w$ is "reversible", there is a reduction from counting the number of "minimal supports" for $q$ in the input database to computing $\wShapley_{q}^{w}$, by \Cref{lem:wsmsvsmsc}.
   Leveraging this fact, we now build an input database that has as many "minimal supports" for $q$ as there are simple paths from $v_s$ to $v_t$ in $G$.

   Since $\+L$ is infinite and regular, the pumping lemma for regular languages (\aka star lemma)  \cite[Lemma 1.12]{sakarovitchElementsAutomataTheory2009} states that there exists a sublanguage $\alpha\beta^*\gamma \subseteq \+L$, with $\alpha,\beta,\gamma \in \Sigma^+$.\footnote{Note that we require that all three are non-empty.}
   To build our input database $\D_G$, we start by taking all vertices in $V$ as constants and adding, for each $(u,v)\in E$, a path labelled by $\beta$ from $u$ to $v$, using as many fresh intermediate constants as needed. We then add the distinguished constants $c$ and $d$, with a path labelled by $\alpha$ from $c$ to $v_s$ and one labelled by $\beta$ from $v_t$ to $d$. The final construction is depicted in \Cref{fig:const-apq}. 

   \begin{figure}[tb]
   \centering
      \begin{tikzpicture}
	\coordinate (00) at (-3, 0);
	\coordinate (01) at (-5, 0);
	\coordinate (02) at (-1, 0);
	\coordinate (03) at (-4, -0.75);
	\coordinate (04) at (-3.25, 0.5);
	\coordinate (05) at (-2, 0);
	\coordinate (06) at (4, 0);
	\coordinate (07) at (2, 0);
	\coordinate (08) at (6, 0);
	\coordinate (09) at (3, -0.75);
	\coordinate (010) at (3.75, 0.5);
	\coordinate (011) at (5, 0);
	\coordinate (012) at (0, 0);
	\coordinate (013) at (1, 0);
	\coordinate (014) at (7, 0);
	\coordinate (015) at (2.5, 1.25);
	\coordinate (016) at (-4.5, 1.25);
	\begin{pgfonlayer}{nodelayer}
		\draw [decorate,decoration={bumps, amplitude=1ex, segment length=4.5ex}] (00) circle [x radius=-2, y radius=1.2];
		\node [draw, circle,fill=white] (1) at (01) {$v_s$};
		\node [draw, circle,fill=white] (2) at (02) {$v_t$};
		\node [draw, circle] (3) at (03) {};
		\node [draw, circle] (4) at (04) {};
		\node [draw, circle] (5) at (05) {};
		\draw [decorate,decoration={bumps, amplitude=1ex, segment length=4.5ex}] (06) circle [x radius=-2, y radius=1.2];
		\node [draw, circle,fill=white] (7) at (07) {$v_s$};
		\node [draw, circle,fill=white] (8) at (08) {$v_t$};
		\node [draw, circle] (9) at (09) {};
		\node [draw, circle] (10) at (010) {};
		\node [draw, circle] (11) at (011) {};
		\node [] (12) at (012) {$\Rightarrow$};
		\node [draw, circle] (13) at (013) {$c$};
		\node [draw, circle] (14) at (014) {$d$};
		\node [] (15) at (015) {$G$};
		\node [] (16) at (016) {$G$};
	\end{pgfonlayer}
	\begin{pgfonlayer}{edgelayer}
		\draw [->,>=stealth] (1) to (4);
		\draw [->,>=stealth] (3) to (1);
		\draw [->,>=stealth] (5) to (2);
		\draw [->,>=stealth] (11) to node[midway, above] {$\beta$} (8);
		\draw [->,>=stealth] (7) to node[midway, above] {$\beta$} (10);
		\draw [->,>=stealth] (9) to node[near start, above] {$\beta$} (7);
		\draw [->,>=stealth] (13) to node[midway, above] {$\alpha$} (7);
		\draw [->,>=stealth] (8) to node[midway, above] {$\gamma$} (14);
	\end{pgfonlayer}
\end{tikzpicture}%
   
   \caption{Illustration of the graph database $\+D_G$. Each arrow represents a path with the given label.}
   \label{fig:const-apq}
   \end{figure}

   First recall that we have assumed that either $\epsilon\notin\+L$ or $c\neq d$, meaning that the "minimal supports" for $q$ must be paths of lenght at least 1. By construction, every simple path from $c$ to $d$ corresponds to a unique simple path from $v_s$ to $v_t$ in $G$, and it will be labelled by some word in $\alpha\beta^*\gamma\subseteq\+L$, meaning it will be a "minimal support" for $q$. Conversely every "minimal support" for $q$ has to be a path from $c$ to $d$ and by minimality this means that there are no other "minimal supports" than the simple paths. Overall, the number of "minimal supports" for $q$ in $\D_G$ matches the number of simple paths from $v_s$ to $v_t$ in $G$ as desired, implying that $\wShapley_{q}^{w}$ is $\sP$-hard.
   \end{proof}

   Note that the proof fails if we consider existentially quantified "RPQ"s. Indeed consider the query $q \eqdef \exists x. ~ R^+(x,x)$ whose "minimal supports" are the simple cycles of $R$-edges: if we build the above reduction for, say $\alpha=\beta=\gamma = R$, each path from $v_s$ to $v_t$ in $\D_G$ will still be a "minimal support" for $q$, but every simple cycle in $G$ will create an additional, unwanted, "minimal support" for $q$.
That being said, the problem of counting all simple cycles is $\sP$-hard as well \cite[Theorem 3.2]{yamamotoApproximatelyCountingPaths2017}, and more generally we have no reason to think that the problem is tractable for any "unbounded" query:
\begin{conjecture}
   Let $w$ be a "reversible" "tractable weight function", and let $q$ be a "homomorphism-closed"
   "Boolean" 
   query. Then $\wShapley_{q}^{w} \in \FP$ if $q$ is equivalent to a "UCQ" and $\sP$-hard otherwise.
\end{conjecture}

\section{Combined Complexity}
\label{sec:combined-complexity}

We will now study the combined complexity of computing "weighted sums of minimal supports".
We first give a general upper bound assuming some minimal conditions, namely that the query evaluation task lies within the polynomial hierarchy and the "weight function" is tractable.

\begin{proposition}
   Let $\+C$ be a class of "Boolean" "monotone" queries for which the (combined) query evaluation %
   is in $\PH$, and $w$ be a "tractable weight function". Then $\wShapley_{\C}^{w}\in\FPsP$.
\end{proposition}

\begin{proof}
   \AP With \Cref{lem:wsmsvsmsc1} in mind, our first step is to show that $\countFMS_{\+C} \in \sPH$, where $\intro*\sPH$ is the class of all computational problems that can be defined as the number of accepting runs of a polynomial-time nondeterministic Turing machine (PNTM) that has access to a $\PH$ oracle \cite{hemaspaandraSatanicNotationsCounting1995}. We thus construct a PNTM with $\PH$ oracle that, on any input query $q\in\+C$ database $\D$, has exactly one accepting run for every "minimal support" of size $k$ for $q$ in $\D$.
   For this, we non-deterministically pick a subset $S$ of $\D$ and then use the $\PH$ oracle to check if it is a "minimal support". Since answering $q$ is in $\PH$, a first call can decide if $S\models q$ and a second one that for all $S'\subseteq S$, $S'\not\models q$.

Now that we have shown $\countFMS_{\+C} \in \sPH$, \Cref{lem:wsmsvsmsc} implies $\wShapley_{\C}^{w} \in \FPsPH$, and we conclude by the fact that $\FPsPH = \FPsP$ \cite[Theorem 4.1]{todaPolynomialtime1TuringReductions1992}.
\end{proof}

In the following subsections, we identify conditions that determine whether $\wShapley^{w}$ is tractable or intractable for certain classes of conjunctive queries.

\subsection{Hardness Results for (Unions of) Conjunctive Queries}
\label{sec:hardness-(U)CQs}
\color{black}

It is known that counting "homomorphisms" from "CQs" (or, equivalently, counting the answers to non-Boolean "CQs" without any existential variable, \aka\ `full "CQs"') is $\sP$-complete in general \cite[Theorem 1.1]{DyerG00},
but it becomes tractable for "ACQs" \cite[Theorem 1]{PichlerS13} and more generally for classes of 
"CQs" having bounded `"generalized hypertree width"' \cite[Theorem 7]{PichlerS13}.
However, as the next example illustrates, counting "homomorphisms" is not the same as counting "minimal supports":
\begin{example}\label{ex:supports-vs-answers}
    Consider the "CQ" $\exists xyzw ~ R(x,y) \land R(y,z) \land R(z,w)$ and the "database" $\D= \set{R(a,b),$ $R(b,c), R(c,a)}$. 
    While the query has only one (minimal) "support" (\ie $\D$ itself), it has three different "homomorphisms" characterized by their values on $x,y,z,w$: $\set{(a,b,c,a), (b,c,a,b), (c,a,b,c)}$.
\end{example}

Interestingly, we show here that the aforementioned tractability results for counting answers to "ACQs" do not generalize to counting "minimal supports". In fact, we prove that counting "minimal supports" for "ACQs" is $\sP$-hard even in the absence of constants and over a fixed "schema".

The next results will be shown by reduction from the $\AP\intro*\numPerMatch$ problem of, given a bipartite graph $G=(V, E)$, counting the number of subsets $P \subseteq E$ (called \AP""perfect matchings"") such that each vertex of $V$ is contained in exactly one edge of $P$. This problem is known to be $\sP$-complete \cite[Problem 2]{valiantComplexityEnumerationReliability1979} under polynomial-time "Turing reductions", and was later shown to be hard even under polynomial-time "$1$-Turing reductions" \cite{Zanko91}.

\begin{theorem}\label{thm:countMS-ACQ-shP}
    $\evalCountMS_{\ACQ}$ is $\sP$-hard under polynomial-time "$1$-Turing reductions", even in the absence of constants.
    Hence, $\evalCountFMS_{\ACQ}$ is $\sP$-hard under polynomial-time "Turing reductions".
\end{theorem}
\begin{proof}
   We reduce from $\numPerMatch$ for a given bipartite graph $G=(V, E)$. Suppose $V$ is partitioned into $V_1 \cup V_2$ and assume $V_1 = \set{a_1, \dotsc, a_n}$, $V_2 = \set{b_1, \dotsc, b_n}$ and denote $I=\set{1,\dotsc,n}$ (observe that if $|V_1| \neq |V_2|$ there are no "perfect matchings").
   We will build a "database" $D_G$ and a constant-free "ACQ" $q_G$ such that the number of "perfect matchings" on $G$ is equal to $\countMS_{q_G}(D_G) - n$.
    For clarity, we first show a reduction which uses an unbounded relational "schema" $\sigma_G =  \set{R_1, \dotsc, R_n}$ where each $R_i$ is binary.

    \proofcase{$D_G$} We shall abuse notation and use the elements of $V_1 \cup V_2$ as constant names.
    The "database" $D_G$ is built over the domain $V_1 \cup V_2 \cup \set{c_{i,v,v'} : i \in [n-1], v,v' \in V_1 \cup V_2}$. The $c$-constants are there only to build some paths between the elements of $V_1 \cup V_2$ reading `$R_1 \dotsb R_n$' or `$R_n \dotsb R_1$'. \Cref{fig:thm:countMS-ACQ-shP} contains an example from which the reader can deduce how $D_G$ is defined (the $c$-constants are implicit in the figure). More formally, we define $D_G$ as the smallest set containing:
    \begin{itemize}
        \item All "facts" $R_1(v,v), \dotsc, R_n(v,v)$ for every $v \in V_2$.
	\item A path reading $R_n, R_{n-1}, \dotsc, R_1$ from $a_i$ to $a_{i+1}$ for every $i$ ---formally, all "facts" $R_n(a_i,c_{1,a_i,a_{i+1}})$, $R_{n-1}(c_{1,a_i,a_{i+1}},c_{2,a_i,a_{i+1}}), \dotsc,$ $R_1(c_{n-1,a_i,a_{i+1}}, a_{i+1})$ for every $i \in [n-1]$.
        \item A path reading $R_1, R_2, \dotsc, R_n$ between any $a_i$ and $b_j$ connected by an edge in $G$ ---formally, all
        "facts" $R_1(a_i,c_{1,a_i,b_j}),$ $R_2(c_{1,a_i,b_j},c_{2,a_i,b_j}), \dotsc,$ $R_n(c_{n-1,a_i,b_j}, b_j)$ for every $\set{a_i, b_j} \in E$.
    \end{itemize}

    \proofcase{$q_G$} The "ACQ" $q_G$ is defined as in \Cref{fig:thm:countMS-ACQ-shP}.
    We give names only to some of the variables of $q_G$: $X=\set{x_1, \dotsc, x_n}$ which will be mapped to the $a_i$'s, and $Y = \set{y_1, \dotsc, y_n}$ which will be mapped to the $b_i$'s. All others will be mapped to intermediate constants along paths in $\D_G$.%

It is easy to see that the $\+S_j\defeq{R_I(b_j,b_j)}$ (\ie ${R_1(b_j,b_j),\dots,R_n(b_j,b_j)}$) are all "minimal supports" for $q_G$ in $\+D_G$. Similarly, if we take a "perfect matching" in $G$ under the form of a bijective $\mu:V_1\to V_2$ \st $\forall j\in I. (a_j,\mu(a_j))\in E$, then we can easily build a "minimal support" for $q_G$ by taking an isomorphic image of $q_G$ with $(x_j,y_j)\mapsto (a_j, \mu(a_j))$. We now show that there is no "minimal support" for $q_G$ in $\+D_G$ other than these.%

Assume that $\+S$ is a "minimal support" for $q_G$ in $\+D_G$ that doesn’t contain any $\+S_j$. Since $\+S\vDash q$, then $\+S$ must be the homomorphic image of $q_G$ by some homomorphism $h$. First consider the spine $x_1 \xrightarrow{R_{n\dots 1}} \dots \xrightarrow{R_{n\dots 1}} x_n$ of $q_G$. It must necessarily map to $a_1 \xrightarrow{R_{n\dots 1}} a_n$ because that is the only path in $\+D_G$ that doesn’t contain any $\+S_j$ and is homomorphic to $x_1 \xrightarrow{R_{n\dots 1}} \dots \xrightarrow{R_{n\dots 1}} x_n$. In other words $\forall j\in I. h(x_j)=a_j$.%

Now consider a branch $x_i \xrightarrow{R_{1\dots n}} y_i$ of $q_G$. Since $h(x_i) = a_i$ and $x_i$ has at most two outgoing edge, only one of which is a $R_1$, the branch must necessarily map to some $a_i \xrightarrow{R_{1\dots n}} b_j$, with $(a_i,b_j)\in E$. In other words, $h(y_j)$ is always a neighbour of $a_j$ in $G$. Moreover, if we were to have $h(y_i) = h(y_{i'}) = b_j$ with $i\neq i'$, then $\+S$ would necessarily contain ${R_{I\setminus\{i\}}(b_j,b_j)} \cup {R_{I\setminus\{i'\}}(b_j,b_j)} = \+S_j$, which is impossible. Overall, $\+S$ must therefore induce a "perfect matching" as described above, with $\mu(a_j)\defeq h(y_j)$.

    \begin{figure}
	\resizebox{\columnwidth}{!}{%
      \begin{tikzpicture}[]
	\coordinate (00) at (-8, -0.25);
	\coordinate (01) at (-8, 1.75);
	\coordinate (02) at (-8, 2.75);
	\coordinate (03) at (-6.5, 2.75);
	\coordinate (04) at (-6.5, 1.75);
	\coordinate (05) at (-6.5, -0.25);
	\coordinate (06) at (-6, 3.25);
	\coordinate (07) at (-8.5, -0.75);
	\coordinate (08) at (-6, 3.25);
	\coordinate (09) at (-8.5, 2.25);
	\coordinate (010) at (-8, 0.75);
	\coordinate (011) at (-6.5, 0.75);
	\coordinate (012) at (-3.5, 3.5);
	\coordinate (013) at (-1.5, 3.5);
	\coordinate (014) at (-3.5, 2);
	\coordinate (015) at (-1.5, 2);
	\coordinate (016) at (-3.5, 0.5);
	\coordinate (017) at (-1.5, 0.5);
	\coordinate (018) at (-3.5, -1);
	\coordinate (019) at (-1.5, -1);
	\coordinate (020) at (-5, -1.5);
	\coordinate (021) at (0, 4);
	\coordinate (022) at (-5, 2.75);
	\coordinate (023) at (3, 3.5);
	\coordinate (024) at (5, 3.5);
	\coordinate (025) at (3, 2);
	\coordinate (026) at (5, 2);
	\coordinate (027) at (3, 0.5);
	\coordinate (028) at (5, 0.5);
	\coordinate (029) at (3, -1);
	\coordinate (030) at (5, -1);
	\coordinate (031) at (1, -1.5);
	\coordinate (032) at (7, 4);
	\coordinate (033) at (1, 2.75);
	
	\begin{pgfonlayer}{nodelayer}
		\node [draw, circle] (0) at (00) {$a_n$};
		\node [draw, circle] (1) at (01) {$a_2$};
		\node [draw, circle] (2) at (02) {$a_1$};
		\node [draw, circle] (3) at (03) {$b_1$};
		\node [draw, circle] (4) at (04) {$b_2$};
		\node [draw, circle] (5) at (05) {$b_n$};
		\node [] (6) at (06) {};
		\node [] (7) at (07) {};
		\node [] (8) at (08) {};
		\node [fill=white] (9) at (09) {$G$};
		\node [circle, fill=none] (10) at (010) {\dots};
		\node [circle, fill=none] (11) at (011) {\dots};
		\node [draw, circle, text=black] (12) at (012) {$a_1$};
		\node [draw, circle, text=black] (13) at (013) {$b_1$};
		\node [draw, circle, text=black] (14) at (014) {$a_2$};
		\node [draw, circle, text=black] (15) at (015) {$b_2$};
		\node [circle, fill=none] (16) at (016) {\dots};
		\node [fill=none] (17) at (017) {\dots};
		\node [draw, circle, text=black] (18) at (018) {$a_n$};
		\node [draw, circle, text=black] (19) at (019) {$b_n$};
		\node [] (20) at (020) {};
		\node [] (21) at (021) {};
		\node [fill=white] (22) at (022) {$\+D_G$};
		\node [draw, circle, text=black] (23) at (023) {$x_1$};
		\node [draw, circle, text=black] (24) at (024) {$y_1$};
		\node [draw, circle, text=black] (25) at (025) {$x_2$};
		\node [draw, circle, text=black] (26) at (026) {$y_2$};
		\node [circle, fill=none] (27) at (027) {\dots};
		\node [fill=none] (28) at (028) {\dots};
		\node [draw, circle, text=black] (29) at (029) {$x_n$};
		\node [draw, circle, text=black] (30) at (030) {$y_n$};
		\node [] (31) at (031) {};
		\node [] (32) at (032) {};
		\node [fill=white] (33) at (033) {$q_G$};	
	\end{pgfonlayer}
	\begin{pgfonlayer}{edgelayer}
		\draw (1) to (5);
		\draw (0) to (4);
		\draw (2) to (3);
		\draw [dashed] (8) rectangle (7);
		\draw [dashed] (20) rectangle (21);
		\draw [->,>=stealth] (12) to node[midway, above, sloped] {$R_{1\dots n}$}(13);
		\draw [->,>=stealth] (14) to node[near start, above, sloped] {$R_{1\dots n}$}(19);
		\draw [->,>=stealth] (18) to node[near start, above, sloped] {$R_{1\dots n}$}(15);
		\draw [->,>=stealth] (12) to node[midway, left] {$R_{n\dots 1}$}(14);
		\draw [->,>=stealth] (14) to node[midway, left] {$R_{n\dots 1}$}(16);
		\draw [->,>=stealth] (16) to node[midway, left] {$R_{n\dots 1}$}(18);
		\draw [<-,>=stealth, loop right] (19) to node[midway] {$R_I$}();
		\draw [<-,>=stealth, loop right] (15) to node[midway] {$R_I$}();
		\draw [<-,>=stealth, loop right] (13) to node[midway] {$R_I$}();
		\draw [dashed] (31) rectangle (32);
		\draw [->,>=stealth] (23) to node[midway, above, sloped] {$R_{1\dots n}$} (24);
		\draw [->,>=stealth] (23) to node[midway, left] {$R_{n\dots 1}$}(25);
		\draw [->,>=stealth] (25) to node[midway, left] {$R_{n\dots 1}$}(27);
		\draw [->,>=stealth] (27) to node[midway, left] {$R_{n\dots 1}$}(29);
		\draw [<-,>=stealth, loop right] (30) to node[midway] {$R_{I\setminus\{n\}}$}();
		\draw [<-,>=stealth, loop right] (26) to node[midway] {$R_{I\setminus\{2\}}$}();
		\draw [<-,>=stealth, loop right] (24) to node[midway] {$R_{I\setminus\{1\}}$}();
		\draw [->,>=stealth] (25) to node[midway, above, sloped] {$R_{1\dots n}$} (26);
		\draw [->,>=stealth] (29) to node[midway, above, sloped] {$R_{1\dots n}$} (30);
		\draw (2) to (4);
		\draw [->,>=stealth] (12) to node[midway, above, sloped] {$R_{1\dots n}$} (15);
	\end{pgfonlayer}
\end{tikzpicture}%
   }
	\caption{Construction for the reduction in the proof of \Cref{thm:countMS-ACQ-shP}. $\xrightarrow{R_{1\dots n}}$ (resp. $\xrightarrow{R_{n\dots 1}}$) is an abbreviation for the path $\xrightarrow{R_1}\dots\xrightarrow{R_n}$ (resp. $\xrightarrow{R_n}\dots\xrightarrow{R_1}$), and $\xrightarrow{R_J}$ means a parallel edge for every $j\in J$.}
        \label{fig:thm:countMS-ACQ-shP}
    \end{figure}

    One shortcoming of the reduction above is that we have used a "schema" which depends on the problem instance. However, this was done for the sake of clarity of exposition, and we show how to modify the reduction above to use a \emph{bounded} "schema" $\sigma = \set{S, T, R}$, where $S$ is unary and $T,R$ are binary. For this, we replace everywhere (in the query and in the database) each appearance of $R_i(t,t')$ with a path from $t$ to $t'$ reading $T R^i$ (by using fresh $i$ new constants or variables). The proof follows the same lines.
\end{proof}

Observe how the hardness proof above uses, in a crucial way, the fact that "acyclic" queries may have many "atoms" with the same "relation name".
\AP
As we shall see in \Cref{ssec:counting-homs-ms-char}, in fact the hardness result no longer holds when one restricts to the class $\intro*\sjfCQ$ of ""self-join free"" "CQ"s, which consist of all "CQ"s which do not contain two distinct "atoms" with the same "relation name". 
However, as soon as we allow \emph{unions} of such queries, hardness is recovered, as shown next.

\begin{theorem}[{\cite[Proposition 13]{ourKR25}}]\label{prop:UsjfACQ-MS-hard}
    $\evalCountMS_{\UsjfACQ}$ is $\sP$-hard under polynomial-time "$1$-Turing reductions" for the class \AP$\intro*\UsjfACQ$ of unions of "self-join free" "acyclic" "CQs". This holds even on binary signatures and for queries without constants.
\end{theorem}
\begin{proof}
    This is an adaptation of the $\sP$-hardness proof of \cite[Theorem 6]{PichlerS13} for counting the number of answers of unions of acyclic full conjunctive queries (where `full' means that queries have no existentially quantified variables). The adaptation will also address some extra restrictions, namely: (i) accounting for counting "minimal supports" instead of answers, (ii) ensuring that the queries are "self-join free", and (iii) working on binary signatures instead of ternary. While this makes the reduction considerably more technical, we stress that the underlying idea remains that of Pichler and Skritek.

    \newcommand{\xorina}{\textit{OR!-in}^{\text{a}}}%
    \newcommand{\xorinb}{\textit{OR!-in}^{\text{b}}}%
    \newcommand{\xorout}{\textit{OR!-out}}%
    \newcommand{\zero}{\textit{Zero}}%
    \newcommand{\one}{\textit{One}}%
    \newcommand{\rneg}{\textit{Not}}%
    \newcommand{\qval}{\textit{Val}}%

    Let the graph $G=(V, E)$ be an instance of $\numPerMatch$. We will build a "database" $\+D$ and constant-free queries $q_1,q_2 \in \UsjfACQ$ such that the number of "perfect matchings" on $G$ is equal to $\countMS_{q_1}(\+D) - \countMS_{q_2}(\+D)$. In fact, $q_1$ is an ("acyclic", "self-join free") "CQ" rather than a "UCQ", and its evaluation is in polynomial time by \Cref{cor:countFMS-sjfACQ-ptime}. Hence, this is a "$1$-Turing reduction".
    
    Suppose $V$ is partitioned into $A \cup B$ and assume $A = \set{a_1, \dotsc, a_n}$, $B = \set{b_1, \dotsc, b_n}$, and $E \subseteq A \times B$ (observe that if $|A| \neq |B|$ there are no "perfect matchings").
    For $a_i \in A$, let $B_i = \set{b_j \in B : (a_i,b_j) \in E}$ and we define $A_j$ for $b_j \in B$ analogously.
    The signature for the queries $q_1,q_2$ has binary predicates $\xorina_{i,j}, \xorinb_{i,j},\xorout_{i,j}$, and unary predicates $\zero_{i,j}$, $\one_{i,j}$, $V_{i,j}$ for every $1 \leq i,j \leq n$.
    The database $\+D$ has 5 constants $\set{0,1,00,01,10}$ and consists of the following set of "facts" 
    \begin{align*}
      \xorina_{i,j}(0,00), \xorina_{i,j}(0,01),
      \xorina_{i,j}(1,10), \\ %
      \xorinb_{i,j}(0,00), \xorinb_{i,j}(0,10), 
      \xorinb_{i,j}(1,01), \\
      \xorout_{i,j}(00,0), \xorout_{i,j}(01,1), \xorout_{i,j}(10,1), \\ %
      V_{i,j}(0), V_{i,j}(1), \one_k(1), \zero_k(0)
    \end{align*}
    for every $i,j,k \in [n]$.  
    The $\xorina, \xorinb, \xorout$ "facts" encode part of the `OR' truth table (with two `\emph{in}' arguments and one `\emph{out}' result) in the following sense: The evaluation of the  "CQ" $q_{i,j}^{\mathsf{OR!}}(x,y,z) \defeq \exists w ~ \xorina_{i,j}(x,w) \land \xorinb_{i,j}(y,w) \land \xorout_{i,j}(w,z)$ on $\+D$  is, precisely, the OR truth table without the $(1,1,1)$ line. %
    That is, whenever both arguments $x,y$ are `true', $q_{i,j}^{\mathsf{OR!}}$ doesn't encode the output `true' nor `false' (the constants 1/0) it simply does not hold, which can be thought of as a blocking error state.
   This is done in such a way that the $n$-ary `exists unique' operator $\exists!$ can be defined by chaining the binary operator $\mathsf{OR!}$ (in the same way that chaining the standard $\mathsf{OR}$ yields $\exists$).
    Formally, the evaluation of the "acyclic" and "self-join free" "CQ"\AP\phantomintro{\poneone}:%
    \begin{multline*}
      \poneone{i}(x_0, \dotsc,x_m) \defeq \exists z_1,\dotsc,z_m ~ q_{i,1}^{\mathsf{OR!}}(x_0,x_1,z_1) \land {}\\{}\land q_{i,2}^{\mathsf{OR!}}(z_1,x_2,z_2) \land 
      \dotsb \land q_{i,n}^{\mathsf{OR!}}(z_{n-1},x_m,z_m) \land \one_i(z_m)
    \end{multline*}
    for any $m\leq n$
     on $\+D$ returns all $m$-tuples $\vect a \in  \set{0,1}^m$ having exactly one component equal to 1.%
  
    Let $\+X$ be the set of variables $x_{i,j}$ such that $(a_i,b_j) \in E$, and let $\+X_{i,*} \defeq \set{x_{i,j} : (a_i,b_j) \in E}$, $\+X_{*,j} \defeq \set{x_{i,j} : (a_i,b_j) \in E}$.
    The valuations of these $\+X$ variables will denote subsets of $E$: if $x_{i,j}$ is assigned $1$ it is in the subset and otherwise (we will make sure that it is assigned $0$) it is not in the subset.
    To be able to recover the assignment form the "minimal supports", we will use the predicates $V_{i,j}$: $\qval \defeq \bigwedge_{x_{i,j} \in \+X} V_{i,j}(x_{i,j})$.
    We now define 
    \[
    q_1 \defeq \qval \land \bigwedge_{i\in [n]} \poneone{i}(\vect x_i).
    \]
    where each $\vect x_i$ is the tuple of variables containing $\+X_{i,*}$.
    Each "minimal support" of $q_1$ on $\+D$ corresponds to a set of edges $E' \subseteq E$ (given by the $V_{i,j}$-facts assigned $1$) such that each vertex from $A$ appears in exactly one edge from $E'$; further the number of "minimal supports" coincides with the number of such set of edges. Observe that $q_1$ is "self-join free" and "acyclic", hence $q_1 \in \UsjfACQ$.
    However, some of these $E'$ subsets may not be "perfect matchings", since there may be some $b_j$ which does not appear in any edge of $E'$.
    We shall now build a query $q_2$ whose number of "minimal supports" coincides with the number of such $E'$ subsets which are not "perfect matchings".

    For each $j$ let $p_j$ be the query stating that none of the variables from $\+X_{*,j}$ is assigned $1$, that is,
    \begin{align*}
      p_j &\defeq   \zero_1(x_{i_1,j}) \land \dotsb \land \zero_k(x_{i_k,j})
    \end{align*}
    where $\+X_{*,j} = \set{x_{i_1,j}, \dotsc, x_{i_k,j}}$.
    Note that each "minimal support" of $p_j \land q_1$ corresponds to such a set of edges $E' \subseteq E$ which additionally verifies that there are no edges incident to $b_j$. Further, observe that such formula is "self-join free". What is more, the number of "minimal supports" of 
    \[
        q_2 \defeq \bigvee_{j\in [n]} (p_j \land q_1)
    \]
    corresponds to the number of subsets $E' \subseteq E$ such that (i) each vertex from $A$ appears in exactly one edge from $E'$ and (ii) there exists at least one vertex $b_j \in B$ which is not adjacent to any $E'$. Observe that $q_2$ is "self-join free" since each $p_j$ is "self-join free" and there are no predicates in common between $p_j$ and $q_1$, and it is further "acyclic".
  
    It then follows that the number of "minimal supports" of $q_1$ minus the number of "minimal supports" of $q_2$ gives the number of "perfect matchings" of $G$. 
\end{proof}

We can obtain the following result as a direct consequence of \Cref{lem:wsmsvsmsc}, \Cref{lem:reversibleWSMS},  \Cref{thm:countMS-ACQ-shP}, and \Cref{prop:UsjfACQ-MS-hard}.

\begin{corollary}
    $\wShapley_{\ACQ}^{w}$ and $\wShapley_{\UsjfACQ}^{w}$ are $\sP$-hard for each "reversible" "weight function" $w$. In particular for $\msShapley_{\ACQ}$ and $\msShapley_{\UsjfACQ}$. %
\end{corollary}

  \subsection{Counting Minimal Supports via Counting Homomorphisms}
  \label{ssec:counting-homs-ms-char}
  \AP In order to identify some tractable cases, we will turn our attention to queries for which counting the number of "homomorphisms" coincides with counting the number of "minimal supports".
\AP To this end, recall that 
$\countAns_q$ is the numeric query which counts the number of 
"homomorphisms" $q\homto \D$. We will 
use $\intro*\evalCountAns_{\class}$ to refer to the associated 
evaluation problem for a given class $\class$ of "CQs".

We observe that for every "CQ" $q$ and database $\D$, each "homomorphism" $h : q \homto \D$ \AP""induces@@support"" the "support" $\set{h(\alpha) : \alpha \text{ "atom" of } q}$, hence $\countMS_{q}(\D) \leq \countAns_{q}(\D)$.

   \begin{lemma}
       For every "CQ" $q$ and "database" $D$, $\countMS_{q}(D) \leq \countAns_{q}(D)$.
   \end{lemma}
   \begin{proof}
       Each "minimal support" gives rise to a different "homomorphism" (possibly many, as in \Cref{ex:supports-vs-answers}), but two distinct "minimal supports" cannot be the image of the same "homomorphism".
   \end{proof}

   The hardness result of \Cref{thm:countMS-ACQ-shP} relies crucially in the use of queries having "atoms" with the same "relation name" (and thus not "self-join free"). We first observe that over classes of "self-join free" queries counting "homomorphisms" is the same as counting "minimal supports".
    
  \begin{remark}\label{rk:sjfCQ-MS-Ans}
    For every "self-join free" "CQ" $q$ and "database" $\D$, $\countAns_{q}(\D) = \countMS_{q}(\D) = \countFMS_{q}(k,\D)$, where $k = |\atoms(q)|$ (in particular, $\countFMS_{q}(k',\D) = 0$ for any $k' \neq k$).
  \end{remark}
  Since $\evalCountAns_{\sjfACQ}$ is tractable for the class \AP$\intro*\sjfACQ$ of "acyclic" "self-join free" "CQs"\phantomintro{sjfACQ} \cite[Theorem 1]{PichlerS13}, we obtain as a corollary:
  \begin{corollary}\label{cor:countFMS-sjfACQ-ptime}
    $\evalCountFMS_{\sjfACQ}$ is in polynomial time.
  \end{corollary}
  This means that we can reuse the efficient divide-and-conquer algorithms exploiting the tree-like structure of queries that have been previously devised for counting "homomorphisms". Unfortunately, this tractability does not extend to \emph{unions} of "sjfACQs" as was shown in \Cref{prop:UsjfACQ-MS-hard}.

  Of course, tractability for counting "minimal supports" of "conjunctive queries" can be established not only for "acyclic" "self-join free" "CQs" but more generally  for any class $\+C$ of "CQs" with tractable $\evalCountAns_{\+C}$ problem, as soon as it verifies that the number of "homomorphisms" coincides with the number of "minimal supports" (as in \Cref{rk:sjfCQ-MS-Ans}). 
  But then: \emph{Which queries satisfy such bijection between "homomorphisms" and "minimal supports"?}
  
  We shall next prove that this bijection holds, in particular, for the class of queries having no two `unifiable' atoms, where two atoms $\alpha, \alpha'$ are \AP""unifiable"" if there is a function $f : \vars(\alpha) \cup \vars(\alpha') \to \Const$ such that $f(\alpha) = f(\alpha')$, where such function is called a ""unifier"". 
  For example, the "atoms" $R(x,a)$ and $R(b,x)$ (where $a,b$ are distinct "constants", and $x$ is a "variable") are not "unifiable" and neither are $S(x,y,x,y)$ and $S(z,z,a,b)$. On the other hand, $T(x,x,a)$ and $T(x,y,x)$ are "unifiable" via the "unifier" mapping both variables to $a$.
  A "CQ" is \reintro{unifiable} if it contains two distinct "unifiable" "atoms".
  In particular, no "self-join free" query is "unifiable". Observe that this is a tractable property, ensuring that all "induced@@support" "supports" are of equal size, and we further show that this class verifies the bijection between "homomorphisms" and "minimal supports".
  \begin{remark}\label{rk:unifiable-tractable}
      The problem of testing whether a "CQ" is "unifiable" is in polynomial time (\eg\ by using any unification algorithm, such as Robinson's \cite{robinson1965machine}).
  \end{remark}
  \begin{lemma}\label{lem:unifiable-size}
  A "CQ" $q$ is non-"unifiable" if, and only if, every "induced@@support" "support" of $q$ is of size $|\atoms(q)|$.
\end{lemma}
\begin{proof}
  From left to right, if there is a "homomorphism" $h: q \homto D$ "inducing@@support" a "support" with less than $|\atoms(q)|$ "facts", then there must be two "atoms" $\alpha,\beta \in \atoms(q)$ such that $h(\alpha)=h(\beta)$, and thus $\alpha,\beta$ are "unifiable", contradicting the hypothesis.
  
  From right to left, if there are two "atoms" "unifiable" via some "unifier" $f$, then $f$ can be extended to a "homomorphism" "inducing@@support" a "support" of size strictly smaller than $|\atoms(q)|$.
\end{proof}

\begin{proposition}\label{prop:nonunif-implies-minsups-eq-nrhoms}
  For any non-"unifiable" "CQ" $q$, the number of "homomorphisms" $q \homto D$ equals the number of "minimal supports" of $q$ in $D$.
\end{proposition}
\begin{proof}
  Note that the statement trivially holds for any query having only one "atom".
  Suppose now that $q$ has two (non-"unifiable") "atoms", \ie $q=\exists \vect x ~ \alpha \land \beta$.
  By means of contradiction, let us suppose that there are two distinct "homomorphisms" $h_1,h_2 : q \to D$ "inducing@@support" the same "support". That is, $h_1(\alpha)=h_2(\beta)$ and $h_1(\beta)=h_2(\alpha)$. In particular, $\alpha$ and $\beta$ are over the same "relation name". We shall show that this implies that $\alpha,\beta$ are "unifiable".
  
  Let two "terms" $t,t' \in \qterms(q)$ be \AP""$1$-related"" if $\set{t,t'} = \set{\alpha[i],\beta[i]}$ for some $i$. Let two "terms" $t,t'$ be \reintro{$(n+1)$-related} if for some "term" $\hat t$ we have that $t,\hat t$ are "$1$-related" and $\hat t,t'$ are "$n$-related". We call two "terms" simply \reintro{related} if they are "$n$-related" for some $n$. 
  Let $\AP\intro*\RelatedTerms$ be the set of all pairs of "related" terms.
  
  \begin{claim}\label{claim:term-constant-coincide}
    Given a "term" $t$ and a "constant" $c$ such that $(t,c) \in \RelatedTerms$, we have $h_1(t)=h_2(t)=c$.
  \end{claim}
  \begin{nestedproof}
      We prove this by induction on $n$ where $t,c$ are "$n$-related".

      \proofcase{Base case.} 
      If $n=1$, then either (i) $\alpha[i]=c$ and $\beta[i]=t$ or (ii) $\alpha[i]=t$ and $\beta[i]=c$. Without any loss of generality let us suppose (i) holds. Since $h_1(\alpha)=h_2(\beta)$, we have $h_2(t)=c$. Since $h_1(\beta)=h_2(\alpha)$, we have $h_1(t)=c$.

      \proofcase{Inductive step.} 
      If $n>1$, then $t,\hat t$ are "$1$-related" and $\hat t,c$ are "$(n-1)$-related" for some $\hat t$. 
      Then we have, for some $i$, that either $\alpha[i] = t$ and $\beta[i]=\hat t$, or $\alpha[i] = \hat t$ and $\beta[i]=t$. Let us assume the former (\ie $\alpha[i] = t$ and $\beta[i]=\hat t$) without any loss of generality.
      Since $h_1(\alpha)=h_2(\beta)$ we have $h_1(t)=h_2(\hat t)$, 
      and since $h_1(\beta)=h_2(\alpha)$ we have $h_1(\hat t)=h_2(t)$.
      By inductive hypothesis on $\hat t,c$ we have that $h_1(\hat t)=h_2(\hat t)=c$. Hence, $h_1(t)=h_2(t)=c$.
  \end{nestedproof}
  
  For each equivalence class $\sim$ of $\RelatedTerms$ containing no "constants" (\ie only variables), consider a fresh constant $c_\sim$.
  We shall now define a "unifier" $f$ for $\alpha,\beta$, thus contradicting the non-"unifiable" hypothesis.
  For any "variable" $x$ belonging to an equivalence class $\sim$ of $\RelatedTerms$ containing no "constant", let $f(x) \defeq c_\sim$.
  For any other "variable" $y$ we define $f(y)\defeq h_1(y)$;
  observe that in this case, by the previous \Cref{claim:term-constant-coincide}, we must also have $h_1(y)=h_2(y)$.
  Since
   $h_1(\alpha) =h_2(\beta)$ and $h_1(\beta) = h_2(\alpha)$, 
   by the definition of $\RelatedTerms$, it follows that 
   $f(\alpha)=f(\beta)$ and $f(\beta)=f(\alpha)$: indeed the $i$-th term $t$ of $\alpha$ must be sent via $f$ to the same constant as the $i$-th term $t'$ of $\beta$ since $t,t'$ belong to the same equivalence class of $\RelatedTerms$.
  Thus, $f(\alpha)=f(\beta)$ concluding that $\alpha,\beta$ are "unifiable".

  \medskip
  
  Now suppose $q$ has more than two (pairwise non-"unifiable") "atoms": $q = \exists \vect x ~ \gamma_1\land \dotsb \land \gamma_n$ for $n>2$.
  For a "homomorphism" $h$ we shall write $h(\gamma_1 \dotsb \gamma_n)$ to denote $h(\gamma_1) \dotsb h(\gamma_n)$.
  Again, suppose by means of contradiction that there are two distinct "homomorphisms" $h_1,h_2$ such that $h_1(\gamma_1 \dotsb \gamma_n) = h_2(\gamma_{i_1} \dotsb \gamma_{i_n})$ for a non-identitary permutation $i_1, \dotsc, i_n$ of $\set{1,\dotsc, n}$. Here we are using that "induced@@support" "supports" have size $|\atoms(q)|$ by \Cref{lem:unifiable-size}. Note that the permutation cannot be the identity since $\bigwedge_i h_1(\gamma_i) = h_2(\gamma_i)$ would imply $h_1=h_2$.
  Assuming the naming $\gamma_i = R_i(\vect t_i)$, consider two new "atoms" $\alpha=R(\vect t_1 \dotsc \vect t_n)$ and $\beta=R(\vect t_{i_1} \dotsc \vect t_{i_n})$ over a fresh "relation name" $R$ of the corresponding arity (\ie the sum of arities of all $\gamma_i$'s).
  Let $q'= \exists \vect x ~ \alpha \land \beta$. 
  Since $q'$ has two "atoms" and $h_1(\alpha\beta)=h_2(\beta\alpha)$ 
  it follows, by the two-"atom" case analyzed before, that $\alpha$ and $\beta$ must be "unifiable", that is, $f(\alpha)=f(\beta)$ for some "unifier" $f$.
  Since the permutation is non-identitary there must be some index $j$ such that $\gamma_j \neq \gamma_{i_j}$. We then have $f(\gamma_j)=f(\gamma_{i_j})$ and hence $\gamma_j,\gamma_{i_j}$ are "unifiable", in contradiction with our hypothesis.
\end{proof}

\begin{corollary}
  For any class $\+C$ of non-"unifiable" "CQs" with bounded "generalized hypertree width", $\evalCountFMS_{\+C}$ is in polynomial time, and so is $\wShapley_{\C}^{w}$ for any "tractable weight function" $w$. 
  In particular, $\msShapley_{\+C}$ is tractable.
\end{corollary}
\begin{proof}
  By \cite[Theorem 7]{PichlerS13} $\evalCountAns_{\+C}$ is in polynomial time, and by \Cref{prop:nonunif-implies-minsups-eq-nrhoms} with \Cref{lem:unifiable-size} $\evalCountFMS_{\+C}$ can be reduced to $\evalCountAns_{\+C}$.
  The tractability of $\wShapley_{\C}^{w}$ then follows from \Cref{lem:wsmsvsmsc1}.
\end{proof}
    
\begin{remark}
  The previous \Cref{prop:nonunif-implies-minsups-eq-nrhoms} cannot be extended to an ``if and only if''. For example, the query
    \[
      \exists x,y,z ~ R(x,y) \land R(x,z) \land S(y) \land T(z)
      \]
      has two "unifiable" "atoms" but every "homomorphism" "induces@@support" a distinct "minimal support".
    \end{remark}

    The previous remark then begs the question of when, exactly, the number of minimal supports coincides with the number of homomorphisms. In other words: Which is the largest subclass of "CQs" for which we can directly use the methods known for counting queries?
    We shall next give a characterization.
    \knowledgenewrobustcmd{\Tsubq}[1][q]{\cmdkl{\+T_{#1}}}
    Given a "CQ" $q$, let \AP$\intro*\Tsubq$ be the set of all equivalence relations on $\qterms(q)$ such that no two distinct "constants" are in the same equivalence class.
    \knowledgenewrobustcmd{\quotientq}[2]{{#1}\cmdkl{{/}}{#2}}
    For any equivalence relation $\sim$ from $\Tsubq$, let \AP$\intro*\quotientq q \sim$ be the result of replacing each variable $x$ with its congruence class $[x]_\sim$ in $q$, where each $[x]_\sim$ is considered as a "variable".\footnote{We assume that in $\quotientq q \sim$ we remove any "atom" repetition.} We call such query a \AP""quotient of $q$"".\footnote{The term `"reduct"' was used in the proof of \Cref{thm:rewritting-SQL}, but here we prefer to use `"quotient@quotient of $q$"' to emphasize that it is the quotient of an equivalence relation $\sim$.} Of course, it follows that $q \homto q/\sim$.
    We can now state a characterization of when the number of "homomorphisms" coincides with that of "minimal supports".
  
    \begin{proposition}[\textnormal{$\countMS_{q} = \countAns_{q}$} characterization]\label{prop:char-count-MS-eq-count-hom}
  The following are equivalent for every "CQ" $q$:
  \begin{enumerate}[(1)]
    \item \label{it:char-count-MS-eq-count-hom:1} the number of "minimal supports" of $q$ coincides with the number of "homomorphisms" on every "database" $\D$ (\ie $\countMS_{q}(\D) = \countAns_{q}(\D)$),
    \item \label{it:char-count-MS-eq-count-hom:2} all "quotients of $q$" are pairwise non-isomorphic "cores" with only one "automorphism@@cq".
  \end{enumerate}
\end{proposition}
\begin{proof}
  \proofcase{\ref{it:char-count-MS-eq-count-hom:2} $\Rightarrow$ \ref{it:char-count-MS-eq-count-hom:1}} First note that each "induced@@support" "support" must be "minimal@minimal support" as otherwise the corresponding quotient would not be a "core".
  Suppose now, by means of contradiction, that every "quotient of $q$" is a 
  "core" but that there are two distinct "homomorphisms" $h_1,h_2: q \homto \+D$ "inducing@@support" the same "(minimal) support@minimal support" $S$. These "homomorphism" are associated with equivalence relations $\sim_1,\sim_2$ of $\Tsubq$ in such a way that there are strong-onto injective "homomorphisms" $g_1 : \quotientq{q}{\sim_1} \homto S$ and $g_2 : \quotientq{q}{\sim_2} \homto S$.
  Hence, we also have $S \homto  \quotientq{q}{\sim_1}$ and $S \homto \quotientq{q}{\sim_2}$. Thus, $\quotientq{q}{\sim_1} \homto \quotientq{q}{\sim_2}$ and $\quotientq{q}{\sim_2} \homto \quotientq{q}{\sim_1}$. Since they are "cores" and "homomorphically" equivalent, $\quotientq{q}{\sim_1}, \quotientq{q}{\sim_2}$ are isomorphic and hence (since by hypothesis distinct "quotients of $q$" are non-isomorphic) we have  $\sim_1=\sim_2$. Hence, summing up, for $\sim \defeq \sim_1=\sim_2$, we how have that there are two distinct (strong onto, injective) "homomorphisms" $g_1,g_2 : \quotientq{q}{\sim} \homto \+D$ "inducing@@support" $S$. But this would imply that $g_1 \circ g_2^{-1} : \quotientq{q}{\sim} \homto \quotientq{q}{\sim}$ is a non-identitary "automorphism@@cq". This is not possible since $\quotientq{q}{\sim}$ has only the identity as "automorphism@@cq" by hypothesis.
  
  \medskip
  
  \proofcase{\ref{it:char-count-MS-eq-count-hom:1} $\Rightarrow$ \ref{it:char-count-MS-eq-count-hom:2}}
  Remember that for every "minimal support" $S$ of $q$ in a "database" $\D$ there is a "homomorphism" $q \homto \D$ "inducing@@support" $S$.
  
  If there is some "quotient@quotient of $q$" $\quotientq{q}{\sim}$ of $q$ which is not a "core", then when evaluated on the "canonical database" $\+D_\sim$ of $\quotientq{q}{\sim}$, there will be a "homomorphism" "inducing@@support" $\+D_\sim$ although it is not a "minimal support". This will make the number of "homomorphisms" strictly larger than the  number of "minimal supports".
  
  Otherwise, if there are two distinct $\sim_1, \sim_2$ in $\Tsubq$ such that $\quotientq{q}{\sim_1}$ and $\quotientq{q}{\sim_2}$ are isomorphic, this means that when $q$ is evaluated on the "canonical database" $\+D_1$ of $\quotientq{q}{\sim_1}$ there will be two distinct "homomorphisms" "inducing@@support" $\+D_1$, and hence the "minimal support" $\+D_1$ will be counted twice (at least), making the number of "homomorphisms" strictly larger than the number of "minimal supports".
  
  Finally, if there is a "quotient@quotient of $q$" $\quotientq{q}{\sim}$ of $q$ which is a "core" but has two distinct "automorphisms@@cq" $g_1,g_2 : \quotientq{q}{\sim} \homto \quotientq{q}{\sim}$, there will be at least two "homomorphisms" onto the "canonical database" $\+D_\sim$ of $\quotientq{q}{\sim}$ and hence the "minimal support" $\+D_\sim$ will be counted twice.
\end{proof}

\begin{example}
  For example, 
  \[
    q = \exists x,y,z ~ R(x,y) \land R(y,z) \land R(z,x) \land A(x) \land B(y) \land C(z)
    \] 
    satisfies the condition of \Cref{prop:char-count-MS-eq-count-hom}, but 
    \[
      q' = \exists x,y,z ~ R(x,y) \land R(y,z) \land R(z,x) \land A(x) \land B(y)
      \] 
      does not. Indeed, the "quotient@quotient of $q$" $\exists x,z ~ R(x,x) \land R(x,z) \land R(z,x) \land A(x) \land B(x)$ of $q'$ is not a "core" since the "atoms" $R(x,z), R(z,x)$ are redundant.
    \end{example}

    In principle, testing the condition of the lemma above seems like it should be $\coNP$-hard, since already testing whether a "CQ" is a "core" is a $\coNP$-complete problem. However, the property also requires that all "quotients@quotients of $q$" are also "cores", which may induce such a severe restriction as to trivialize the property.  We do not currently know what is the complexity of checking such property.
    \begin{open}
      What is the complexity of testing whether for a given "CQ" the number of "homomorphisms" equals the number of "minimal supports" for all "databases"?%
    \end{open}

    However, one sufficient and tractable condition to verify the condition of \Cref{prop:char-count-MS-eq-count-hom} is not to have two "atoms" that may be "unifiable" (\cf\ \Cref{rk:unifiable-tractable} and \Cref{prop:nonunif-implies-minsups-eq-nrhoms}).

    \subsection{Conjunctive Queries with Bounded Self-joins}
    \label{ssec:cq-bounded-sj}

We have seen in \Cref{thm:countMS-ACQ-shP} that in the presence of "self-joins" the problem of counting "minimal supports" remains $\sP$-hard even for "acyclic" queries.
However, the proof of \Cref{thm:countMS-ACQ-shP} uses an \emph{unbounded} number of "self-joins", raising the obvious question of what happens if the number of "self-joins" is bounded. 
We shall now show a positive result, namely,  that counting "minimal supports" is in polynomial time if we simultaneously bound the number of "self-joins" and restrict to acyclic or more generally tree-like "conjunctive queries".
We will further work with a weaker notion for counting "self-joins", where queries containing "atoms" such as $R(a,x)$ and $R(b,y)$ do not increase the ``"self-join width"''.
The measure of ``tree-like'' which we will use here is that of \AP""generalized hypertree width"", which we will not define here but can be found, \eg\ in \cite[Definition~3.1]{GottlobGLS16}. "ACQs" correspond to "CQs" of "generalized hypertree width" 1.

Two distinct atoms $\alpha,\beta$ of a "CQ" $q$ are \AP""mergeable""  if there are two "homomorphisms" $h_\alpha : \alpha \homto \D$, $h_\beta : \beta \homto \D$
to an arbitrary "database" $\D$  such that $h_\alpha(\alpha)=h_\beta(\beta)$ (in particular they must have the same "relation name"). 
Observe that "unifiable" "atoms" are "mergeable", but the converse does not always hold (\eg\ $R(a,x)$ and $R(x,b)$ are "mergeable" but not "unifiable").
An "atom" $\alpha$ is (individually) "mergeable" if it is "mergeable" with some other "atom" in $q$.
The \AP""self-join width"" of a "CQ" $q$ is the size of 
\AP
\begin{align*}
    \intro*\Unif \defeq \set{ t \in \qterms(\alpha) : \alpha \text{ is a "mergeable" "atom" of $q$}}. 
\end{align*} 
In particular, "self-join free" queries have "self-join width" 0, but a "CQ" like  $q\defeq \exists xy. R(c,x) \land R(c',y)$, where $c,c'$ are distinct constants, %
also has "self-join width" 0 since it has no two "mergeable" "atoms".
\begin{remark}["unifiable" vs.\ "mergeable"] \label{remark-sjwidth}
A  natural alternative would have been to use "unifiable" atoms as a basis for the "self-join width", \ie
to define the "self-join width" of a "CQ" $q$ to be the size of the set
\begin{align*}
\set{ t \in \qterms(\alpha) : \alpha \text{ is "unifiable" with some other "atom" of $q$}}.
\end{align*} 
As explained above, "unifiable" implies "mergeable", hence this alternative definition would yield a looser notion of "self-join width". However, we do not currently know how to prove the analog of the upcoming tractability result of \Cref{thm:bounded-selfjoin} under such alternative definition of "self-join width".
\end{remark}

\begin{lemma}
Computing the "self-join width" of a "CQ" is in polynomial time. A "CQ" $q$ is equivalent to a "CQ" of "self-join width" at most $k$ if, and only if, $\core(q)$ has "self-join width" at most $k$.
\end{lemma}
\begin{proof}
For every pair of distinct atoms $\alpha,\alpha'$ one can check in linear time whether and $\alpha$ and $\alpha'$ are "mergeable".
Further, the "self-join width" of a "subquery" of $q$ is smaller or equal than the "self-join width" of $q$, and thus the "self-join width" of $\core(q)$ is the minimum "self-join width" among all queries equivalent to $q$.
\end{proof}
Observe that for queries with no "mergeable" "atoms" (\ie of "self-join width" $0$) the number of homomorphisms coincides with the number of "minimal supports" (by \Cref{prop:nonunif-implies-minsups-eq-nrhoms}, since "unifiable" implies "mergeable"). However, for queries "self-join width" higher than 0 this is no longer the case, and a "minimal support" may be "induced@@support" by several "homomorphisms". The main technical difficulty of the tractability result we will show next indeed resides in keeping track of how many "homomorphisms" "induce@@support" each "minimal support". The main result of this section is as follows.
\begin{theorem}\label{thm:bounded-selfjoin}
For any class $\+C$ of "CQs" having bounded "generalized hypertree width" and bounded "self-join width", $\evalCountFMS_{\+C}$ and $\evalCountMS_{\+C}$ are in polynomial time.%
\footnote{
We understand this statement as a ``promise problem'', that is, the input to the evaluation problem is a query from $\+C$, as opposed to some arbitrary string.
}
\end{theorem}
Before going into the details of the proof, we introduce \AP$\intro*\CQeqneq$ and \AP$\intro*\CQeq$, which allow for \AP""equality atoms"" $x = t$. We say that $q$ is a ""CQneq"", a ""CQeqneq"" or a ""CQeq"" when it is in the respective class. Observe, however, that "equality atoms" can always be removed by replacing every occurrence of $x$ with $t$ if $x$ is distinct from $t$, or simply removing the equality if $t$ is equal to $x$. We write 
\AP$\intro*\elimEq q$ to denote the canonical\footnote{Canonical in the sense of unique up to variable renaming.} equivalent "CQneq" (resp.\ "CQ") query for any "CQeqneq" (resp.\ "CQeq") $q$. 
Notice that a "CQeqneq" is "monotone", although not necessarily closed under "homomorphisms".
\AP The notion of ""homomorphism@@eq"" is adapted to "CQeqneq" queries by additionally requiring that $h(t) = h(t')$ if the atom $t = t'$ is in the query. The query resulting from replacing each variable $x$ of $q$ with $h(x)$ is called the ""homomorphic image"" of $q$ via $h$.
The "generalized hypertree width" of a "CQneq" query is defined just as for "CQs", by treating $\neq$ as any binary relation name.

\begin{remark}\label{rk:homimage:ghw}
For any two "Boolean" "CQ"s (or "CQeq", "CQneq", "CQeqneq") $q,q'$ if $h: q \homto q'$ and further $q'$ is the "homomorphic image"
of $q$ via $h$, then the "generalized hypertree width" of $q$ is smaller or equal to that of $q'$.
\end{remark}

Further, it has been shown that counting the number of "homomorphisms" for any class of bounded "generalized hypertree width" "CQs" is in polynomial time \cite[Theorem 7]{PichlerS13}, and it can easily be shown that the same holds for "CQneq"s.

\begin{corollary}[{of \cite[Theorem 7]{PichlerS13}}]\label{cor:boundgpw:cqneq}
For any class of "CQneq" queries of bounded "generalized hypertree width", counting the number of "homomorphisms" is in polynomial time.
\end{corollary}
\begin{proof}
It suffices to materialize a relation for $\neq$ applying the algorithm of \cite[Theorem 7]{PichlerS13}.
\end{proof}

We can now turn to the proof of the main result of the section, \Cref{thm:bounded-selfjoin}.
\begin{proof}[Proof of \Cref{thm:bounded-selfjoin}]

Let $q$ be a "CQ" from $\+C$. 
Consider \AP$\intro*\Equivs$ to be the set of all \AP""equivalence relations"" over $\Unif$ (\ie relations $E \subseteq \Unif \times \Unif$ being reflexive, symmetric and transitive) not containing any pair $(c,c')$ of distinct "constants".
For any $E \in \Equivs$, define
\AP$\intro*\qE \in \CQ$ as $\core(\elimEq{p})$, where $p \in \CQeq$ is defined as the result of adding to the query $q$ the "equality atoms" $t = t'$ for every $(t,t') \in E$.
Let \AP$\intro*\qEneq \in \CQneq$ be $\qE$ with some added "inequality atoms" $t \neq t'$ for 
$t,t' \in  \qterms(\qE)$ such that $(t,t') \in (\Unif \times \Unif) \setminus E$.      
Let $\intro*\Punif \defeq \set{\qE : E \in \Equivs \text{ s.t. } \qEneq[E'] \homto \qEneq \text{ implies } \qEneq \homto \qEneq[E'] \text{ for all } E' \in \Equivs}$. %

The proof of the statement can be shown by establishing the following facts:
\begin{enumerate}
\item \label{it:induced-size}
If a "homomorphism@@eq" $h:\qEneq \homto \D$ "induces@@support" a "support" $S$, then $|S|$ is the number of "atoms" of $\qE$ (shown in \Cref{cl:induced-size}).
\item \label{it:Punif-poly}
$\Punif$ can be computed in polynomial time (\Cref{cl:Punif-poly}).
\item \label{it:partition-MS}
The set of "minimal supports" of $q$ on $\D$ is partitioned into 
\[
\set{M \subseteq \D : M \text{ is "induced@@support" by a "homomorphism@@eq" } \qEneq \homto \D}_{\qE \in \Punif}
\qquad \text{(\Cref{cl:partition-MS}).}\]
\item \label{it:ans-automorphisms-ms}
For every $\qE \in \Punif$ and "database" $\D$, we have 
\[
\countAns_{\qEneq}(\+D) = |\Auto{\qEneq}| \cdot |\set{M : M \text{ "induced@@support" by "homomorphism@@eq" } \qEneq \homto \+D}|
\]
where $\Auto{\qEneq}$ is the set of "automorphisms@@cqneq" $\qEneq \homto \qEneq$ (\Cref{cl:ans-automorphisms-ms}).
\item \label{it:polytimeAns}
$|\Auto{\qEneq}|$ and $\countAns_{\qEneq}(\+D)$ can be computed in polynomial time (\Cref{cl:polytimeAns}).
\end{enumerate}

Once these claims are established, for every fixed $k$, $\countFMS_{q}(k,\D)$ can be computed in polynomial time. Indeed, let $P_k \subseteq \Punif$ be the set of all queries having exactly $k$ "atoms".
Then, by items \eqref{it:partition-MS} and \eqref{it:ans-automorphisms-ms} it follows that 
\[
\countFMS_{q}(k,\D) = \sum_{\qE \in P_k} \frac{\countAns_{\qEneq}(\+D)}{|\Auto{\qEneq}|}.
\]
which can be computed in polynomial time by item \eqref{it:polytimeAns}.

\medskip

We now turn our attention to the proofs of the items stated above.
We first remark that, since each $\qE$ of $\Punif$ is a "core", and the "core" is unique modulo \AP""isomorphism"" (\ie the relation of variable renaming), we can (and shall) assume that every two distinct queries from $\Punif$ are non-"isomorphic".

We introduce some notation and terminology used in the following proofs.
For any two Boolean "CQs" $p,p'$ we have that $p \homto p'$ iff $\D_{p'} \models p$, where $\D_{p'}$ is the "canonical database" of $p'$.
For an "atom" $\alpha = R(\vect t)$, let us write $\alpha[i]$ to denote the $i$-th element of $\vect t$ (a "constant" or "variable").
A \AP""subquery"" of a "CQ" $q$ is any query resulting from removing atoms from $q$ -- in particular $\core(q)$ is a "subquery" of $q$.

\begin{claim}[item \eqref{it:induced-size}]\label{cl:induced-size}
If a "homomorphism@@eq" $h:\qEneq \homto \D$ "induces@@support" a "support" $S$, then $|S|$ is the number of "atoms" of $\qE$.
\end{claim}
\begin{nestedproof}
By means of contradiction, take any two distinct "atoms" $\alpha,\beta$ of $\qE$ such that $h(\alpha)=h(\beta)$. Since $\alpha$ and $\beta$ are distinct, they are "mergeable" and there is some pair of distinct "terms"  $(t,t') \in \qterms(\alpha)\times\qterms(\beta) \subseteq \Unif \times \Unif$ such that $h(t)=h(t')$.
Remember that $\qE$ has no "equality atoms" and every pair of distinct terms from $\Unif$ (such as $t,t'$) present in $\qE$ must appear as an "inequality atom" in $\qEneq$. Hence, $t \neq t'$ is in $\qEneq$, contradicting the fact that $h(t)=h(t')$ for a "homomorphism" $h: \qEneq \homto \D$.
\end{nestedproof}

Let us remark next that the "generalized hypertree width" in the queries considered before preserves boundedness, which is necessary to establish tractability of $\Punif$.

\begin{remark}\label{rk:hatp:ghw}
If $q$ has "generalized hypertree width" $k$ and "self-join width" $\ell$, then $\elimEq{p}$ in the definition of $\qE$ cannot have "generalized hypertree width" greater than $k+\ell$.
    Indeed, if we add all variables of $\vars(\elimEq{p}) \cap \Unif$ to all bags of a tree decomposition of generalized hypertree width $k$, we end up with a decomposition of width at most $k+\ell$.
\end{remark}
\begin{remark}\label{rk:qE:ghw}
If $q$ has "generalized hypertree width" $k$ and "self-join width" $\ell$, then  $\qE$ cannot have "generalized hypertree width" greater than $k+\ell$. This is because $\qE$ is the "homomorphic image"\footnote{The "core" of a query is always a "homomorphic image" of the query.} of a query of width $\leq k+\ell$ 
by \Cref{rk:hatp:ghw}, and thus by \Cref{rk:homimage:ghw} the width does not increase.
Further $\qEneq$ cannot have "generalized hypertree width" greater than $k+2\ell$, by adding all variables of $\vars(\qE) \cap \Unif$ to all bags.
\end{remark}

\begin{claim}[item \eqref{it:Punif-poly}]\label{cl:Punif-poly}
$\Punif$ can be computed in polynomial time.
\end{claim}
\begin{nestedproof}
First observe that $\Equivs$ is of constant size since $\Unif$ is assumed to be of constant size.
Consider $E \in \Equivs$ and the corresponding $\qE  = \core(\elimEq{p})$ where, remember, $p$ is the addition of a constant number of equalities to $q$, and $\elimEq{p}$ is the result of making a constant number of variable replacements. Hence, $\elimEq{p}$ can be computed in polynomial time and has "generalized hypertree width" bounded by $k+\ell$ by \Cref{rk:hatp:ghw}.
This implies that $\core(\elimEq{p})$ can be computed in polynomial time: we can check if a "subquery" $p'$ of $\elimEq{p}$
is equivalent to $p$
by simply testing $\D_{p'} \models \elimEq{p}$, which can be done in polynomial time due the bounded "generalized hypertree width" of $\elimEq{p}$, where $\D_{p'}$ is the "canonical database" of $p'$ by \Cref{cor:boundgpw:cqneq}.
We can thus obtain $\core(\elimEq{p})$ by greedily removing atoms until no subquery is equivalent to the query, since the "core" ensures that any way of removing atoms leads to the same query.
\end{nestedproof}

The proofs of \Cref{cl:partition-MS,cl:ans-automorphisms-ms} will use the following argument.
\begin{claim}\label{cl:compose-homs}
If for some $E,E' \in \Equivs$ there are "homomorphisms@@eq" $\qEneq \homto \D$ and $\qEneq[E'] \homto \D$ 
"inducing@@support" the "supports" $S$ and $S'$ respectively with $S' \subseteq S$, then there is a "homomorphism@@eq" $\qEneq[E'] \homto \qEneq$.
\end{claim}
\begin{nestedproof}
Let $h:\qEneq \homto \D$ and $h':\qEneq[E'] \homto \D$ the "homomorphisms@@eq" of the statement.
We define $g: \qterms(\qEneq[E']) \to \qterms(\qEneq)$ as follows:
\begin{itemize}
\item For every $t \in \qterms(\qEneq[E']) \setminus \Unif$ we define $g(t) \defeq t$. Note that these are terms appearing exclusively in "non-mergeable" "atoms" of $q$, and hence they are present in every $\qE[E'']$-query for $E'' \in \Equivs$.
\item 
Now suppose $t \in \qterms(\qEneq[E']) \cap \Unif$. 
Let $\tilde h : q \homto \qE$ and $\tilde h': q \homto \qE[E']$.
Since $t$ is in the $\tilde h'$-image of some "mergeable" "atom" $\alpha$ of $q$, 
the $h$-preimage of $h'(\tilde h'(\alpha))$ must also be in the $\tilde h$-image of some "mergeable" "atom" $\beta$ of $q$ (otherwise, we would have $h(\tilde h(\beta)) = h'(\tilde h'(\alpha))$ contradicting that $\beta$ is "non-mergeable" in $q$).
Moreover, since we ask that all $\Unif$-variables are pairwise distinct in $h$, there is exactly one element $t'$ in $h^{-1}(h'(t))\cap \Unif$. 
We then define $g(t) \defeq t'$.
\end{itemize}
We finally show that $g$ is indeed a "homomorphism@@eq".
First, note that for every pair of distinct terms $t,t'$ in $\Unif \cap \vars(\qEneq[E'])$, we have $g(t) \neq g(t')$ (otherwise, $h'$ would not be a "homomorphism@@eq" $h': \qEneq[E'] \homto M'$). 
Hence, it suffices to show that $g$ is a "homomorphism@@eq" $g: \qE[E'] \homto \qE$.

Take an "atom" $\alpha$ of $\qE[E']$ and let us show that $g(\alpha)$ is in $\qE$.
Consider the "fact" $h'(\alpha)$ of $M'$, and let $\beta$ be an "atom" of $\qE$ such that $h(\beta)=h'(\alpha)$.
We show that $g(\alpha)=\beta$.

Let $x \in \qterms(\alpha)$. If $x \not\in \Unif$, then $\alpha$ is comes from a "non-mergeable" "atom" of $q$ and if $x = \alpha[i]$ then $g(x) = \beta[i]$ for every $i$. 
If $x \in \Unif$ and $x = \alpha[i]$, by definition of $g$ it follows that $g(x) \in \Unif$ and further that $g(\alpha[i]) = \beta[i]$ for every $i$.%
\end{nestedproof}

\begin{claim}[item \eqref{it:partition-MS}]\label{cl:partition-MS}
The set of "minimal supports" of $q$ in $\D$ is partitioned into 
\[
\set{M \subseteq \D : M \text{ is "induced@@support" by a "homomorphism@@eq" } \qEneq \homto \D}_{\qE \in \Punif}.
\]
\end{claim}
\begin{nestedproof}
\proofcase{Supports are minimal.}
Let $M$ be a "support" "induced@@support" by a "homomorphism@@eq" $h: \qEneq \homto \D$. Let us show first that $M$ is a "minimal support" of $q$.
By means of contradiction, suppose it is not the case, and let $M' \subsetneq M$ be "induced@@support" by a "homomorphism@@eq" $h': q \homto M'$.
Let $E' \in \Equivs$ be the reflexive-symmetric-transitive closure\footnote{We assume that $h(c)=c$ for every "constant" $c$.} of $\set{(t,t') \in \Unif \times \Unif : h'(t) = h'(t')}$. It follows that $h': \qEneq[E'] \homto M'$.

We have that $\qEneq[E'] \homto \qEneq$ due to \Cref{cl:compose-homs}. By definition of $\Punif$ we must also have the converse "homomorphism@@eq" $\qEneq \homto \qEneq[E']$, which means that $M'$ is also "induced@@support" by some "homomorphism@@eq" $\qEneq \homto M'$. But this contradicts \Cref{cl:induced-size} since $|M'| < |M|$.

\proofcase{Supports are distinct.}
By means of contradiction, suppose there is a "minimal support" $M$ "induced@@support" by two distinct $\qE, \qE[E'] \in \Punif$. In other words, there are $h: \qEneq \homto \D$ and $h': \qEneq[E'] \homto \D$ "inducing@@support" the same "(minimal) support@minimal support" $M \subseteq \D$. By \Cref{cl:compose-homs} it follows that $\qEneq$ and $\qEneq[E']$ are "homomorphically" equivalent, and then so are $\qE$ and $\qE[E']$. Since $\qE$ and $\qE[E']$ are "cores" by definition, they must be "isomorphic", and this contradicts that they are distinct by definition of $\Punif$.
\end{nestedproof}

\begin{claim}[item \eqref{it:ans-automorphisms-ms}]\label{cl:ans-automorphisms-ms}
For every $\qE \in \Punif$ and "database" $\D$ we have 
\[
\countAns_{\qEneq}(\+D) = |\Auto{\qEneq}| \cdot |\set{M : M \text{ "induced@@support" by "homomorphism@@eq" } \qEneq \homto \+D}|.
\]
\end{claim}
\begin{nestedproof}
Take any "support" $M$ "induced@@support" by $h: \qEneq \homto \D$.
Any "homomorphism@@eq" $h':\qEneq \homto \D$ "inducing@@support" $M$ is so that $h'$ is the composition of an "automorphism@@cqneq" of $\qEneq$ with $h$ due to \Cref{cl:compose-homs}. Hence, the set of "homomorphisms@@eq" "inducing@@support" $M$ is $\set{ h \circ g : g \in \Auto{\qEneq}}$. Hence there are $|\Auto{\qEneq}|$ many "homomorphisms@@eq" in $\qEneq(\D)$ for each "support" in $\set{M : M \text{ "induced@@support" by "homomorphism@@eq" } \qEneq \homto \+D}$. 
Further, two distinct "induced@@support" "supports" cannot arise from the same "homomorphism@@eq".
\end{nestedproof}

\begin{claim}[item \eqref{it:polytimeAns}]\label{cl:polytimeAns}
$|\Auto{\qEneq}|$ and $\countAns_{\qEneq}(\+D)$ can be computed in polynomial time.
\end{claim}
\begin{nestedproof}
Recall that the "generalized hypertree width" of $\qEneq$ cannot exceed $k+2\ell$ by \Cref{rk:qE:ghw}, where $k$ is the "generalized hypertree width" of $q$ and $\ell$ its "self-join width".
Hence, by \Cref{cor:boundgpw:cqneq} $\countAns_{\qEneq}(\+D)$ can be computed in polynomial time.
Finally, we observe that $|\Auto{\qEneq}|$ is equal to $\countAns_{\qEneq}(\+D_{\qE})$, where $\+D_{\qE}$ is the "canonical database" of $\qE$,\footnote{In this case, we materialize the $\neq$ relation and treat the $\neq$ relation of $\qE$ as any binary relation.} and thus $|\Auto{\qEneq}|$ can be also computed in polynomial time.
\end{nestedproof}

This concludes with the proof of \cref{it:induced-size,it:Punif-poly,it:partition-MS,it:ans-automorphisms-ms,it:polytimeAns}, and hence of the proof of \Cref{thm:bounded-selfjoin}.
\end{proof}

\section{Related Measures}
\label{sec:relwork}

We review some measures known in the literature and compare them with the "WSMS" family.

\subsection{Shapley Value of Inconsistency Measures}\label{ssec:inconsistency}
As mentioned earlier, the "drastic@@shapley" and the "MS Shapley" values are closely related to the previously studied drastic and 
\textit{MI} inconsistency measures, denoted by $I_d$ and $I_\mathsf{MI}$ in the literature %
\cite{hunterMeasureConflictsShapley2010,LivshitsK22}. Indeed, $I_d$ assigns 1 if the considered database (or knowledge base) is inconsistent, and 0 otherwise, similarly to how the "drastic@@shapley" value employs a wealth function that gives 1 or 0 depending on whether the query is satisfied. 
Likewise, $I_\mathsf{MI}$ counts the number of minimal inconsistent subsets of the given database (or knowledge base),
similarly to how the "MS Shapley" "wealth function" counts the number of minimal supports for the query. 

Livshits and Kimelfeld \cite{LivshitsK22} have considered three further inconsistency measures for databases, denoted by $I_\mathsf{P}$, $I_{\mathsf{MC}}$ and $I_{\mathsf{R}}$ respectively, defined in the absence of exogenous atoms. The first two measures translate into the following families of "wealth functions":
\begin{itemize}
   \item\AP $\intro*\Pscorefun$ is the analogue of $I_{\mathsf{P}}$: $\pscorefun_{q,\emptyset}(S)$ is the number of facts that are "relevant" to $q$ in $S$.
   \item\AP $\intro*\MCscorefun$ is the analogue of $I_{\mathsf{MC}}$: $\mcscorefun_{q,\emptyset}(S) = mc(S) + ss(S) -1$ where $mc(S)$ is the number of maximal subsets $\Delta\subseteq S$ s.t.\ $\Delta\not\models q$, and
$ss(S) \defeq |\set{s \in S : \{s\}\models q}|$.\footnote{The definition in \cite{LivshitsK22} is simpler than ours because in their context (functional dependencies), the analogue of such a ``singleton support'' cannot appear, hence $ss(S)=0$; the full definition can be found in \cite[Definition 2]{grantMeasuringConsistencyGain2011} ($I_M$ measure).}
\end{itemize}

We did not extend these definitions to add "exogenous" facts, because the following result proves that they are not suitable "responsibility measures" using our criteria, even without them.

\begin{proposition}%
   The Shapley values of $\Pscorefun$ and $\MCscorefun$ don’t satisfy \ref{MS:t}, even assuming $\Dx=\emptyset$.
\end{proposition}
\begin{proof}
   The Shapley value of $\Pscorefun$ attributes the score $1$ to all facts in the example depicted in \Cref{fig:counter-ex}-(a) ($k=l=1,|A_1| = 1, |B_1|=2$ with the notations of \ref{MS:t}), and the one of $\MCscorefun$ attributes the score $\nicefrac{144}{120}$ to all facts in the variant defined by $k=l=1,|A_1| = 2, |B_1|=3$.
\end{proof}

We do extend the definition of the third to account for $\Dx$, in order to prove its suitability. %

\begin{itemize}
   \item\AP $\intro*\Rscorefun$ is the analogue of $I_{\mathsf{R}}$: $\rscorefun_{q,\D}(S)$ is the size of the
smallest $\Gamma\subseteq S$ such that $(S\setminus\Gamma)\cup\Dx\not\models q$.
\end{itemize}

   \begin{proposition}\label{prop:rsatmstest}
   The Shapley value of $\Rscorefun$ satisfies \ref{Shdb:1}, \ref{Shdb:2} and \ref{MS:t}.
   \end{proposition}
   \begin{proof}
      \proofcase{\ref{Shdb:1}} If two queries $q_1,q_2$ define the same Boolean function, and two facts $\alpha,\beta\in \Dn$ are equivalent in the sense that for all $S\subseteq \Dn\setminus\{\alpha,\beta\}$, $S\cup\{\alpha\}\cup\Dx \models q_1$ iff $S\cup\{\beta\}\cup\Dx \models q_2$, then the optimal repairs will be equivalent in the sense that for all $S\subseteq \Dn \setminus \{\alpha,\beta\}$, $\rscorefun_{q_1,\D}(S\cup\{\alpha\})=\rscorefun_{q_2,\D}(S\cup\{\beta\})$. \ref{Shdb:1} is therefore satisfied as a consequence of \ref{Sh:1}.\medskip

      \proofcase{\ref{Shdb:2}} An "irrelevant" fact will have no consequence on the size of the optimal repair, meaning it will be a "null player" for $\rscorefun_{q,\D}$. \ref{Shdb:2} is therefore satisfied as a consequence of \ref{Sh:2}.\medskip

      \proofcase{\ref{MS:t}} Consider an instance $\Dn=A_1 \cup \dots \cup A_k \cup B_1 \cup B_l$ and $q$ like in the definition of \ref{MS:t}. The optimal repair in that instance is $\{\alpha,\beta\}$, which is of size two. Intuitively, there will be two points of wealth, one shared amongst the facts that appear in some $A_i$, another amongst those that appear in some $B_j$, and for every coalition where the addition of $\beta$ changes $\rscorefun_{q,\D}$ there is a corresponding coalition where the addition of $\alpha$ does.

   Formally, since for all $i\le l. ~ |A_i|\le |B_i|$, we can define for all $i\le l$ an injective mapping $\eta_i: A_i \hookrightarrow B_i$, and given the intersections we can assume that $\eta_i(\beta)=\alpha$ for all $i$. We can define the permutation $\eta$ that swaps every $\gamma\in A_i$ with $\eta_i(\gamma)$, thereby extending all $\eta_i$ at once (it is well-defined because there are no intersections besides $\{\alpha\}$ and $\{\beta\}$).

   Now we have the permutation $\eta$, consider a coalition $S\subseteq \Dn$ such that $\rscorefun_{q,\D}(S\cup\{\beta\}) \neq \rscorefun_{q,\D}(S)$. This happens when the addition of $\beta$ completes one of the $B_j$, thereby increasing $\rscorefun_{q,\D}$ by 1. This means $B_j\setminus\{\beta\}\subseteq S$ for some $j$ hence $A_j\setminus\{\alpha\} \subseteq \eta^{-1}(S)$ by definition of $\eta$, and $\rscorefun_{q,\D}(\eta^{-1}(S)\cup\{\alpha\}) - \rscorefun_{q,\D}(\eta^{-1}(S)) =1$. We finally apply \Cref{formul:slike} to obtain the equality \eqref{1.a} below, since the "Shapley value" is a "Shapley-like" score:

   \begin{align}
      \Sh_{\rscorefun_{q,\D}}(\beta)
	 &= \displaystyle\sum_{S\subseteq \Dn \setminus \set \beta } c(|S|,|\Dn|)(\rscorefun_{q,\D}(S\cup\{\beta\}) - \rscorefun_{q,\D}(S))\tag{a}\label{1.a}\\
	 &= \displaystyle\sum_{\substack{
	    S\subseteq \Dn \setminus \{\beta\}\\
	    \rscorefun_{q,\D}(S\cup\{\beta\}) - \rscorefun_{q,\D}(S) =1
	 }} c(|S|,|\Dn|)\notag\\
	 &\le \displaystyle\sum_{\substack{
	    \eta^{-1}(S)\subseteq \Dn \setminus \{\alpha\}\\
	    \rscorefun_{q,\D}(\eta^{-1}(S)\cup\{\alpha\}) - \rscorefun_{q,\D}(S) =1
	 }} c(|\eta^{-1}(S)|,|\Dn|)\tag{b}\label{1.b}\\
	 &= \Sh_{\rscorefun_{q,\D}}(\alpha)\notag
   \end{align}

   \eqref{1.b} comes from the fact that $\rscorefun_{q,\D}(S\cup\{\beta\}) - \rscorefun_{q,\D}(S) =1$ implies $\rscorefun_{q,\D}(\eta^{-1}(S)\cup\{\alpha\}) - \rscorefun_{q,\D}(S) =1$. The converse, however, is false provided $k>l$ or $|A_i|<|B_i|$ for some $i$: in this case there will be some $A_i$ such that $\eta(A_i)$ isn’t some $B_j$, and $S\defeq\eta(A_i\setminus\{\alpha\})$ provides a counter-example, which makes the inequality \eqref{1.b} strict, yielding $\Sh_{\rscorefun_{q,\D}}(\beta) < \Sh_{\rscorefun_{q,\D}}(\alpha)$ as desired.
   \end{proof}

These results show that $\Pscorefun$ and $\MCscorefun$ do not yield suitable %
 "responsibility measures", but $\Rscorefun$ remains a reasonable option.
 While we shall leave a more thorough investigation of $\Rscorefun$ to future work (in particular, it remains to determine if the Shapley value of $\Rscorefun$  satisfies \ref{MS:1} and \ref{MS:2} as defined in \Cref{ssec:ax1}),
 we can already show that it is relatively difficult to compute. %
Indeed, we show that, due to its similarity with $\Dscorefun$, it does not enjoy the data complexity tractability of "WSMSs". %

\begin{proposition}\label{prop:repairishard}
There exists a "CQ" $q$ such that $\rShapley_{q}$ is $\sP$-hard.
\end{proposition}
\begin{proof}
Define two "sjfCQ"s $q_0\defeq\exists x, y, z\ R(x) \land S(x,y) \land T(y) \land P(z)$, and $q_1 \defeq \exists x, y\ R(x) \land S(x,y) \land T(y)$. We prove the hardness of $\rShapley_{{q_0}}$ by showing $\dShapley_{{q_1}} \polyrx \rShapley_{{q_0}}$, since $\dShapley_{{q_1}}$ is known to be $\sP$-hard \cite[Proposition 4.6]{livshitsShapleyValueTuples2021}. In fact the following proof is heavily inspired by the one for \cite[Proposition 4.6]{livshitsShapleyValueTuples2021}.\medskip

The key observation to build the reduction is the following: if we only restrict ourselves to partitioned databases $\D$ that only contain a single "endogenous" $P$-atom $\mu\defeq P(c)$ (and no "exogenous" ones), then the optimal repair measure will never be higher than 1 (removing $\mu$ will always falsify $q_0$) and more precisely we will have $\rscorefun_{q_0,\D}(X\cup\{\mu\}) = \dscorefun_{q_1,\D}(X)$ for every $X\in\Dn\setminus \{\mu\}$.

With this observation in mind, let $\D^{in}$ be an input instance for $\dShapley_{{q_1}}$. Without loss of generality we can assume that all atoms in $\D^{in}$ have predicates in $\{R,S,T\}$: any atom on other predicates are "irrelevant" to $q_1$ and can be removed without affecting the "Shapley value". Let $n\defeq |\Dn[\D^{in}]|$ and for every $i\in\lBrack 0,n\rBrack$ build the "partitioned database" $\D^i$ defined by $\Dn[\D^i] \defeq \Dn[\D^{in}] \uplus \{\mu\} \uplus \{S(c_0,c_k): k\in[i]\}$ and $\Dx[\D^{in}] \uplus \{R(c_0)\} \uplus \{T(c_k) : k\in[i]\}$, with the $c_k$ being fresh constants. Also denote $\beta_k \defeq S(c_0,c_k)$ for short. This construction is depicted in \Cref{fig:sameproof}.

\begin{figure}[tb]
\centering
      \begin{tikzpicture}
	\coordinate (00) at (0, 0);
	\coordinate (01) at (2, 0);
	\coordinate (02) at (4, 0);
	\coordinate (03) at (4, -1);
	\coordinate (04) at (-4.5, -1);
	\coordinate (05) at (-2, 0.5);
	\coordinate (06) at (4, -0.5);
	\coordinate (07) at (-4, 0.25);
	\coordinate (08) at (-4.5, -0.25);
	\begin{pgfonlayer}{nodelayer}
		\node [draw, circle] (0) at (00) {$c$};
		\node [above] at (0.north) {$P$};
		\node [left] at (0.west) {$\alpha$};
		\node [draw, circle,color=red] (1) at (01) {$c_0$};
		\node [draw, circle,color=red] (2) at (02) {$c_1$};
		\node [draw, circle,color=red] (3) at (03) {$c_i$};
		\node [] (4) at (04) {};
		\node [] (5) at (05) {};
		\node [] (6) at (06) {$\cdots$};
		\node [] (7) at (07) {$\D^0$};
		\node [anchor=west] (8) at (08) {\footnotesize (only $\{R,S,T\}$)};
		\path ($(1.north)+(0,0.3)$) node[color=red] {$R$} -- node[midway] {$S$} ($(2.north)+(0,0.3)$) node[color=red] {$T$};
	\end{pgfonlayer}
	\begin{pgfonlayer}{edgelayer}
		\draw [color=red, thick] (4.center) rectangle (5.center);
		\draw [dash pattern=on 1em off 1em, thick] (4.center) rectangle (5.center);
		\draw [->,>=stealth] (1) to node[midway, above=-.12] {$\beta_1$} (2);
		\draw [->,>=stealth] (1) to node[midway, above=-.12, sloped] {$\beta_i$} (3);
	\end{pgfonlayer}
\end{tikzpicture}%
   
\caption{Construction of $\D^i$ from $\D^{in}$, with the "endogenous" parts in black and the "exogenous" ones in red.}
\label{fig:sameproof}
\end{figure}
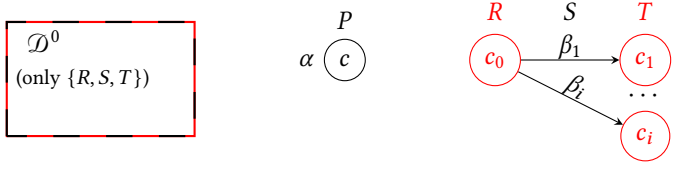

We now use \Cref{formul:slike}
to express $\Sh_{\rscorefun_{q_0,\D^i}}(\mu)$, with $c(j,m)\defeq\frac{j!(m-j-1)!}{m!}$:

\begin{equation*}
\Sh_{\rscorefun_{q_0,\D^i}}(\mu)
= \sum_{X\subseteq \Dn[\D^i] \setminus \{\mu\} } c(|X|,|\Dn[\D^i]|) \left(\rscorefun_{q_0,\D^i}(X\cup\{\mu\}) - \rscorefun_{q_0,\D^i}(X)\right)
\end{equation*}

Which we can rewrite as follows, given $|\Dn[\D^i]|=n+i+1$:

\begin{equation*}
\Sh_{\rscorefun_{q_0,\D^i}}(\mu)
= \sum_{j=0}^{n+i}
\sum_{\substack{X\subseteq \Dn[\D^i] \setminus \{\mu\}\\ |X|=j}}
c(j,n+i+1)\left(\rscorefun_{q_0,\D^i}(X\cup\{\mu\}) - \rscorefun_{q_0,\D^i}(X)\right)
\end{equation*}

Now we need to consider the sets $X\subseteq \Dn[\D^i] \setminus \{\mu\}$ such that $\rscorefun_{q_0,\D^i}(X\cup\{\mu\}) \neq \rscorefun_{q_0,\D^i}(X)$. As discussed above, $\rscorefun_{q_0,\D^i}(X)=0$ since $\mu\notin X$, so this precisely happens when $\rscorefun_{q_0,\D^i}(X\cup\{\mu\})\neq 0$ \ie when $\dscorefun_{q_1,\D^i}(X)\neq 0$ that is when $X\cup \D^i \models q_1$. This happens in the following two mutually exclusive cases:

\begin{enumerate}[(i)]
\item\label{cas:1.1} $X$ contains some $\beta_k$;
\item\label{cas:1.2} $X\subseteq \Dn[\D^{in}]$ and $X\cup \Dx[\D^i] \models q_1$.
\end{enumerate}

To choose some $X$ of size $j$ that doesn’t satisfy \ref{cas:1.1}, one simply has to choose the $j$ elements of $X$ amongst the $n$ elements of $\Dn$ (since all $\beta_k$ are forbidden). Hence the number of $X$ of size $j$ that satisfies \ref{cas:1.1} is $\binom{n+i}{j} - \binom{n}{j}$. As for \ref{cas:1.2}, since none the constants $c_k$ appears in $\D^{in}$ and the variables in $q_1$ are connected, the exogenous elements added during the construction of $\D^i$ cannot help a subset of $\Dn[\D^{in}]$ to satisfy $q_1$, hence the condition $X\cup \Dx[\D^i] \models q_1$ can be replaced by $X\cup \Dx[\D^{in}] \models q_1$. Denote by $\mathrm{fgmc}_j$ the number of such $X$s, that is the number of $X\subseteq \Dn[\D^{in}]$ of size $j$ such that $X\cup \Dx[\D^i] \models q_1$. This name has been chosen because these are the solutions to $\mathrm{FGMC}_{q_1}$ as defined in \cite[\S 3.2]{ourpods24} on the instance $\D^{in}$, which we now try to compute. For this we rearrange the above equation:

\begin{align*}
\Sh_{\rscorefun_{q_0,\D^i}}(\mu)
&= \sum_{j=0}^{n+i}
\sum_{\substack{X\subseteq \Dn[\D^i] \setminus \{\mu\}\\
|X|=j \\ X\models \ref{cas:1.1}}}
c(j,n+i+1)
+
\sum_{j=0}^{n+i}
\sum_{\substack{X\subseteq \Dn[\D^i] \setminus \{\mu\}\\
|X|=j \\ X\models \ref{cas:1.2}}}
c(j,n+i+1)
\\
&= \sum_{j=0}^{n+i}
c(j,n+i+1) \left(\binom{n+i}{j} - \binom{n}{j}\right)
+
\sum_{j=0}^{n}
c(j,n+i+1) \mathrm{fgmc}_j
\end{align*}

For short, denote $\mathrm{Sh^i}\defeq \Sh_{\rscorefun_{q_0,\D^i}}(\mu) - \sum_{j=0}^{n+i} c(j,n+i+1) \left(\binom{n+i}{j} - \binom{n}{j}\right)$. Observe that this value can be computed in polynomial time from a call to a $\rShapley_{{q_0}}$ oracle. Now if we further denote by $Y$ the column vector $(\mathrm{Sh}^0 \dots \mathrm{Sh}^{n})$, by $X$ the column vector $(\mathsf{fgmc}_0 \dots \mathsf{fgmc}_n)$, and by $A$ the square matrix of general term $c(j,n+i+1)=\frac{j!(n+i-j)!}{(n+i+1)!}$ for $i,j\in \lBrack 0,n\rBrack$, the above gives $Y=AX$.\medskip

We now show that $A$ is invertible. For this we multiply every line by $(n+i+1)!$, divide every column by $j!$ and reverse the column order to obtain the matrix of general term $(i+j)!$, which is invertible by \cite[proof of Theorem 1.1]{bacherDeterminantsMatricesRelated2002}. By computing the inverse of $A$ we can finally obtain $X$ in polynomial time using an oracle to $\rShapley_{{q_0}}$. We have thus shown that $\mathrm{FGMC}_{q_1} \polyrx \rShapley_{{q_0}}$, and we can conclude by \cite[Proposition 4.6]{livshitsShapleyValueTuples2021}, which states that $\dShapley_{{q_1}}$ is $\sP$-hard.
\end{proof}

Based on \Cref{prop:repairishard} alone, $\Rscorefun$ is no better than $\Dscorefun$ complexity-wise. In fact it is even worse, because the latter can be efficiently approximated on all
"CQs" \cite[Corollary 4.14]{livshitsShapleyValueTuples2021} (and even all "UCQs" with a trivial adaptation of the proof), %
while the same cannot be done for the former.

Formally, we shall use two notions of approximation of a numerical function $f$, which we briefly recall.
The simplest kind is the \AP ""deterministic multiplicative approximation"", often simply called \reintro{approximation}. It is defined as an algorithm $A(x)$ where $x$ is an input for $f$ that runs in polynomial time in $x$, and guarantees $A(x) \in [f(x); \rho \cdot f(x)]$ for all $x$, where $\rho \in \lR_+$ is a constant called the \AP ""approximation factor"" of the algorithm.\footnote{We only care about $\rho > 1$ here but the case $\rho < 1$ simply defines an under-approximation of $f$.}
\AP The other one is the ""fully polynomial randomized approximation scheme"" (\reintro{FPRAS}), that is defined as a randomized algorithm $A(x,\epsilon,\delta)$ where $x$ is an input for $f$ and $\epsilon,\delta\in (0,1)$ that runs in polynomial time in $x$, $\nicefrac{1}{\epsilon}$ and $\log(\nicefrac{1}{\delta})$ that guarantees $\mathbf{P}\left[\frac{f(x)}{1+\epsilon} \le A(x,\epsilon,\delta) \le (1+\epsilon)f(x)\right] \ge 1-\delta$, where $\mathbf{P}$ is the probability over the randomness of~$A$.
\AP Also recall that $\intro*\BPP$ is the class of decision problems for which there exists a polynomial time randomized algorithm that always gives the correct answer with probability at least $\nicefrac{2}{3}$.

\begin{proposition}\label{prop:repairishard2}
There exists a "CQ" $q$ such that $\rShapley_{q}$ possesses no
"deterministic multiplicative approximation" of "factor@@approx" $1.3606$ unless $\Ptime=\NP$,
nor any
"FPRAS" unless $\NP \subseteq \BPP$.
\end{proposition}

\begin{proof}
   \AP We briefly recall the \phantomintro{vertex cover}$\intro*\VerCov$ problem: Given an undirected graph $G=(V,E)$ and an integer $k$, we ask if there exists a subset $V'\subseteq V$ of size $|V'|\le k$ that touches all edges in $E$. This is a very well-known $\NP$-complete problem that is the subject of various inapproximability results which we shall discuss right after proving a size-preserving reduction from some $\rShapley_{q}$.

Consider the "CQ" $q\defeq \exists x,y\ ~ R(x)\land S(x,y)\land R(y)$ and the "partitioned database" $\D^G$ defined by $\Dn^G\defeq \{R(v): v\in V\}$, $\Dx^G\defeq \{S(u,v): \{u,v\}\in E\}$. By construction, for any $X\subseteq \Dn^G$, $X\cup\Dx\models q$ iff $X$ is an ""independent set"" of $G$ (\ie a set of vertices that induces no edge). Therefore any $\Gamma\subseteq\Dn$ such that $(\Dn\setminus \Gamma)\models q$ is a "vertex cover" of $G$, and $\rscorefun_{q,\D}(\Dn)$ is the size of the smallest "vertex cover" of $G$. At last, observe that $\sum_{\alpha\in\D} \Sh_{\rscorefun_{q,\D}}(\alpha) = \rscorefun_{q,\D}(\Dn)$ as a consequence of the \ref{Sh:3} axiom of the Shapley value. Overall this means that any algorithm that gives a "deterministic multiplicative approximation" of $\rShapley_{q}$ will solve $\VerCov$ with the same "approximation factor".

Now $\VerCov$ (or, pedantically, the function $\intro*\VerCovFun$ which maps every graph $G$ to the size of its minimal "vertex cover") is notoriously difficult to "approximate"; more precisely it is $\NP$-hard to "approximate" within a "factor@@approx" of smaller than $10\sqrt{5}-21=1.3606\dots$ \cite[Theorem 1.1]{dinurHardnessApproximatingMinimum2005}.\medskip

Additionally, if there were to exist an "FPRAS" for $\rShapley_{q}$, the above reduction would also yield an "FPRAS" $A$ for $\VerCovFun$, which we could use with the parameters $\epsilon = \nicefrac{1}{3n}$ and $\delta=\nicefrac{1}{3}$, where $n$ is the number of vertices in the input graph. This would run in polynomial time and guarantee:
\[\mathbf{P}\left[\frac{\VerCovFun(G)}{1+\nicefrac{1}{3n}} \le A(G,\epsilon,\delta) \le (1+\nicefrac{1}{3n})\VerCovFun(G)\right] \ge 1-\nicefrac{1}{3}\]

This can be rewritten as follows by minimal manipulations:
\[\mathbf{P}\left[(1-\nicefrac{1}{3n})\VerCovFun(G) \le A(G,\epsilon,\delta) \le (1+\nicefrac{1}{3n})\VerCovFun(G)\right] \ge \nicefrac{2}{3}\]

Furthermore $\VerCovFun(G) \le n$ because it is the size of a subset of vertices of $G$. This means that the following holds:
\[\mathbf{P}\left[\VerCovFun(G)-\nicefrac{1}{3} \le A(G,\epsilon,\delta) \le \VerCovFun(G)+\nicefrac{1}{3}\right] \ge \nicefrac{2}{3}\]

In other words, the polynomial algorithm $\tilde{A}$ that outputs the closest integer to $A$ will give the correct value of $\VerCov(G)$ with probability at least $\nicefrac{2}{3}$. By definition this would mean $\VerCov\in \BPP$ hence $\NP \subseteq \BPP$ by $\NP$-completeness of $\VerCov$.
\end{proof}

\subsection{Counting Homomorphisms of Minimal Supports}\label{ssec:issues_sash}
\AP
A natural idea for an alternative measure for (Boolean) "UCQ"s would be to replace `counting "minimal supports"' with `counting "homomorphisms"' in the definition of $\msscorefun$ to obtain a "wealth function family" \AP$\intro*\SAscorefun$ based on counting "homomorphisms" (\ie $\sascorefun_{q,\D}(\Dn) \defeq |\set{h : q \homto \Dn}|$).
Adopting $\SAscorefun$ in place of $\MSscorefun$ would allow us to obtain further tractability results, as we have seen that counting "homomorphisms" can be computationally simpler than counting "minimal supports" (\Cref{thm:countMS-ACQ-shP}). Unfortunately, although  $\SAscorefun$ seems reasonable at first sight, we show that it cannot decide "relevance", which may be considered a minimum requirement for any responsibility measure:

\begin{proposition}
The Shapley value of $\SAscorefun$ does not satisfy \ref{Shdb:2}.
Moreover, an "irrelevant" fact may have a larger value than "relevant" ones by an arbitrarily large (multiplicative or additive) margin.
\end{proposition}
\begin{proof}
Consider the "CQ" $q\defeq \exists x_0, x_1, x_2, x_3 ~ R(x_0,x_1)\land R(x_1,x_2)\land R(x_2,x_3)$ and the purely "endogenous" database $\D$ that contains, for every $i\in[k]$ ($k$ is a fixed parameter), $\alpha_i \defeq R(a_i,a_i)$, $\beta_i \defeq R(a_i,b)$ and $\gamma\defeq R(b,c)$ (see \Cref{fig:sash}).
The "homomorphisms" from $q$ to $\D$ are as follows:
\begin{itemize}
\item $\alpha_i$ appears in 3 "homomorphisms", which are characterized by their images: $\{\alpha_i\}$, $\{\alpha_i, \beta_i\}$ and $\{\alpha_i, \beta_i, \gamma\}$. Hence $\Sh_{\sascorefun_{q,(\D,\emptyset)}}(\alpha_i) = 1 + \frac{1}{2} + \frac{1}{3} = \frac{11}{6}$.
\item $\beta_i$ appears in 2 "homomorphisms": $\{\alpha_i, \beta_i\}$ and $\{\alpha_i, \beta_i, \gamma\}$. Hence $\Sh_{\sascorefun_{q,(\D,\emptyset)}}(\beta_i) = \frac{1}{2} + \frac{1}{3} = \frac{5}{6}$.
\item $\gamma$ appears in k "homomorphisms": $\{\alpha_i, \beta_i, \gamma\}$ for every $i\in [k]$. Hence $\Sh_{\sascorefun_{q,(\D,\emptyset)}}(\gamma) = \frac{k}{3}$.
\end{itemize}
The fact with highest score is therefore $\gamma$, and by an arbitrary margin because its score grows linearly with $k$ while the others are constant. However, the only minimal supports in this instance are the $\{\alpha_i\}$, which implies in particular that $\gamma$ is an "irrelevant". 
\begin{figure}[tb]
\centering
      \begin{tikzpicture}
	\coordinate (00) at (-2, 0.75);
	\coordinate (01) at (-2, -0.75);
	\coordinate (02) at (0, 0);
	\coordinate (03) at (2, 0);
	\coordinate (04) at (-2, 0);
	\coordinate (05) at (-9, 0);
	\coordinate (06) at (-8, 0);
	\coordinate (07) at (-7, 0);
	\coordinate (08) at (-6, 0);
	\coordinate (09) at (-5, 0);
	\begin{pgfonlayer}{nodelayer}
		\node [draw, circle] (0) at (00) {$a_1$};
		\node [draw, circle] (1) at (01) {$a_k$};
		\node [draw, circle] (2) at (02) {$b$};
		\node [draw, circle] (3) at (03) {$c$};
		\node [] (4) at (04) {$\cdots$};
	\end{pgfonlayer}
	\begin{pgfonlayer}{edgelayer}
		\draw [->,>=stealth] (0) to node[midway, above=-.1] {$\beta_1$} (2);
		\draw [->,>=stealth] (1) to node[midway, above=-.1] {$\beta_k$} (2);
		\draw [->,>=stealth] (2) to node[midway, above=-.1] {$\gamma$} (3);
		\draw [->,>=stealth, in=-135, out=135, loop] (0) to node[midway, left] {$\alpha_1$} ();
		\draw [->,>=stealth, in=-135, out=135, loop] (1) to node[midway, left] {$\alpha_i$} ();
	\end{pgfonlayer}
\end{tikzpicture}%
   
\caption{Instance where the Shapley value of $\SAscorefun$ gives the largest value to the "irrelevant" $\gamma$.}
\label{fig:sash}
\end{figure}
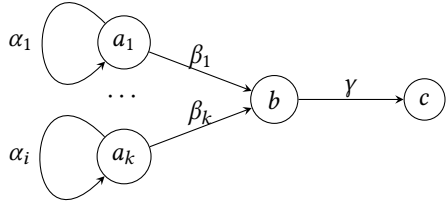
\end{proof}

Note that if we were to set the $\beta_i$ or one of the $\alpha_i$ as "exogenous", $\Sh_{\sascorefun_{q,(\D,\emptyset)}}(\gamma)$ would still be linear in $k$ while in the first case the "relevant" $\alpha_i$ would have the smallest values accross all facts and in the second there would actually be no "relevant" fact at all.

\subsection{Explainable AI and the SHAP Score}\label{ssec:shap}
The "Shapley value" has also been applied to define a popular quantitative measure of responsibility for machine learning, namely the "SHAP score" \cite{lundbergUnifiedApproachInterpreting2017}, which we shall briefly define later. Despite its popularity in practice, recent work \cite{marques-silvaExplainabilityNotGame2024} has cast doubts on the theoretical guarantees this measure provides. Although the context is substantially different from query answering, we discuss it here because we believe the present work can provide valuable insights on the issue.

\AP
Consider a ""classifier"", which we shall simply model as a function $\+M : \lB^{\+F} \to \lB$, where $\+F$ is a finite set of ""features"" and $\lB$ is the Boolean set $\{\top,\bot\}$.\footnote{For the sake of simplicity, we only consider Boolean classifiers here.} Given an input instance $\alpha\in \lB^{\+F}$, we wish to explain the classification $\+M(\alpha)$ by identifying the most relevant "features".

In order to obtain a quantitative score from the "Shapley value", one simply needs to define a relevant "game" whose players are the "features" in $\+F$. In the context of query answering, we have seen that one straightforward approach (giving rise to the "drastic Shapley value") consists in using the query itself, seen as a 0/1 function over the subsets of the input. Alas this is not possible for "classifiers" because $\+M$ needs all the features of its input to have some value.
The solution chosen to define the "SHAP score" consists in taking the average value of $\+M$ across all completions of $\alpha|_{S}$ (that is, $\alpha$ restricted to the features in $S$):
\AP$\intro*\shapscorefun(S) \defeq \frac{1}{2^{|\+F\setminus S|}} \sum_{x\in \lB^{\+F},\; x|_{S}=\alpha|_{S}} \+M(x)$.
\AP 
The ""SHAP score"" is then defined as the "Shapley value" applied to the game $(\+F,\shapscorefun)$.

\AP In order to analyse the "SHAP score", \cite{marques-silvaExplainabilityNotGame2024} defines the notion of an ""abductive explanation"" ("AXp" for short) as an  inclusion-minimal subset $S\subseteq \+F$ such that for all $x\in\lB^{\+F}$, $x|_{S}=\alpha|_{S} \Rightarrow \+M(x)=\+M(\alpha)$. "AXp"s are precisely analogous to the "minimal supports" of the database setting, and \cite{marques-silvaExplainabilityNotGame2024} essentially shows that the "SHAP score" does not satisfy the analogue of the \ref{Shdb:2} axiom, which we argued was a very natural condition in our setting.

However, due to the closeness of "AXp"s and "minimal supports", one can easily adapt \Cref{sec:alt-measures}
of the present paper to define \AP`Weighted Sums of "AXp"s'
(""WSAXp""s). These quantitative measures, which also derive from the "Shapley value", will naturally display all of the good properties of  "WSMS"s and therefore make viable candidates for addressing the theoretical drawbacks of the "SHAP score". 
Unfortunately, however, whereas "minimal supports" can be efficiently counted for relevant classes of database queries,
the problem of counting the number of  the "AXp"s of each size will be intractable for pretty much %
any interesting class of classifiers.
Moreover, this high complexity is inescapable: as deciding feature relevance is $\NP$-hard for very simple classes of classifiers \cite[Proposition 11]{huangFeatureNecessityRelevancy2023}, no measure satisfying the ML analogue of \ref{Shdb:2} can be tractable.

\section{Discussion}
\label{sec:conclusions}
We have revisited the question of how the Shapley value can be used to define
"responsibility measures" %
for "non-numeric queries". Our proposed family  
of "responsibility measures"  --  "weighted sums of minimal supports" ("WSMS"s) -- can be viewed as %
the Shapley values of suitably defined games and have been shown to satisfy desirable semantic properties. 
Moreover, due to their simple definition,
"WSMS"s  are amenable to interpretation %
and can be efficiently computed 
for a large class of queries (including all UCQs, \color{diegochangescolor}via execution of simple SQL queries\color{black}), 
making them 
an appealing alternative to the "drastic Shapley value".
We should point out that the two main measures that have been studied in the literature, namely the "drastic Shapley value" and the "drastic Banzhaf value" \cite{abramovichBanzhafValuesFacts2024} (\aka the Causal Effect \cite{salimiQuantifyingCausalEffects2016}) both satisfy our new desiderata for responsibility measures. More precisely, they satisfy the axioms \ref{Shdb:1}, \ref{Shdb:2} and \ref{MS:t} by \Cref{prop:weak-ax-sat}, as well as the stronger \ref{fMS:1} and \ref{fMS:2} by \Cref{prop:true-ax-sat-sh,prop:true-ax-sat-ms}.
Therefore, both remain perfectly adequate and sensible measures despite the discussed shortcomings.
Moreover, we argue that %
we cannot expect to identify a unique `best' responsibility measure and thus it is entirely natural to exhibit multiple different measures that satisfy our proposed desiderata. 
It may nevertheless be worthwhile to explore additional properties that distinguish these different responsibility measures
and may help to guide the decision as to which particular measure is most suitable for a given application.

Our work has focused on "monotone" queries, but the "drastic Shapley value" and "WSMS" values have also been applied to "(U)CQ@UCQ"s with negative atoms (\AP""UCQneg"") in \cite{ReshefKL20,ourICDT26}.
The approach of \cite{ReshefKL20} extends the "drastic Shapley value" by measuring the so-called (negative or positive) `impact' that a fact has with respect to an ordering of $\Totord(\bse{\scorefun})$. On the other hand, \cite{ourICDT26} studies more broadly the question of what
constitutes a reasonable notion of qualitative explanation or relevance for queries with negated atoms and proposes two additional approaches: the first one assigns scores to (positive) database facts, and the second one considers `signed' facts (positive facts and negated facts). These approaches are orthogonal to the score of \cite{ReshefKL20} and can be used to lift "WSMS"s as well as "drastic Shapley" values for monotone queries to (U)CQs with negated atoms. 
It has been shown in \cite{ourICDT26} that for both approaches (only positive facts, or signed facts), 
the "WSMS" measures are tractable %
in data complexity for all "UCQneg"s. Moreover, building upon \Cref{thm:bounded-selfjoin}, %
 tractability in combined complexity has been established for classes of conjunctive queries with negation
 that enjoy bounded "generalized hypertree width", bounded "self-join width" (or no "mergeable" "atoms" if only positive facts are considered), 
 and bounded negative arity. 
An interesting question for future work is to explore how the desirable properties for responsibility measures put forth in the present paper, in particular \ref{fMS:1} and \ref{fMS:2}, can be suitably adapted to non-monotone queries. 

Responsibility measures have also recently begun to be explored in the context of ontology-mediated query answering (OMQA) \cite{poggiLinkingDataOntologies2008,bienvenuOntologyMediatedQueryAnswering2015,xiaoOntologyBasedDataAccess2018}. 
We recall that in OMQA, query answers are defined using logical entailment, taking into account domain knowledge provided by a logical theory called an ontology (given, for example, as a set of description logic axioms or tuple-generating dependencies). In a first study, we investigated the adaptation of the "drastic Shapley value" to the ontology setting  \cite{ourkr24}, showing how (in)tractability and dichotomy results from the database setting could be transferred to ontology-mediated queries (\AP""OMQs"") formulated in prominent Horn description logics. Very recently, we explored "WSMS"s in the OMQA setting \cite{ourKR25} and identified conditions which ensure that "OMQs" have (in)tractable WSMS computation. In particular, 
we were %
able to establish tractability in combined complexity for "OMQs" $(\mathcal{O},q)$ where $\mathcal{O}$ is a DL-Lite ontology (a popular description logic), $q$ is a "CQ", and the "OMQ" satisfies a certain \emph{non-interaction} property which adapts the notion of queries with no "mergeable" atoms to "OMQs". This result is interesting as such "OMQs" can be equivalently reformulated as "UCQs" and thus provides a natural class of well-behaved "UCQs" for which "WSMS" computation is tractable in combined complexity (recall from \Cref{prop:UsjfACQ-MS-hard} that taking the union of two well-behaved "CQs" easily leads to intractability). It remains to explore whether a suitable analog of "self-join width" for "OMQs" could be used to identify larger classes of "OMQs" admitting tractable "WSMS" computation in combined complexity.   

Finally, we recall that we have
based our definition of "self-join width" on "mergeable" atoms,  but %
a reasonable alternative would be to use "unifiable" atoms instead (see \Cref{remark-sjwidth}). It is an interesting open question whether our tractability result (\Cref{thm:bounded-selfjoin}) holds for this alternative definition of "self-join width". Moreover, we believe that both notions of "self-join width" can be of  independent interest and worth studying in other contexts than responsibility measures. Indeed, it is natural to wonder which of the numerous positive results in the database literature that rely upon the "self-join free" condition can be extended to queries with bounded "self-join width".

\begin{acks}
This work is partially supported by ANR AI Chair INTENDED (ANR-19-CHIA-0014)
and ANR project CQFD (ANR-18-CE23-0003).
\end{acks}

\bibliographystyle{alphaurl}
\bibliography{long,biblio}

\end{document}